\documentclass[11pt]{article}
\usepackage{enumerate,mathrsfs,comment}
\usepackage{bbm}
\usepackage{enumitem}
\usepackage{natbib}
\usepackage[OT1]{fontenc}
\usepackage[hypertexnames=false, colorlinks,linkcolor=red,anchorcolor=blue,citecolor=blue]{hyperref}
\usepackage{fullpage}

\usepackage[protrusion=false,expansion=true]{microtype}
\usepackage{mathtools}

\RequirePackage{amsmath}
\RequirePackage{amssymb}
\RequirePackage{amsthm}
\RequirePackage{bm}
\usepackage{natbib}
\usepackage{multirow}
\usepackage{graphicx}
\usepackage{subfigure}
\usepackage{makecell}
\usepackage{booktabs}
\usepackage{array}
\usepackage{algorithm}
\usepackage{algorithmic}
\usepackage{dsfont}

\newcommand{\Cov}{\mathrm{Cov}}

\newcommand{\diag}{\mathrm{diag}}

\newcommand{\EE}{\mathbb{E}}

\newcommand{\bbE}{\mathbb{E}}

\newcommand{\bbR}{\mathbb{R}}

\newcommand{\vb}{\boldsymbol{b}}

\newcommand{\vI}{\boldsymbol{I}}

\newcommand{\vr}{\boldsymbol{r}}

\newcommand{\vx}{\boldsymbol{x}}
\newcommand{\vX}{\boldsymbol{X}}

\newcommand{\vy}{\boldsymbol{y}}

\newcommand{\vz}{\boldsymbol{z}}

\newcommand{\bbeta}{\boldsymbol{\beta}}

\newcommand{\bgamma}{\boldsymbol{\gamma}}

\newcommand{\bSigma}{\boldsymbol{\Sigma}}

\usepackage{scalerel,stackengine}
\stackMath
\newcommand\reallywidehat[1]{%
\savestack{\tmpbox}{\stretchto{%
  \scaleto{%
    \scalerel*[\widthof{\ensuremath{#1}}]{\kern.1pt\mathchar"0362\kern.1pt}%
    {\rule{0ex}{\textheight}}%
  }{\textheight}%
}{2.4ex}}%
\stackon[-6.9pt]{#1}{\tmpbox}%
}

\newcommand{\bs}{\mathbf{s}}

\newcommand{\cov}{\text{Cov}}
\newcommand{\var}{\text{Var}}

\newcommand{\bx}{\mathbf{x}}  
\newcommand{\by}{\mathbf{y}}
\newcommand{\be}{\mathbf{e}}
\newcommand{\ba}{\mathbf{a}}
\newcommand{\br}{\mathbf{r}}
\newcommand{\cS}{\mathcal{S}}

\newcommand{\cV}{\mathcal{V}}

\newcommand{\hbjplus}{\hat{\beta}_j^+}
\newcommand{\hbjminus}{\hat{\beta}_j^-}

\newcommand{\bjplus}{\beta_j^+}
\newcommand{\bjminus}{\beta_j^-}

\newcommand{\beps}{\bm{\epsilon}}
\newcommand{\veta}{\bm{\psi}}

\newcommand{\fdphat}{\widehat{\text{FDP}}}
\newcommand{\rhoxjxnj}{\rho^{(\vx_j,\vX_{-j})}}
\newcommand{\rhoxjzj}{\rho^{(\vx_j, \vz_j)}}
\newcommand{\rhozjxnj}{\rho^{(\vz_j,\vX_{-j})}}
\newcommand{\sigmaxnj}{\Sigma_{-j}}

\newcommand{\vzj}{v^{\vz_j}}
\newcommand{\FDPbar}{\overline{\text{FDP}}}

\DeclareMathOperator*{\argmin}{arg\,min}

\newtheorem{theorem}{Theorem}
\newtheorem{definition}{Definition}
\newtheorem{lemma}{Lemma}
\newtheorem{assumption}{Assumption}

\ifx\BlackBox\undefined
\newcommand{\BlackBox}{\rule{1.5ex}{1.5ex}}  %
\fi

\ifx\QED\undefined
\def\QED{~\rule[-1pt]{5pt}{5pt}\par\medskip}
\fi

\ifx\proof\undefined
\newenvironment{proof}{\par\noindent{\bf Proof\ }}{\hfill\BlackBox\\[2mm]}
\fi

\title{Controlling False Discovery Rate Using Gaussian Mirrors}

\author{Xin Xing\thanks{ Department of Statistics, Virginia Tech},
 Zhigen Zhao\thanks{ Department of Statistical Science, Temple University}, 
 Jun S. Liu\thanks{ Department of Statistics, Harvard University}}
\date{}

\begin{document}

\maketitle
\begin{abstract}
Simultaneously finding  multiple influential variables and controlling the false discovery rate (FDR) for linear regression models is a fundamental problem. %
We here propose the Gaussian Mirror (GM) method, which creates for each predictor variable a pair of mirror variables by adding and subtracting a randomly generated Gaussian perturbation, and proceeds with a certain regression method, such as the ordinary least-square or the Lasso (the mirror variables can also be created after selection). %
The mirror variables naturally lead to test statistics effective for 
controlling the FDR. Under a mild assumption on the dependence among the covariates, we show that the FDR can be controlled at any designated level asymptotically.  %
We also demonstrate through extensive numerical studies that the GM method is more powerful than many existing methods for selecting relevant variables subject to  FDR control,  especially  for cases when the covariates are highly correlated and the influential variables are not overly sparse. 

\end{abstract}

\section{Introduction}
Linear regression, which dates back to the beginning of the 19th century, is one of the most important statistical tools for practitioners. The theoretical research addressing various issues arising from big data analyses has gained much attention in the last decade. One  important problem is to determine which covariates (aka ``predictors") are ``relevant" or ``important" in a linear regression. %
In early days (before 1970's), people often rely on the  t-test to assess the importance of each individual predictor in a regression model, although the method is known to be problematic due to the existence of high multi-colinearity. A greedy stepwise regression method was later proposed in \cite{efroymson1960multiple} to alleviate some of the flaws. Good criterion, such as Akaike information criteria (AIC, \cite{akaike1998information}) and Bayesian information criteria (BIC, \cite{schwarz1978estimating}), for directing the model search were developed later. In recent years, due to the advance of data-generating technologies in both science and industry, %
researchers discovered that  various regularization based regression methods, such as Lasso \citep{tibshirani1996regression}, SCAD \citep{fan2001variable}, elastic net \citep{zou2005regularization} and many others, are quite effective in dealing with high-dimensional data and selecting relevant covariates with certain guaranteed theoretical properties, such as the selection consistency and the oracle property. 

For the linear regression model, $y=x_1\beta_1 +\cdots+x_p \beta_p+ \epsilon$, we are interested in testing $p$ hypotheses;  $H_j$: $\beta_j =0$ for $j=1,\ldots, p$, simultaneously and  finding a statistical rule to decide which hypotheses should be rejected. Let $\mathcal{S}_0\subset \{1,2,\cdots,p\}$  index the set of ``null'' predictors, i.e., those with $\beta_j=0$;  so  $ \mathcal{S}_1=\mathcal{S}_0^c $ indexes the set of relevant predictors. Let $\hat{\mathcal{S}}_1$ be the set selected based on a statistical rule. The FDR for quantifying the type I error for this statistical rule is defined as
\[
FDR = \mathbb{E} [FDP], \quad\textrm{where}\quad FDP = \frac{\#\{i\;|\; i\in \mathcal{S}_0, i\in\hat{\mathcal{S}}_1\}}{\#\{i \;|\;  i\in\hat{\mathcal{S}}_1\}\vee 1}.
\]

Controlling the FDR is useful when the number of features is large. For example, modern gene expression studies usually involve thousands of genes;  genome-wide association studies (GWAS) routinely examine tens of thousands to millions of single nucleotide polymorphisms (SNPs); and financial data often contain many features of thousands of companies. It is of practical importance to select a small set of the features for  follow-up experimental validations and testings or future predictions, which requires the researchers to have reliable FDR control.

The task of controlling the FDR at a designated level, say $q$, is challenging in two aspects: (i) the test statistic and the corresponding distribution under the null hypothesis is not easily available for high-dimensional problems; and (ii) the dependence structure of the covariates induces a complex dependence structure among the estimated regression coefficients. A heuristic method based on permutation tests was introduced by \citet{chung2016multivariate}, but it fails to yield  valid FDR control  unless $\mathcal{S}_1$ is empty. 
Starting with p-values based on the marginal regression, \cite{fan2012estimating} proposed a novel way to estimate the false discovery proportion (FDP) when the test statistic follows a multivariate normal distribution. However, the consideration of the marginal regression deviates from the original purpose of the study in many cases.

 \citet{barber2015controlling} introduced an innovative approach,  the {\it knockoff filter}, which provides provable FDR control for cases with an arbitrary design matrix when  $p < n/2$ where $n$ is the sample size. When $p<n<2p$, they kept $2p-n$ predictors unchanged and constructed the knockoffs only for the remaining $n-p$ variables. 
 By introducing knockoffs, they obtained a test statistic that  (i) are symmetric about the mean, and (ii) have independent signs, for the null variables.
To  this end, each column $\tilde{X}_j$ of the knockoff matrix $\tilde{\vX}$ needs to be exchangeable with the corresponding column $X_j$ in $\vX$. To gain more power, however, they want $\tilde{\vX}$ and $\vX$ as  dissimilar as possible. When  correlations or conditional correlations among the predictors are high and dense, there is little room for one to choose a good knockoff. This leads to a substantial decrease in  power (see Section \ref{sec:simulation} for more details).

The knockoff method was later extended to deal with high-dimensional problems via either the screening+knockoffs strategy \citep{barber2020robust} or  the model-X knockoff approach \citep{candes2018panning}. In the model-X framework, \citet{candes2018panning} relaxed the linear assumption by considering the joint distribution of the response $Y$ and the predictors $\vX$. The goal is to find the Markov blanket, a minimal set $\mathcal{S}$ such that $Y$ is independent of all the others when conditioning on $\vX_{\mathcal{S}}$. Although the model-X knockoff method can control the FDR, this construct relies heavily on the assumption that the distribution of $\vX$ is completely known, which is generally unrealistic except for some special scenarios. \citet{barber2020robust} constructed a robust version of model-X knockoffs based on an estimated joint distribution of $(X_1, \dots, X_p)$. However, estimating the joint distribution in a high-dimensional setting is not only  challenging even for the  multivariate Gaussian case, but also leads to inaccurate FDR control.
\citet{huang2019relaxing} further relaxed the requirement of Model-X knockoffs by allowing the joint distribution of $(X_1,\dots,X_p)$ to belong to a pre-specified exponential family with unknown parameters.
For high-dimensional data, another line of work is the post-selection inference, which aims at making inference conditional on the selection \citep{berk2013valid, lee2016exact}. In \cite{lee2016exact}, the selection event of LASSO is shown to be a union of polyhedra.  \cite{tibshirani2016exact} provided analogous results for the forward stepwise regression and the least angle regression.
\cite{taylor2018post} extended these results to $L_1$ penalized likelihood models including generalized regression model.

In this paper, we propose the Gaussian Mirror (GM) method, which constructs the test statistic  marginally. Namely, for each variable $\vx_j$, we create a pair of ``mirror variables'',  $\vx_j^+=\vx_j+c_j\vz_j$ and $\vx_j^-=\vx_j-c_j \vz_j$, where $c_j$ is a scalar and $\vz_j\sim N(0,\vI_n)$. %
This construct can be viewed as a quantifiable way of perturbing the feature, and is equivalent to treating $c_j \vx_j$ as a ``knockoff'' variable of $\vx_j$ if OLS is used for the model fitting.
The perturbation magnitude is chosen according to an explicit formula of $c_j$ such that the mirror statistic we introduce in Section \ref{sec:method} has a symmetric distribution about zero when the null hypothesis is true. For high-dimensional cases, we ensure this symmetric property by utilizing the post-selection theory.  With this symmetric property, we could estimate the number of false positives and the FDP. For a designated FDR level, say $q$, we choose a data-driven threshold such that the estimated FDP is no larger than $q$. Under some mild conditions on the design matrix, we show that the proposed  method controls the FDR at the designated level asymptotically. Adding some perturbations not only helps weaken the dependence among the test statistics, but also leads to conclusions that are stable to perturbations, as advocated in the stability framework of \cite{yu2020veridical, yu2013stability}.

In the GM approach, we obtain the mirror statistics by perturbing one variable at a time. Such a simple approach has two advantages. First, comparing with the knockoffs and Model-X knockoffs which double the size of the design matrix, the GM method introduces a smaller perturbation to the original design matrix, which can improve the power as shown in our numerical studies. Second, the algorithm for calculating the mirror statistics is separable, easily parallelizable, and  memory-efficient. All our numerical studies have been implemented using a parallel computation architecture.
Without requiring any distributional assumption on the design matrix $\vX$, %
the GM method bypasses the difficulty  of estimating the joint distribution of $(X_1,\dots, X_p)$, and is thus more applicable in real-data applications, such as GWAS studies where each feature  takes value in $\{0,1,2\}$, representing the number of minor-allele mutations of a single nucleotide polymorphism (SNP). 

In addition, practitioners often have a limited budget to verify the discoveries. For example, if the budget only allows a scientist to do $100$ validation tests, a natural question is how many false discoveries (FDs) they expect to have in a top-100 list and what is an uncertainty measure for the estimated FDs. To answer these questions, we provide an estimate and confidence interval for the expectation of FDs by using the bootstrap method. %

The paper is organized as follows. In Section \ref{sec:method}, we describe the general GM idea and a GM algorithm for the ordinary least squares (OLS) estimation
in the low-dimensional case ($p<n$). In Section \ref{subsec:lasso}, we discuss how to  employ the post-selection adjustment strategy to extend the GM construction for Lasso, which handles the case with $p\geq n$.
In Section \ref{sec:FDR:theorem}, we investigate theoretical properties of the GM procedure. 
In Section \ref{sec:var:fdp}, we introduce an estimator of the number of false discoveries and build its confidence interval using non-parametric bootstrap.
In Section \ref{sec:simulation}, we provide numerical evidences showing the advantage of GM to its competitors via simulations and real data analysis.  In Section \ref{sec:disc}, we conclude the article with a short discussion. All the technical proofs of all the lemmas and theorems stated in the article are put in the Appendix.
The R-code for the GM procedure is publicly available at \url{https://github.com/BioAlgs/GM}.

\smallskip{}

{\bf Notations:} To ease the readability, we summarize some notations used frequently in this paper. Let $\vX$ be the design matrix, $\vx_j = (x_{1j},\dots,x_{nj})^T$ be the $j$-th column of $\vX$, and  $\vX_{-j}$ be the submatrix of $\vX$ with the $j$-th column removed. Without loss of generality, we assume that $\vX$ is normalized so that $\sum_{i=1}^n x_{ij}=0$ and $\frac{1}{n}||\vx_j||_2 =1$, $j=1,\dots, p$.  Let $\vX^j = (\vx_j^+, \vx_j^-, \vX_{-j})$, where $\vx_j^+ = \vx_j + c_j\vz_j$, $\vx_j^- = \vx_j - c_j\vz_j$, $\vz_j\sim N(0, \vI_n)$, and $c_j$ is a scalar. Let $\bbeta$ be the vector of coefficients and let $\bbeta_{-j}$ be the subvector with the $j$-th entry removed. Let $\bbeta^j=(\beta_j^+, \beta_j^-, \bbeta_{-j}^\top)^\top$, where $\beta_j^+$ and $\beta_j^-$ denote the coefficients of the mirror variable pair, and let $\hat{\bbeta}^j$ be an estimate of $\bbeta^j$.  Let $\mathcal{S}_0=\{j: \beta_j=0\}$, $\mathcal{S}_1=\{j:\beta_j\neq 0\}$, $p_0=|\cS_0|$, and $p_1=|\cS_1|$.
The designated FDR level to be controlled at is $q$.

\section{The Gaussian Mirror Methodology for OLS}\label{sec:method}

\subsection{The Gaussian mirror}\label{subsec:ols}
Suppose the data are generated from the following linear regression model
\begin{equation}\label{eqn:lm}
\vy = \vX \bbeta + \beps, 
\end{equation}
where $\vy \in \bbR^{n}$ is the response vector, $\vX=(\vx_1,\cdots, \vx_p) \in \bbR^{n\times p }$ is the design matrix, and $\beps = (\epsilon_1,\dots, \epsilon_n)$ is a vector of i.i.d. draws from $N(0,\sigma^2)$.
We begin with the low-dimensional case with $p<n$ and  focus on the OLS estimator. The high-dimensional case with $p>n$ is deferred to the next section.
As shown  in Figure \ref{flowchart}, we construct the $j$-th {\it Gaussian Mirror} by replacing $\vx_j$ with a pair of variables $(\vx_j^+, \vx_j^-)$ with $\vx_j^+ = \vx_j + c_j \vz_j$ and $\vx_j^- = \vx_j - c_j \vz_j$, where $\vz_j$ is an independently simulated Gaussian random vector with mean $0$ and covariance $\vI_n$ and $c_j$ is a scalar which depends on $\vX$ and $\vz_j$. 
\begin{figure}[h]
\centering 
\includegraphics[width=0.7\textwidth]{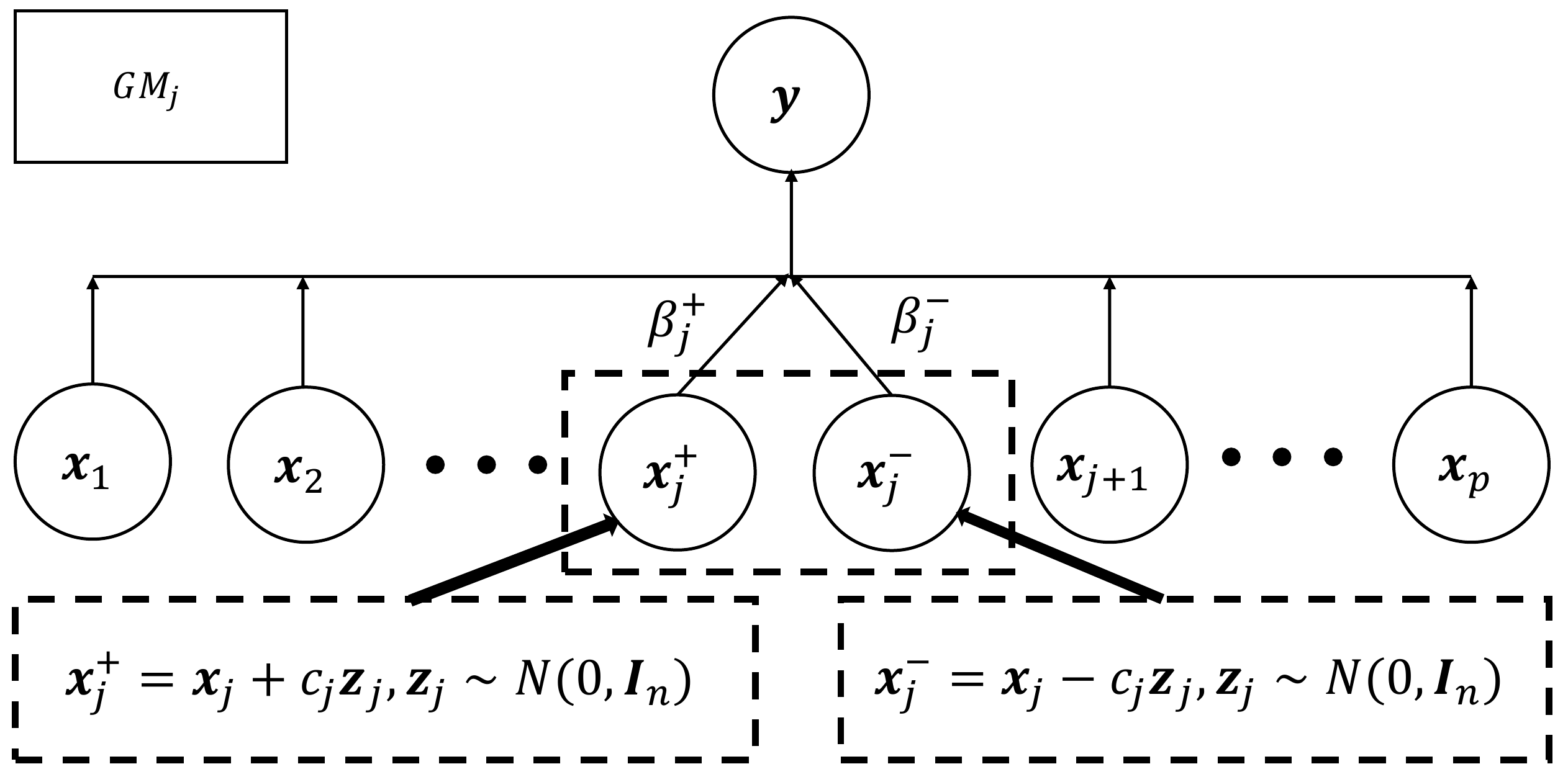}
\caption{Flowchart of the $j$-th {\it Gaussian Mirror}.}\label{flowchart}
\end{figure}

Note that $\beta_{j}^+ =\beta_{j}^- = 0 $ when $j\in S_0$ and $\beta_{j}^+ =\beta_{j}^- = \beta_j/2$ when $j\in S_1$. Let $GM_j$ be the design matrix with the $j$-th Gaussian mirror, i.e.,
\begin{equation}\label{eq:gm_design}
 GM_j:=\vX^j = (\vx_j^+, \vx_j^-, \vX_{-j})= (\vx_j + c_j\vz_j, \vx_j - c_j\vz_j, \vX_{-j}). 
\end{equation}
Fit the model $\vy = \vX^j \bbeta^j + \beps$ and estimate the  coefficients %
by the least squares:
\begin{equation}\label{eq:plsgm}
\hat{\bbeta}^j 
=\argmin_{\bbeta^j=(\bjplus, \bjminus, \bbeta_{-j})}||\vy - \vX^j \bbeta^j||_2^2.
\end{equation}
The mirror statistic for the $j$-th variable is constructed as
\begin{equation}\label{eq:mj}
    M_j = |\hbjplus+\hbjminus|-|\hbjplus-\hbjminus|.
\end{equation}
Another possible construct is the ``signed-max'' statistic as suggested in \cite{barber2015controlling} and studied more carefully in \cite{ke2020power}. %
The mirror statistic $M_j$ has two parts: the first part reflects the importance of the $j$-th predictor; and the second part captures the noise cancellation effect.
Intuitively, when $j\in\mathcal{S}_0$, the coefficients $\hbjplus$ and $\hbjminus$ are due to noise and the mirror statistic $M_j$ should center around zero; when $j\in \mathcal{S}_1$,  $\hbjplus+\hbjminus$  reflects the signal, whereas $\hbjplus-\hbjminus$ cancels out the signal and only reflects the noise effect. %
Based on this observation, we would declare a variable $\vx_j$ as ``significant" if $M_j\geq t$ for some threshold $t>0$. The choice of $t$ depends on the number of false positives and the associated FDP. If we can choose GM properly such that 
$M_j$ follows a distribution symmetric about zero, i.e., $P(M_j \geq s) = P(M_j \leq -s)$, $\forall s$. Then, the number of false positive $\#\{j\in \mathcal{S}_0\,|\, M_j \geq t \}$ can be approximated by $\#\{j\in \mathcal{S}_0 \,|\, M_j \leq -t \}$. 
The $FDP(t)$ can be estimated by 
\begin{equation}\label{eq:fdphat}
    \fdphat(t)  \triangleq  \frac{\#\{j \,|\, M_j \leq -t \}}{\#\{j \,|\, M_j\geq t\}   \vee 1}.
\end{equation}
For any designated FDR level $q$, we could choose $\tau_q$ such that %
\begin{equation}\label{eq:tau}
\tau_q = \min_{t}\{\fdphat(t) \leq q \},
\end{equation}
and reject the $j$-th hypothesis if $M_j\ge \tau_q$. %

The next subsection explains how to construct  GMs in the low-dimensional case when $p<n$ such that the mirror statistics for the null variables are symmetric about zero. %
Section~\ref{subsec:lasso} extends the GM construct to the Lasso estimates for the high-dimensional case.

\subsection{Gaussian mirrors for the OLS estimator}\label{subsec:ols}
The $j$-th GM design ($j=1,\cdots,p$) for the OLS estimates leads to the following quantity: 
\begin{equation*}\label{eq:ols}
 \hat{\bbeta}^j :=(\hbjplus, \hbjminus, \hat{\bbeta}_{-j})=\argmin_{ \bbeta^j=(\bjplus, \bjminus, \bbeta_{-j}) }||\vy  - \vX^j\bbeta^j ||_2^2,
\end{equation*}
which has an explicit expression as $\hat{\bbeta}^j = ((\vX^j)^\top \vX^j)^{-1}(\vX^j)^\top \vy$. It is known that %
$(\hbjplus, \hbjminus)$ follows a bivariate normal distribution.
As a referee pointed out,  when $p<n$ and OLS is used for the fitting, the GM construct has an equivalent ``knockoff form'': one can treat $c_j \vz_j$ as a knockoff of $\vx_j$; directly regress $\vy$ on $(\vX, c_j \vz_j)$ and contrast the coefficients of $\vx_j$ and $c_j\vz_j$. %
This view provides a nice connection with the knockoff idea, but the equivalence no longer holds in more complex cases when OLS is not applicable. 
The following lemma %
 shows why the mirror statistic $M_j$ 
 is symmetrically distributed about zero for $j \in \mathcal{S}_0$. 

\begin{lemma}\label{lemma:symmetric}
Suppose $(U, V)$ follows a bivariate normal distribution with mean zero. If the correlation between $U$ and $V$ is zero,  then $M = |U + V| - |U-V|$ follows a symmetric distribution about zero, i.e., $P(M\geq  t) = P(M\leq -t)$,   $\forall t>0$.
\end{lemma}

The following lemma provides an explicit formula of $c_j$ such that the correlation between $\hbjplus$ and $\hbjminus$ is zero, resulting in a symmetric distribution of $M_j$.
\begin{lemma}\label{lemma:olscor}
For the $GM_j$ design in $(\ref{eq:gm_design})$, choose
\begin{equation}\label{eq:olscj}
    c_j = \sqrt{\frac{\vx_j^\top(\vI_n - \vX_{-j}(\vX_{-j}^\top \vX_{-j})^{-1} \vX_{-j}^\top) \vx_j}{\vz_j^\top(\vI_n - \vX_{-j}(\vX_{-j}^\top \vX_{-j})^{-1} \vX_{-j}^\top) \vz_j}}.
\end{equation}
Then $\mbox{corr}(\hbjplus, \hbjminus \mid \vz_j, \vX_{-j})=0$ and $M_j$ is symmetric with respect to zero.
\end{lemma}

\begin{definition}\label{gm:ols}
(Gaussian Mirrors for OLS) For $j=1,\cdots,p$, one first generates an $n$-dimensional Gaussian random vector $\vz_j\sim N(0, \vI_n)$, and then computes $c_j$ based on equation (\ref{eq:olscj}). The $j$-th GM design is $GM_j = \{\vx_j + c_j \vz_j,\vx_j - c_j \vz_j, \vx_1, \dots, \vx_{j-1}, \vx_{j+1}, \dots, \vx_p\} $.
\end{definition}

According to Lemma \ref{lemma:olscor},  $GM_j$ defined in Definition \ref{gm:ols} guarantees that the covariance between $\hbjplus$ and $\hbjminus$ is zero. We further construct the mirror test statistics $M_j$ %
as in (\ref{eq:mj}). The following theorem is a direct consequence of Lemma \ref{lemma:symmetric}.
\begin{theorem}\label{thm:coinflip1}
Let $M_j$'s be the test statistics defined in (\ref{eq:mj}) based on $GM_j$ in Definition \ref{gm:ols}, then
\[
P(M_j \leq -t \mid \vz_j) = P(M_j \geq t \mid \vz_j), \; \forall t>0, 
\]
for $j\in \mathcal{S}_0$.
\end{theorem}

\begin{algorithm}[ht]
\begin{enumerate}

\item Parallel FOR $j=1, 2,\cdots, p$:

\begin{enumerate}

\item Generate $\vz_j$ from Gaussian distribution with mean $0$ and variance $\vI_n$.

\item Calculate $c_j$ using (\ref{eq:olscj}) and get the $j$-th GM design, $GM_j$, as in Definition \ref{gm:ols}. %

\item Obtain the ordinary least square estimator of $\hbjplus$ and $\hbjminus$:
    \begin{equation*}
      (\hbjplus, \hbjminus, \hat{\bbeta}_{-j}) =  \argmin_{ \beta_j^+, \beta_j^-, \mathbf{\bbeta}_{-j}}||\vy -  \vX_{-j}\mathbf{\bbeta}_{-j} - \vx_j^+ \beta_j^+ -\vx_j^- \beta_j^- ||_2^2.
    \end{equation*}
    
\item Calculate the mirror statistic $M_j= |\hbjplus+ \hbjminus| - |\hbjplus-\hbjminus|$.
\end{enumerate}
END parallel FOR loop

\item For a designated FDR level $q$, calculate the cutoff $\tau_q$ as
\begin{equation*}
    \tau_q = \min_t\left\{t:\frac{\#\{j \,|\, M_j \leq -t \}}{\#\{j \,|\, M_j\geq t\}	\vee 1} \leq q \right\}.
\end{equation*}

\item Output the index of the selected variables: $\hat{\mathcal{S}}_1 = \{j \,|\, M_j \ge \tau_q \}$.
\end{enumerate}
\caption{Gaussian mirror algorithm for OLS}\label{alg:gm:ols}
\end{algorithm}

Perturbation has been widely used in statistics to ensure stability  and quantify uncertainty \citep{yu2020veridical,yu2013stability}.   For example, jackknife and bootstrap \citep{efron1992bootstrap} are two effective data perturbation methods with broad applications. As summarized in Algorithm~\ref{alg:gm:ols}, we use data perturbations to control the FDR, which can
 improve both the robustness and power of the multiple testing procedure. Since GM deals with one variable at a time, the resulting algorithm can be easily parallelized.
 Theorem \ref{thm:coinflip1} guarantees that the distribution of $M_j$ is symmetric about zero for the null variables. Therefore, for any threshold $t>0$, and $\hat{\mathcal{S}}_1=\{j:M_j\ge t\}$,
a natural estimate of the FDP in $\hat{\mathcal{S}}_1$ is given by (\ref{eq:fdphat}) and a   data-driven threshold $\tau_q$ can be obtained as in (\ref{eq:tau}). 
In Section \ref{sec:FDR:theorem}, we show that the data-driven choice of $\tau_q$ in Algorithm \ref{alg:gm:ols} guarantees that FDR is controlled asymptotically under some weak dependence assumptions of the mirror statistics.

\section{Gaussian Mirrors for High Dimensional Regression}\label{subsec:lasso}
In the high-dimensional setting when the number of features $p$ is greater than the number of subjects $n$, one may still create mirror variables as $\vx_j^+ = \vx_j + c_j \vz_j$ and $\vx_j^- = \vx_j - c_j \vz_j$,  and construct
 the mirror statistics  using the Lasso estimates in (\ref{eq:lassocj}). 
 Although this simple extension works well empirically, the following challenges arise.
 First, the easy formula for $c_j$ as in the OLS case is no longer available. We can estimate the precision matrix of $\vX$ under sparsity assumptions and construct $c_j$  to make $\vx_j^+$ and $\vx_j^-$ conditionally independent asymptotically. Though this approach works well in simulations, a rigorous theoretical justification is difficult to come by. Second, the Lasso estimator is biased because of the regularization via $L_1$ penalty. This implies that $\mathbb{E}(\hbjplus)\neq 0$ and $\mathbb{E}(\hbjminus)\neq 0$ for some predictor variables in the null set.
 Last but not the least, 
 the Lasso estimator is a nonlinear function of $\by$ and  the Lasso variable selection process imposes nonlinear constraints on the sample space. %
We here propose a post-selection strategy to design mirror variables such that the resulting mirror statistic $M_j$  satisfies the symmetric property when $j\in \mathcal{S}_0$.%

\subsection{Literature on High-dimensional Inference}
In order to derive a proper inferential method 
for high-dimensional regression  post variable selection,
\cite{zhang2014confidence} and \cite{van2014asymptotically} proposed de-biased Lasso,  \cite{wasserman2009high} advocated using data splitting  \citep{cox1975note} to overcome difficulties caused by variable selection,
and  \cite{berk2013valid} suggested a post-selection adjustment framework conditioning on the selection.
Under the post-selection framework,
\cite{lee2016exact} proposed a procedure that first selects variables using Lasso and then obtains the OLS estimates for the selected model. By characterizing the selection event as a series of linear constraints on the post-selection OLS estimates, they provided a valid post-selection inference method on certain linear combinations of the coefficients of the selected variables. 

Our goal is to make simultaneous inference for all the parameters such that the FDR can be controlled at a designated level. %
We consider a post-selection procedure based on Lasso. More specifically, similar to \cite{lee2016exact}, we first use Lasso to select variables and then re-fit the selected model with OLS and construct mirror statistics the same way as in the low-dimensional case.  %
In the following, we explain how to construct the mirror statistics, adjusted according to the selection, such that the symmetric property holds asymptotically. %

\subsection{Post-selection adjustment for Gaussian mirrors}
Recall that Lasso solves the minimization problem  
\begin{equation}\label{eq:org:lasso}
    \tilde{\bbeta} 
=\argmin_{\bbeta} ||\by - \vX \bbeta ||_2^2 + \lambda_n || \bbeta ||_1. 
\end{equation}
Let $\hat{\cS} = \{j: \tilde{\beta}_j \ne 0\}$ and  $ \hat{\bs}= \mbox{sign}(\tilde{\beta}_{\hat{\cS}})$ denote the set of selected variables and  their coefficients' signs, respectively, obtained by  Lasso. %
By Lemma 4.1 in \cite{lee2016exact}, %
the event $\{ (\hat{\cS}, \hat{\bs})= (\cS, \bs)\}$ can be rewritten as a  series of constraints on $\by$ as 
\begin{equation}\label{eq:postineq}
  \{\hat{\cS} = \cS, \hat{\bs} = \bs \} = \begin{Bmatrix}
  \begin{pmatrix}
  A_0(\cS, \bs)\\
  A_1(\cS, \bs) 
  \end{pmatrix} \by \le
  \begin{pmatrix}
  \vb_0(\cS, \bs)\\
  \vb_1(\cS, \bs)
  \end{pmatrix}
  \end{Bmatrix}, 
\end{equation}
where  $A_0$ and $A_1$ are matrices shown below, which encode the ``inactive'' constraints determining the selection set and the ``active'' constraints determining the sign of nonzero coefficients, respectively:
\begin{equation}\label{eq:constraint}
\begin{aligned}
    A_0(\cS, \bs) &=\frac{1}{\lambda_n}\begin{pmatrix} \vX^\top_{-\cS}(\vI_n - P_{\cS})\\
    - \vX^\top_{-\cS}(\vI_n - P_{\cS})
    \end{pmatrix},  \\
    \vb_0(\cS, \bs) &= \begin{pmatrix}
    1 - \vX^\top_{-\cS}(\vX_{\cS})^{+}\bs \\
    1 + \vX^\top_{-\cS}(\vX_{\cS})^{+}\bs
    \end{pmatrix}, \\
    A_1(\cS, \bs) &= -\text{diag}(\bs)(\vX_{\cS}^\top  \vX_{\cS})^{-1} \vX_{\cS}^\top,\\
    \vb_1(\cS, \bs) &= -\lambda_n \diag(\bs)(\vX_{\cS}^\top  \vX_{\cS})^{-1} \bs,
\end{aligned}
\end{equation}
where $P_{\cS}$ is the projection matrix of $\cS$ corresponding to $\vX_{\mathcal{S}}$. 

Let $\vX_{-j(\cS)}$ denote the the submatrix of $\vX_{\cS}$ by removing its $j$-th column, and let $c_j$ be  
\begin{equation}\label{eq:lassocj}
    c_j = \sqrt{\frac{\vx_j^\top(\vI_n - \vX_{-j}(\vX_{-j(\cS)}^\top \vX_{-j(\cS)})^{-1} \vX_{-j(\cS)}^\top) \vx_j}{\tilde{\vz}_j^\top(\vI_n - \vX_{-j(\cS)}(\vX_{-j(\cS)}^\top \vX_{-j(\cS)})^{-1} \vX_{-j(\cS)}^\top) \tilde{\vz}_j}},
\end{equation}
where  $\tilde{\vz}_j = (I - P_{\cS})\vz_j$ with $\vz_j$ generated from $N(0, \vI_n)$. Then, we have

\begin{definition}\label{gm:lasso}
(Post-selection Gaussian Mirror) Given the post-selection set $\cS$ and the corresponding design matrix $\vX_\cS$, for $j\in \{ 1,\dots, |\cS|\}$, we generate $n$-dimensional random vector $\vz_j$ from $N(0, \vI_n)$, then compute $c_j$ based on equation (\ref{eq:lassocj}). %
The $j$-th GM is designed as 
$
\vX^j_{\cS} = \{\vx_{\cS[j]} + c_j \vz_j,\vx_{\cS[j]} - c_j \vz_j, \vx_{\cS[1]}, \dots, \vx_{\cS[j-1]}, \vx_{\cS[j+1]}, \dots, \vx_{\cS[|\cS|]}\}$ where $\cS[j]$ is the $j$th index in $\cS$. 

\end{definition}

In contrast to (\ref{eq:olscj}) in the OLS case, here $c_j$ is constructed by first projecting $\vz_j$ on the orthogonal space of $\vX_{\cS}$. 
Let $\veta^\top_{1} = \be_1^\top ((\vX^j_{\cS})^\top\vX^j_{\cS})^{-1} (\vX^j_{\cS})^\top $ and $\veta^\top_{2} = \be_2^\top ((\vX^j_{\cS})^\top\vX^j_{\cS})^{-1} (\vX^j_{\cS})^\top$ where $\be_\ell$ is the standard basis vector with the $\ell$-th entry as one and the others as zero.
Note that  the first two estimated coefficients from the OLS fitting of $y$ 
on $\vX_\cS^j$ are simple  
\begin{equation}\label{eq:post-beta}
    \hbjplus = \veta_{1}^\top \by \ \mbox{ and } \ \hbjminus = \veta_{2}^\top \by.
\end{equation}
Before deriving the joint post-selection distribution of $\hbjplus$ and $\hbjminus$, we first characterize the linear constraints on $(\hbjplus + \hbjminus,\hbjplus - \hbjminus)$ due to the post-selection event $\cS$. 

\begin{lemma}\label{lem:constraints}
Let $\veta=(\veta_1, \veta_2)$, $A_0(\cS, \bs)$ and $A_1(\cS,\bs)$ be matrices defined in (\ref{eq:constraint}). Then, there exist $\ba_0\in \bbR^{2(p-|\cS|)}$ and $\ba_1 \in \bbR^{|\cS|}$, such that 
\begin{equation}\label{eq:constrain}
    A_0(\cS, \bs)\veta  =\ba_0 (1,-1)  \mbox{ and } A_1(\cS, \bs)\veta =\ba_1 (1,1).
\end{equation}
\end{lemma}

In Eq (\ref{eq:constrain}), we show that $ A_0(\cS, \bs)\veta$ and $A_1(\cS, \bs)\veta$ are both rank one. Write $\by= P_{\veta} \by + (\vI_n- P_{\veta})\by$, and let 
$\br = (\vI_n - P_{\veta})\by$. It is easy to verify that $\br$ is uncorrelated with  $\veta^\top\by$, hence independent of $\veta^\top\by$. 
Let $\alpha = \veta_1^T\veta_1 = \veta_2^T \veta_2$. The constraint $A_0(\cS, \bs) \by \leq \vb_0(\cS, \bs)$ in (\ref{eq:postineq}) is equivalent to 
\begin{align*}
 A_0(\cS, \bs) \veta (\veta^\top\veta)^{-1}\veta^\top\by + A_0(\cS, \bs) \br &= \ba_0(1,-1)\diag(\alpha,\alpha) (\hbjplus, \hbjminus)^\top  + A_0(\cS, \bs) \br \\
 & = \alpha \ba_0(\hbjplus- \hbjminus)  + A_0(\cS, \bs) \br < \vb_0(\cS, \bs);
\end{align*} 
i.e., the ``inactive'' constraints on $(\hbjplus, \hbjminus)$ are applied on the direction of $[1,-1]$. Similarly, we have $\alpha \ba_1 (\hbjplus+ \hbjminus) + A_1(\cS, \bs) \br < \vb_1(\cS, \bs))$, that is, the ``active'' constraints are applied on the direction of $[1,1]$. 
 As shown in Figure \ref{fig:betalasso} (left panel), the constraint regions for $\hbjplus$ and $\hbjminus$ are along the line with slope $1$ and $-1$. By rotating the coordinate system by $45^{\circ}$, we have the joint distribution of $\hbjplus+\hbjminus$ and $\hbjplus-\hbjminus$ shown in Figure \ref{fig:betalasso} (right panel), where the constraint regions are parallel to the $x$-axis and $y$-axis. We characterize the constraints provided by the selection event $\cS$ in  Lemma \ref{le:exact:constraint}. 

\begin{figure}[h]
\centering 
\includegraphics[width=0.8\textwidth]{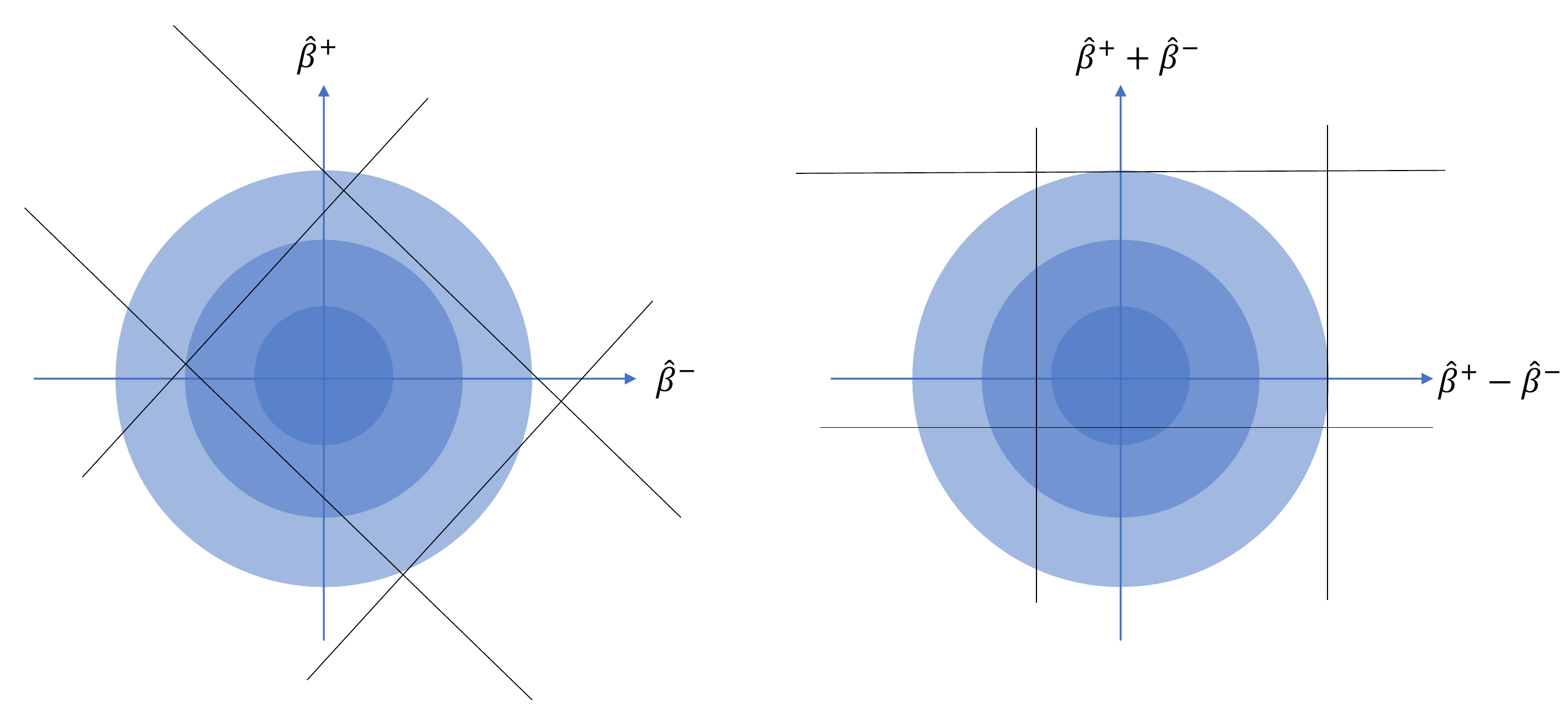}
\caption{Left: the joint distribution of $\hbjplus$ and $\hbjminus$. Right: Rotate the coordinate system by $45^{\circ}$ and obtain the joint distribution of  $\hbjplus+\hbjminus$ and $\hbjplus-\hbjminus$.} \label{fig:betalasso}
\end{figure}

\begin{lemma} \label{le:exact:constraint}
The selection event can be written as
\begin{align*}
   & \{A\by \leq \vb\} = \{A_0(\cS,\bs)\by \le \vb_0(\cS,\bs) \} \cap \{A_1(\cS,\bs) \by \le \vb_1(\cS,\bs)\} \\
    =&   \{\cV_0^{L}(\br) \le \hbjplus-\hbjminus \le \cV_0^{U}(\br),\, \cV_0^{N}(\br)>0 \}   \cap \{\cV_1^{L}(\br) \le \hbjplus+\hbjminus \le \cV_1^{U}(\br),\,\cV_1^{N}(\br)>0 \}, 
\end{align*}
where 
\begin{equation}\label{eq:v}
\begin{aligned}
    \cV_0^{L}(\br)  := \max_{j:a_{0j} <0}
    \frac{b_{0j} - (A_0\br)_j}{\alpha a_{0j}}, \, \cV_0^{U}(\br)  := \min_{j:a_{0j} >0}\, \frac{ b_{0j} - (A_0\br)_{j}}{\alpha a_{0j}}, \,
    \cV_0^{N}(\br)  := \min_{j: a_{0j} = 0}
     \vb_{0j} -  (A_0\br)_{j}, \\
    \cV_1^{L}(\br)  := \max_{j: a_{1j} <0}
    \frac{b_{1j} -  (A_1\br)_{j}}{\alpha a_{1j}}, \,
    \cV_1^{U}(\br)  := \min_{j: a_{1j} >0}
    \frac{ b_{1j} - (A_1\br)_{j}}{ \alpha a_{1j}}, \,
    \cV_1^{N}(\br)  := \min_{j: a_{1j} = 0}
     b_{1j} - (A_1\br)_{j},
\end{aligned}
\end{equation}
with $\alpha= \veta_1^\top \veta_1 = \veta_2^\top \veta_2$, $a_{0j}, a_{1j}$ being the $j$-th element of $\ba_0, \ba_1$ in (\ref{eq:constrain}), $b_{0j}, b_{1j}$ being the $j$th element of $ \vb_0(\cS,\bs)$ and $\vb_1(\cS,\bs)$ in (\ref{eq:constraint}), respectively. 
\end{lemma}

Note that  $\Cov(\hbjplus + \hbjminus, \hbjplus - \hbjminus) = \sigma^2(\veta_1 + \veta_2)^\top (\veta_1 - \veta_2)=0$. 
Since $(\hbjplus, \hbjminus)$ are independent of $\veta^\top \by$, we can treat $\br$  as “fixed” quantities for the distribution of  $\hbjplus+\hbjminus$ and $\hbjplus-\hbjminus$ when conditioning on $\cS$. Thus $\hbjplus+\hbjminus | \{A\by \leq \vb\}$ and $\hbjplus-\hbjminus |\{A\by \leq \vb\}$ are two independent, truncated normal random variables. %
Therefore, we have the following theorem.

\begin{theorem}\label{thm:trun:dist}
Let $F_{\mu, \sigma^2}^{[a, b]}$ denote the cumulative distribution function of a $N(\mu, \sigma^2)$ random variable truncated to the interval $[a,b]$, that is,
\begin{equation}\label{eq:truncatednormal}
    F_{\mu, \sigma}^{[a, b]}(x) = \frac{\Phi((x-\mu)/\sigma)-\Phi((a-\mu)/\sigma)}{\Phi((b-\mu)/\sigma)-\Phi((a-\mu)/\sigma)},
\end{equation}
where $\Phi$ is the cumulative distribution function of a standard normal random variable. 
When $\cS_1	\subseteq \cS$ and for any $j\in \cS_0$,
 we have
 \begin{equation*}\label{eq:truncation}
\begin{aligned}
F_{0,\sigma^2}^{[\cV_1^{L}(\br), \cV_1^{U}(\br)]} (\hbjplus+\hbjminus )\,\mid\, \hat{\cS} = \cS, \hat{\bs} = \bs \sim  Unif(0,1), \\
F_{0,\sigma^2}^{[\cV_0^{L}(\br), \cV_0^{U}(\br)]} (\hbjplus - \hbjminus )\,\mid\, \hat{\cS} = \cS, \hat{\bs} = \bs \sim  Unif(0,1), 
\end{aligned}
\end{equation*}
where $\cV_0^{L}(\br)$, $\cV_0^{U}(\br)$, $\cV_1^{L}(\br)$ and $\cV_1^{U}(\br)$ are defined in (\ref{eq:v}).
\end{theorem}

We define the mirror statistic $M_j$ as 
\begin{equation}\label{eq:hdmirror}
M_j = |\sigma\Phi^{-1}F_{0,\sigma^2}^{[\cV_1^{-}(\br), \cV_1^{+}(\br)]} ( \hbjplus + \hbjminus )| -  |\sigma\Phi^{-1}F_{0,\sigma^2}^{[\cV_0^{-}(\br), \cV_0^{+}(\br)]} ( \hbjplus - \hbjminus  )|. 
\end{equation}
For $j\in \cS_0$, $M_j$ is symmetrically distributed about zero. When $j\in \cS_1$, as shown in \cite{lee2016exact}, the truncation points are different from zero. When the magnitude of $\beta_j$ is large, the truncated distribution function is close to the non-truncated version, %
indicating that the $M_j$ defined in (\ref{eq:hdmirror}) is close to $|\hbjplus + \hbjminus|-|\hbjplus - \hbjminus|$. In practice, the mirror statistics tend to be conservative since we use the cumulative distribution function of a normal random variable with mean zero for both null and non-null variables. For a non-null variable, the means for $\hbjplus+\hbjminus$ and $\hbjplus-\hbjminus$ before the truncation are $\veta_1^\top \mu$ and $\veta_2^\top \mu$, respectively, where $\mu=\bbE\by$. Numerically, we find that replacing $\mu$ by a reasonable estimate, such as $\vX_{\cS}\tilde{\beta}$, helps increase the power without losing  FDR control. 

We follow the selected model framework in \cite{tibshirani2013lasso} and \cite{fithian2014optimal} by assuming that   $\{ \cS_1\subset \widehat{\cS}\}$.
Such a property can be guaranteed by the $L^q$-consistency of the Lasso procedure under certain conditions
 \citep{candes2009near,knight2000asymptotics, zhao2006model, van2008high,zhang2008sparsity,meinshausen2006high,meinshausen2009lasso}. 
 Assumption \ref{lasso:a1} below
 is from 
 \cite{buhlmann2011statistics}, which  guarantees the $L^q$-consistency. %

\begin{assumption}\label{lasso:a1} 
\begin{enumerate}[label=(\alph*)]
\item \label{lasso:a1:a} (Compatibility Condition): Let $\phi_0>0$ be a constant. 
For a $p\times 1$ vector $\alpha$ satisfying $||\alpha_{\mathcal{S}_0}||_1 \leq c_0||\alpha_{\mathcal{S}_1}||_1$, we assume  
\begin{equation*}
    ||\alpha_{\mathcal{S}_1} ||_1^2 \leq \frac{p_1}{\phi_0^2} \alpha^T \vX^T \vX \alpha,
\end{equation*}
where the entry $\alpha_i\in \alpha_{\mathcal{S}_0}$ if $i \in \mathcal{S}_0$, and $\alpha_{\mathcal{S}_1}=\alpha\setminus\alpha_{\mathcal{S}_0}$, $c_0$ is a constant depending on the choice of $\lambda_n$, and $\phi_0$ is the compatibility constant. 

\item \label{lasso:a1:b} (Signal Strength): $\min_{j\in \mathcal{S}_1} |\beta_j| >  \frac{64\sigma s_1 }{\phi_0^2} \sqrt{\frac{\log(p)}{n}}$.%
\end{enumerate}
\end{assumption}

\begin{lemma}\label{le:lasso:consist}
 Suppose Assumption \ref{lasso:a1} holds. Consider Lasso with regularization parameter $\lambda_n = 4 \sigma \sqrt{\log(p)/n}$. Then, we have  
$$
P( \mathcal{S}_1 \subset \hat{\mathcal{S}} )\geq 1-\frac{2}{p},
$$
where $\hat{\mathcal{S}}$ is the index set of the nonzero entries in the Lasso estimator $\tilde{\bbeta}$ in  (\ref{eq:org:lasso}). 
\end{lemma}

Lemma \ref{le:lasso:consist} follows directly from the $L_1$ %
convergence result in \cite{buhlmann2011statistics}, Theorem 6.1, and the signal strength condition in 
Assumption \ref{lasso:a1}\ref{lasso:a1:b}. Based on this lemma, event
   $ \{\mathcal{S}_1 \subset \widehat{\mathcal{S}}\} $ 
holds with high probability, implying that the Lasso estimator selects all variables in $\cS_1$. 
The mirror statistic $M_j$ %
further helps to pick out falsely included null variables.

\begin{algorithm}[h!]
\begin{enumerate}

\item Fit Lasso with respect to the original design matrix $\vX$, i.e.,
 \begin{equation*}
 \tilde{\bbeta}=\argmin_{\bbeta}||\vy - \vX \bbeta||_2^2 + \lambda_n || \bbeta ||_1
 \end{equation*}
by cross validation.

\item Parallel FOR $j\in 1,\dots, \hat{\cS}$:
\begin{enumerate}

\item Generate $\vz_j$ from Gaussian distribution with mean zero and covariance matrix $\vI_n$. 

\item Calculate debiased $c_j$ via (\ref{eq:lassocj}).

\item Calculate the mirror statistic
\begin{equation*}
    M_j = |\sigma\Phi^{-1}F_{0,\sigma^2}^{[\cV_1^{-}(\br), \cV_1^{+}(\br)]} ( \hbjplus + \hbjminus )| -  |\sigma\Phi^{-1}F_{0,\sigma^2}^{[\cV_0^{-}(\br), \cV_0^{+}(\br)]} ( \hbjplus - \hbjminus  )|,
\end{equation*}
where $\hbjplus$ and $\hbjminus$ are defined in (\ref{eq:post-beta}) and $\cV_0^{+}, \cV_0^{-}, \cV_1^{+}, \cV_1^{-}$ are defined in (\ref{eq:v}).

\end{enumerate}
END parallel FOR loop.

\item Calculate the cutoff to control FDR at target $q$,
\begin{equation}\label{eq:tauhd}
    \tau_q= \min_t\big\{\frac{\#\{j \;|\;   M_j \leq -t \}}{\#\{j \;|\; M_j\geq t\}  \vee 1} \leq q \big\}.
\end{equation}
\item Output the index of the selected variables: $\hat{\mathcal{S}}_1 = \{j \,|\, M_j \ge \tau_q\}$.
 
\end{enumerate}
\caption{Gaussian mirror algorithm for Lasso.}\label{alg:gm:lasso}

\end{algorithm}

\begin{theorem}\label{thm:coinflip2lassofinite}
Consider the GM design in Definition \ref{gm:lasso} for linear models with the Lasso estimates.
Let $M_j$ be the mirror statistic defined in  (\ref{eq:hdmirror}) and computed using Algorithm \ref{alg:gm:lasso}.  We have 
\begin{equation*}
P(M_j \le -t\;|\;\vz_j) = P(M_j \ge t\;|\;\vz_j) , \,\, \forall\, t>0,
\end{equation*}
for $j\in \mathcal{S}_0$ with probability $1-2/p$.
\end{theorem}

Based on Theorem \ref{thm:coinflip2lassofinite}, 
we use $\#\{j \,|\,M_j \leq  -t \}$ as an (over)-estimate of the number of false positive set, and define an estimate of FDP as 
\begin{equation*}
    \fdphat(t)  \triangleq  \frac{\#\{j \,|\, M_j \leq -t ,  j\in \hat{\cS}_1\}}{\#\{j \,|\, M_j\geq t,  j\in \hat{\cS}_1\}   \vee 1}.
\end{equation*}
Algorithm \ref{alg:gm:lasso} details  the GM procedure for variable selection with FDR control in  high-dimensional regression models. %
Similar to the low-dimensional case  in Section \ref{subsec:ols}, the mirror statistics $M_j$'s can be computed efficiently using parallel computing. %

\section{Theoretical Properties of Gaussian Mirrors}\label{sec:FDR:theorem}
Section \ref{sec:method} introduces the Gaussian mirror method for estimating the FDP. 
We abuse the notation and  use $M_j$ to indicate %
both the one defined in (\ref{eq:mj}) and the one defined in (\ref{eq:hdmirror}).
The key is to design the Gaussian mirror appropriately such that for $j\in \mathcal{S}_0$, %
$
P(M_j\le -t\;|\; \vz_j)=P(M_j\ge t\; | \; \vz_j), \; \forall t>0.
$ 
When we select the $j$-th variable if $M_j\ge t$,  the FDP satisfies
\begin{equation}\label{eqn:fdp1}
\frac{\#\{j\in \mathcal{S}_0: M_j\geq t \}}{\#\{j: M_j\geq t\}	\vee 1} \approx \frac{\#\{j\in \mathcal{S}_0: M_j\leq -t \}}{\#\{j: M_j\geq t\}\vee 1} \leq \frac{\#\{j: M_j\leq -t \}}{\#\{j: M_j\geq t\}	\vee 1}, 
\end{equation}
where the last term is  $\widehat{FDP}(t)$, an (over-)estimate of the FDP defined in (\ref{eq:fdphat}). As shown in Algorithms \ref{alg:gm:ols} and \ref{alg:gm:lasso}, a data-driven threshold $\tau_q>0$ is chosen in (\ref{eq:tau}) as the smallest value such that $\widehat{FDP}(\tau_q)\le q$.  Intuitively, for (\ref{eqn:fdp1}) to hold stably, we need the following assumptions.
\begin{assumption}\label{assump}
(a) (The low-dimensional case): There exists a constant $c\in(0,1)$ such that $p=[cn]$. 
Let $\Omega^0 = (X_{[\cS_0]}^T X_{[\cS_0]})^{-1}$. Then,
\begin{equation*}
\sum_{j,k\in \cS_0} \frac{|\Omega^0_{jk}|} {\sqrt{\Omega^0_{jj}}\sqrt{\Omega^0_{kk}}} < C_1 p_0^\alpha 
\end{equation*}
holds for a constant $C_1>0$ and $\alpha\in(0,2)$.

(b) (The high-dimensional case):  Let $\Sigma$ be the covariance matrix of $X$, $p = O(n^\omega )$ with $\omega >1$, and $p_1= o(n)$. Then $\Sigma$ satisfies the regularity condition
$    1/C_2 \leq \lambda_{\min}(\Sigma) \leq \lambda_{\max}(\Sigma) \leq C_2,
$
where $X_2>0$ is a constant.
\end{assumption}

Assumption \ref{assump}(a) requires that the sum of pairwise partial correlations of the design matrix is bounded by $C_1 p_0^\alpha$, where $\alpha\in(0,2)$. This requirement is satisfied by many designs such as the autoregressive and the constant-correlation design. 
For the high-dimensional case, we consider a bounded eigenvalue condition in Assumption 2(b), which is commonly used as a regularization condition in high-dimensional analysis \citep{cai2017confidence}. Additionally, assume that $p_1=o(n)$. 
The following two lemmas characterize the sum of pairwise covariances of the mirror statistics, which is the key to establish asymptotic FDR control for GM. 
In the low-dimensional setting, the sum is taken over the null set $\cS_0$, whereas in 
the high-dimensional setting, the sum is taken over the selection set $\hat{\cS}$ since the mirror statistics are derived from the post-selection Lasso.

\begin{lemma}\label{lemma:suff:ols}
Consider the low-dimensional setting in which the mirror statistics $M_j$'s are defined as in (\ref{eq:mj}). If Assumption~\ref{assump}(a) holds, then 
\begin{eqnarray}\label{def:weak:dep}
&& \sum_{j,k\in \mathcal{S}_0}\cov\big( \mathbbm{1}(M_j\ge t), \mathbbm{1}(M_k\ge t) \big) \le C_1^\prime |\mathcal{S}_0|^{\alpha_1}, \; \forall \;  t,\nonumber
\end{eqnarray}
where $\alpha_1\in(0,2)$, $C_1^\prime>0$ is a constant, and  $\mathbbm{1}(\cdot)$ is the indicator function.
\end{lemma}

\begin{lemma}\label{lemma:suff:lasso}
Consider the high-dimensional setting, in which the  mirror statistics $M_j$'s are defined as in (\ref{eq:hdmirror}). If Assumptions~\ref{lasso:a1}  and \ref{assump}(b) hold, then  
\begin{eqnarray}\label{def:weak:dep}
&& \sum_{j,k\in\widehat{\mathcal{S}}\cap\mathcal{S}_0}\cov\big( \mathbbm{1}(M_j\ge t), \mathbbm{1}(M_k\ge t) \big) \le C_2^\prime |\widehat{\mathcal{S}}\cap\mathcal{S}_0|^{\alpha_2}, \; \forall \;  t>0, \nonumber
\end{eqnarray}
where $\alpha_2\in(0,2)$, $C_2^\prime>0$ is a constant, and $\mathbbm{1}(\cdot)$ is the indicator function.
\end{lemma}

Next, we derive some asymptotic properties of the GM method.   For the low-dimensional case, we define a few quantities:
\begin{gather*}
V(t) =\frac{\#\{ j:  j\in \mathcal{S}_0, M_j \le -t \}}{p_0} ,\, V'(t) =\frac{\#\{ j:  j\in \mathcal{S}_0, M_j \ge t \} }{p_0},\,
V^{1}(t) =\frac{\#\{ j:  j\in \mathcal{S}_1, M_j \ge t \}}{p_1} , \\
G_0(t) = \lim_p \frac{\sum_{j\in\mathcal{S}_0} \mathbb{E} 1(M_j\le -t)}{p_0},\, G_1(t) = \lim_p \frac{\sum_{j\in\mathcal{S}_1} \mathbb{E} 1(M_j\ge t)}{p_1},
\end{gather*}
and
\begin{equation*}
    FDP(t) :=\frac{V^{\prime}(t)}{(V^{\prime} + r_p V^1(t))\vee 1/p}, \quad FDP^{\infty}(t) :=\lim_{p}\frac{V^{\prime}(t)}{(V^{\prime} + r_p V^1(t))\vee 1/p},
\end{equation*}
where $r_p = p_0/ p_1$ and $FDP^{\infty}(t)$ is the pointwise limit of $FDP(t)$.  

With a slight abuse of the notation, we define a similar set of quantities for the high-dimensional setting as:
\begin{gather*}
V(t) =\frac{\#\{ j:  j\in \mathcal{S}_0, M_j \le -t \}}{|\widehat{\cS}\cap\cS_0|} ,\, V'(t) =\frac{\#\{ j:  j\in \mathcal{S}_0, M_j \ge t \} }{|\widehat{\cS}\cap\cS_0|},\,
V^{1}(t) =\frac{\#\{ j:  j\in \mathcal{S}_1, M_j \ge t \}}{|\widehat{\cS}\cap \cS_1|} , \\
G_0(t) = \lim_p \frac{\sum_{j\in\mathcal{S}_0} \mathbb{E} 1(M_j\le -t)}{|\widehat{\cS}\cap\cS_0|},\, G_1(t) = \lim_p \frac{\sum_{j\in\mathcal{S}_1} \mathbb{E} 1(M_j\ge t)}{|\widehat{\cS}\cap\cS_1|},
\end{gather*}
and
\begin{equation*}
    FDP(t) :=\frac{V^{\prime}(t)}{(V^{\prime} + r_p V^1(t))\vee 1/|\widehat{\cS}|}, \quad FDP^{\infty}(t) :=\lim_{p}\frac{V^{\prime}(t)}{(V^{\prime} + r_p V^1(t))\vee 1/|\widehat{\cS}|}
\end{equation*}
where $r_p = |\widehat{\cS}\cap \cS_0|/|\widehat{\cS}\cap\cS_1|$. We have the following results.

\begin{lemma}\label{lem:slln}
Suppose Assumption \ref{assump}(a) holds for the low dimensional case,  Assumption~\ref{lasso:a1} and \ref{assump}(b) hold for the high-dimensional case, and $G_0(t)$ is a continuous function. Then, we have
\[
 \sup_t \left|{V(t)} - G_0(t)\right|\xrightarrow[]{p}0,\quad  \sup_t\left|{V'(t)}- G_0(t)\right|\xrightarrow[]{p}0, \quad \sup_t\left|{V^1(t)}- G_1(t)\right|\xrightarrow[]{p}0.\]
\end{lemma}
With the aid of this lemma, we show that the proposed GM method controls the FDR asymptotically for linear models in both low- (Theorem~\ref{thm:fdr}) and high- (Theorem~\ref{thm:fdr:lasso}) dimensional cases.  

 \begin{theorem}\label{thm:fdr}
 Let $M_j$'s be the mirror statistics calculated in Algorithm~\ref{alg:gm:ols} for the low-dimensional setting. For any given $q\in (0,1)$, we choose $\tau_q >0 $ according to (\ref{eq:tau}). If Assumption \ref{assump}(a) holds and there exists $t>0$ such that $FDP^{\infty}(t)<q$, then, as $p\rightarrow\infty$,
\begin{equation*}
    \mathbb{E}\left[\frac{\#\{j: j \in \mathcal{S}_0, \hbox{ and } j \in \hat{\mathcal{S}}_1\}}{\#\{j: j \in \hat{\mathcal{S}}_1\} \vee 1}\right] \le q + o(1).
\end{equation*}
\end{theorem}

\begin{theorem}\label{thm:fdr:lasso}
Let $M_j$'s be the mirror statistics calculated using Algorithm \ref{alg:gm:lasso} for the high-dimensional setting. For any given $q\in (0,1)$, we choose $\tau_q >0 $ according to (\ref{eq:tauhd}). If Assumptions \ref{lasso:a1} and \ref{assump}(b) hold and there is a $t>0$ such that $FDP^{\infty}(t)<q$, then, as $p\to\infty$, 
\begin{equation*}
    \mathbb{E}\left[\frac{\#\{j: j \in \mathcal{S}_0, \hbox{ and } j \in \hat{\mathcal{S}}_1\}}{\#\{j: j \in \hat{\mathcal{S}}_1\} \vee 1}\right] \le q + o(1).
\end{equation*}
 \end{theorem}
When conditioning on the event  $\{S_1\subset \hat{\mathcal{S}} \}$, we can follow the proof of Theorem \ref{thm:fdr} to show that FDR conditioning on the selection is less than or equal to $q$ asymptotically. According to Lemma \ref{le:lasso:consist}, the probability that $\mathcal{S}_1\subseteq \hat{\mathcal{S}}$ is greater than $1-\frac{2}{p}$. Consequently, as $p\to \infty$ , the FDR is less than or equal to $q+o(1)$, which concludes the proof of the theorem.

\section{Estimating the Number of False Discoveries}\label{sec:var:fdp}

The GM approach  introduced in Section \ref{sec:method} provides a way to select significant features subject to FDR control at a designated level. %
Next, we address a related problem that is of interest to some practitioners. Suppose a scientist is only allowed to explore no more than 100 selected features due to a limited budget, then the following questions are of immediate concerns: (1) How should they select the list of top 100? (2) How many false discoveries (FDs) do they expect to have? (3)  Can statisticians provide an uncertainty measure for the estimated FDs?

The first question is easy to address based on the mirror statistic $\{M_j\}$ $(j=1,\cdots,p)$. We order the mirror statistics  decreasingly as $M_{(1)}\geq \dots \geq M_{(p)}$ and choose the top $k$ features as those corresponding to  the set of the top-$k$  mirror statistics. Let the selected set be $\mathcal{I}_k$, i.e.,
\begin{equation*}
    \mathcal{I}_k = \big\{j: M_j \in \{M_{(1)}, \dots, M_{(k)}\} \big\},
\end{equation*}
where $k\leq \#\{j\;|\;M_j > 0\}$ since the GM procedure does not select variables with negative mirror statistics.
Let  $FD(k)= |\mathcal{I}_k \cap \mathcal{S}_0| $ denote the true number of FDs in the top-$k$ list and let $\EE[FD(k)]$ denote the expected number of FDs. %
Since the mirror statistics are distributed symmetrically for the null variables, we estimate $\EE[FD(k)]$ as 
\begin{equation}\label{eq:est:number:fdr}
\widehat{FD}(k)= \#\{M_j < -M_{(k)}\}.
\end{equation}
Under high dimensional settings, Theorem~\ref{thm:fd} below
shows that the error bound between $\widehat{FD}(k)$ and $\EE[FD(k)]$ is  $o_p(k)$. As $k$ increases, the error bound also increases although the relative error gets smaller. Theorem \ref{thm:fdmean} states that $\widehat{FD}(k)$ is an unbiased estimator of $\EE[FD(k)]$ with high probability.

\begin{algorithm}[h]
\begin{enumerate}
\item Parallel FOR $b=1,\cdots,B$:
\begin{enumerate}
\item  For the low-dimensional setting, fit linear regression model via  OLS, i.e., minimizing (\ref{eq:plsgm}); \\
For the high-dimensional setting, get the Lasso solution by minimizing  (\ref{eq:org:lasso}). Let $\hat{\vy} = (\hat{y}_{1}, \hat{y}_{2},\cdots, \hat{y}_{n})$ be the fitted values, and $ \vr = \vy - \hat{\vy}= (r_{1}, r_{2}, \cdots, r_{n})$ be the residuals.  %

\item  Sample from $\vr$ with replacement to get $\vr^{(b)}=(r_{1}^{(b)}, r_{2}^{(b)},\cdots, r_{n}^{(b)})$. Let $\vy^{(b)}=\hat{\vy} +\vr^{(b)}$.

\item Apply Algorithm \ref{alg:gm:ols} (low-dimensional) or  Algorithm \ref{alg:gm:lasso} (high-dimensional) by replacing $\vy$ with $\vy_j^{(b)}$, and calculate the mirror statistics $M_j^{b}$, for $j=1,\dots, p$, accordingly.

\end{enumerate}
End Parallel FOR.

\item  Output the  bootstrap estimate  $\{M_{1}^{b}, M_{2}^{b},\cdots, M_{p}^{b}\}_{b=1}^B$.

\item  For $b=1,\dots, B$, we calculate the estimate of the number of false discoveries based on  $\{M_{1}^{b}, M_{2}^{b},\cdots, M_{p}^{b}\}$ as 
\begin{equation*}
    \widehat{FD}^b(k) = \#\{M^{b}_j < -M^{b}_{(k)}\},
\end{equation*}
where $\{M^{b}_{(1)}, \dots M^{b}_{(p)}\}$ are the decreasingly ordered mirror statistics.

\item Construct the confidence interval of $\EE[FD(k)]$ as $[\widehat{FD}_{(\alpha/2)}(k), \widehat{FD}_{(1- \alpha/2)}(k)]$ and upper confidence interval as $[0, \widehat{FD}_{(1- \alpha)}(k)]$.

\end{enumerate}
\caption{Bootstrap distribution of $M_j$ and $\widehat{FD}(k)$. }\label{alg:bootstrap}
\end{algorithm}

\begin{theorem}\label{thm:fd}
Suppose that Assumptions 1 and 2(b) hold, $p/n\to \infty$, $k\leq \#\{j\;|\;M_j > 0\}$, and $k/p_1 = O(1)$. Let $M_j$ be the mirror statistic calculated from Algorithm \ref{alg:gm:lasso}.
We have
\begin{equation*}
    \lim_{n\to \infty} P\left(\frac{1}{k}\left|\widehat{FD}(k) - \EE[FD(k)]\right| > \epsilon\right) = 0,
\end{equation*}
for any $\epsilon >0$.
\end{theorem}

\begin{theorem}\label{thm:fdmean}
Suppose that Assumptions 1 and 2(b) hold, $k\leq \#\{j\;|\;M_j > 0\}$. Let $M_j$ be the mirror statistic calculated from Algorithm \ref{alg:gm:lasso}.
We have
\begin{equation*}
    P\left(\EE[\widehat{FD}(k)] = \EE[FD(k)]\right) > 1-\frac{2}{p}.
\end{equation*}
\end{theorem}

Next, we describe a bootstrap method to construct a confidence interval for $\EE[FD(k)]$. 
The method starts by fitting the regression model based on the original design matrix using either the least squares (if $p<n$) or the Lasso method, and obtaining the fitted values $\hat{\vy}=(\hat{y}_{1},\cdots, \hat{y}_{n})$ as well as the residuals  $\bgamma =\vy-\hat{\vy}$.
Then, for $b=1,\ldots, B$, we generate independently the $b$-th ``bootstrap sample" $\vy^{(b)}= \hat{\vy} + \bgamma^{(b)}$, where the $\bgamma^{(b)}$ is drawn randomly from $\bgamma$ with replacement. With the bootstrapped ``observations" $\vy^{(b)}$, we calculate the mirror statistics $\{M_1^{(b)}, \dots, M_p^{(b)}\}$ using Algorithm \ref{alg:gm:ols} (Algorithm \ref{alg:gm:lasso} for Lasso). The $B$ sets of bootstrap mirror statistics are 
denoted as $\{M_{1}^{(b)}, M_{2}^{(b)},\cdots, M_{p}^{(b)}\}_{b=1}^B$, which give rise to a set of $B$ bootstrap estimates of $FD(k)$: ${\cal B}_{FD}=\{\widehat{FD}^{(1)}(k),\dots, \widehat{FD}^{(B)}(k)\}$. 
A bootstrap confidence interval for $\EE[FD(k)]$ can be constructed as  the upper and lower $\alpha/2$ quantiles of the sample ${\cal B}_{FD}$, denoted as $\widehat{FD}_{(\alpha/2)}(k)$ and $\widehat{FD}_{(1-\alpha/2)}(k)$, respectively. We may also first re-center the set ${\cal B}_{FD}$ to have mean $\widehat{FD}(k)$ and then use the corresponding quantiles. 
If a budget-sensitive domain scientist is only interested in a $(1-{\alpha})100\%$ upper confidence bound of $\EE[FD(k)]$, then she/he can use $\widehat{FD}_{(1-\alpha)}(k)$.
\section{Numerical Studies}\label{sec:simulation}
In this section, we compare  numerical performances of various multiple testing methods for regression models.  The first class of methods is based on the procedure of \cite{benjamini1995controlling}, abbreviated as BH. As shown in the following two subsections, the test statistics are calculated differently depending on $p$ and $n$.
The second class of methods is based on  knockoffs,
including the knockoff filter when $p<n/2$ \citep{barber2015controlling}, and 
the model-X knockoff when $p$ is large \citep{candes2018panning},
For the model-X knockoff, we construct knockoffs with both a known covariance matrix (model-X-true) and a second-order estimate of the covariance matrix (model-X-est). To our surprise, we observe that the knockoff method based on the known covariance matrix becomes extremely conservative for the constant correlation setting with correlation coefficient larger than 0.5. 
We implement a modified version of model-X knockoff  (model-X-fix), which significantly improves the power.
In all simulation settings, we set the designated FDR level $q$ as 0.1.

\subsection{Low-dimensional case ($p<n$)}\label{sec:sim:low:p}

We set $n=1000$ and $p=300$. For the GM approach, we calculate the mirror statistics based on Algorithm \ref{alg:gm:ols}. For the BH method, we calculate the $z$-statistics $z_1,\dots,z_p$ for the OLS estimates, i.e., $z_j=\hat{\beta}_j/(\sigma\sqrt{(X^T X)^{-1}_{jj}})$. The $j$-th variable is selected if $|z_j|>\tau_q^{BH}$,  where $\tau_q^{BH}$ is chosen as
\begin{equation*}
\tau_q^{BH} = \min_t\left\{ \frac{P(|N(0,1)|\geq t)}{\# \{j\,|\, |z_j| \geq t\} } \leq \frac{q}{p}\right\}.
\end{equation*}
The knockoff and model-X knockoff statistics are calculated based on the Lasso estimators. For different methods, we calculate the FDR and the power based on 100 replications. %
 In each of the following settings, we randomly set 240 coefficients of $\beta_i$'s as zero and generate the remaining $60$ nonzero coefficients independently from  $N(0, (20/\sqrt{n})^2)$. The response variable $\vy$ is generated according to Eq (\ref{eqn:lm}) with $\sigma=1$. The design matrix are composed of i.i.d. rows with each row generated from $N(\mathbf{0},\Sigma)$.

\textbf{(i) Power decay correlation.} The covariance matrix $\bSigma$ is autoregressive, i.e., $\sigma_{ij}$, the element at the i-th row and j-th column is $ \kappa^{|i-j|}$, where $\kappa$ is taken among $0, 0.2, 0.4, 0.6,$ and $0.8$, respectively. Figure \ref{fig:ld} (a1-a2) show plots of FDPs and proportions of true rejections of these five methods. Note that the middle bar in the box is the mean of these FDPs and proportions, a.k.a., the FDR and power. As shown in Figure \ref{fig:ld} (a1), the FDPs of the five methods are all around the targeted value $0.1$. Figure \ref{fig:ld} (a2) shows that the GM method  has the highest power among these five methods; and the discrepancies between the GM and other methods increase with respect to  $\kappa$.  %
Particularly, when  $\kappa$ is $0.8$, the power of the knockoff method becomes very small with some extreme cases of having no rejections at all. 

\begin{figure}[hp!]
\centering 
\begin{tabular}{cc}
{\hspace{-30pt}\footnotesize (a1) Power decay auto-correlation }& {\hspace{-30pt}\footnotesize (a2) Power decay auto-correlation} \\
    \includegraphics[width=0.45\textwidth]{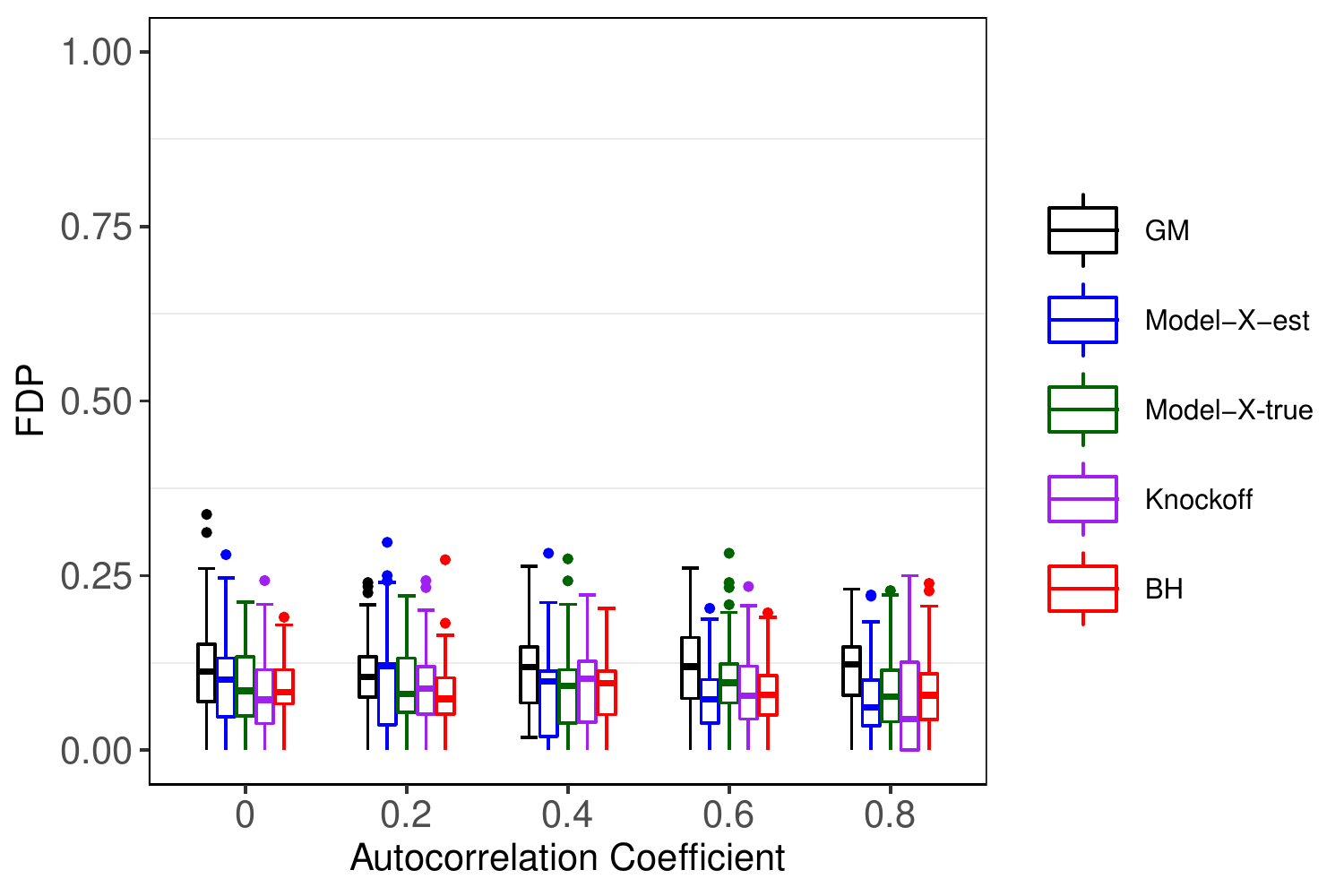}&\includegraphics[width=0.45\textwidth]{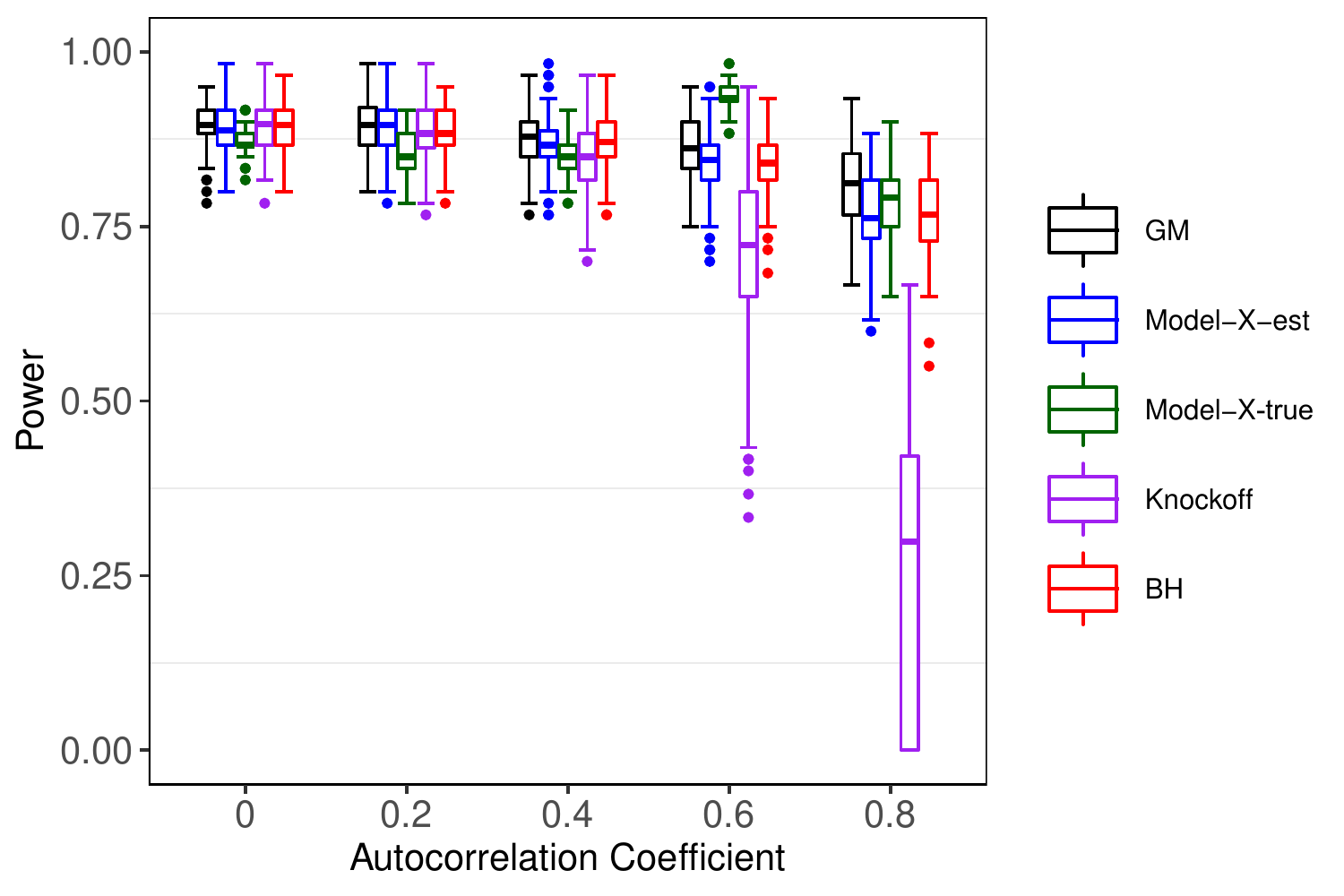} \\
{\hspace{-30pt}\footnotesize (b1) Constant positive correlation }& {\hspace{-30pt}\footnotesize (b2) Constant positive correlation } \\
    \includegraphics[width=0.45\textwidth]{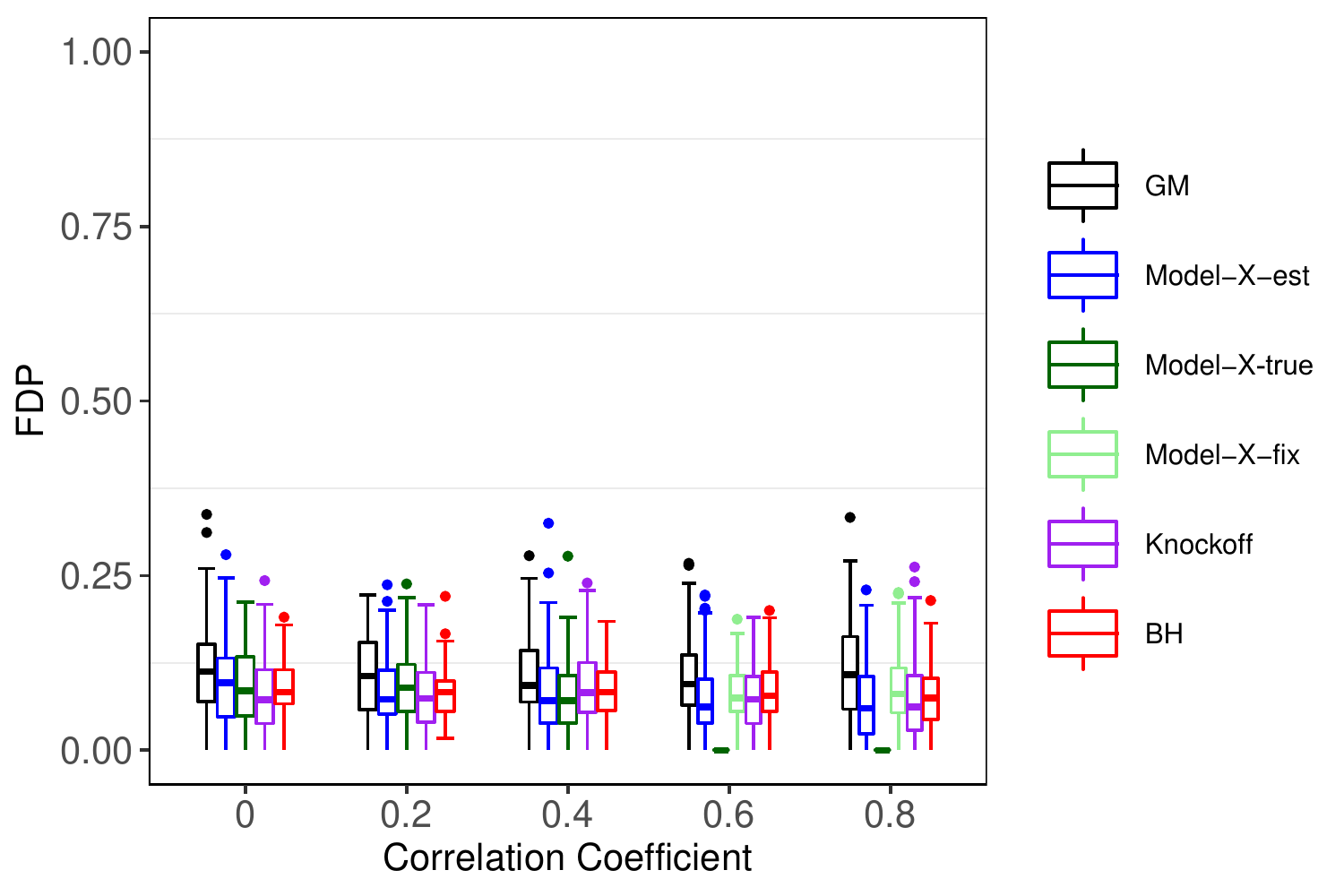}&\includegraphics[width=0.45\textwidth]{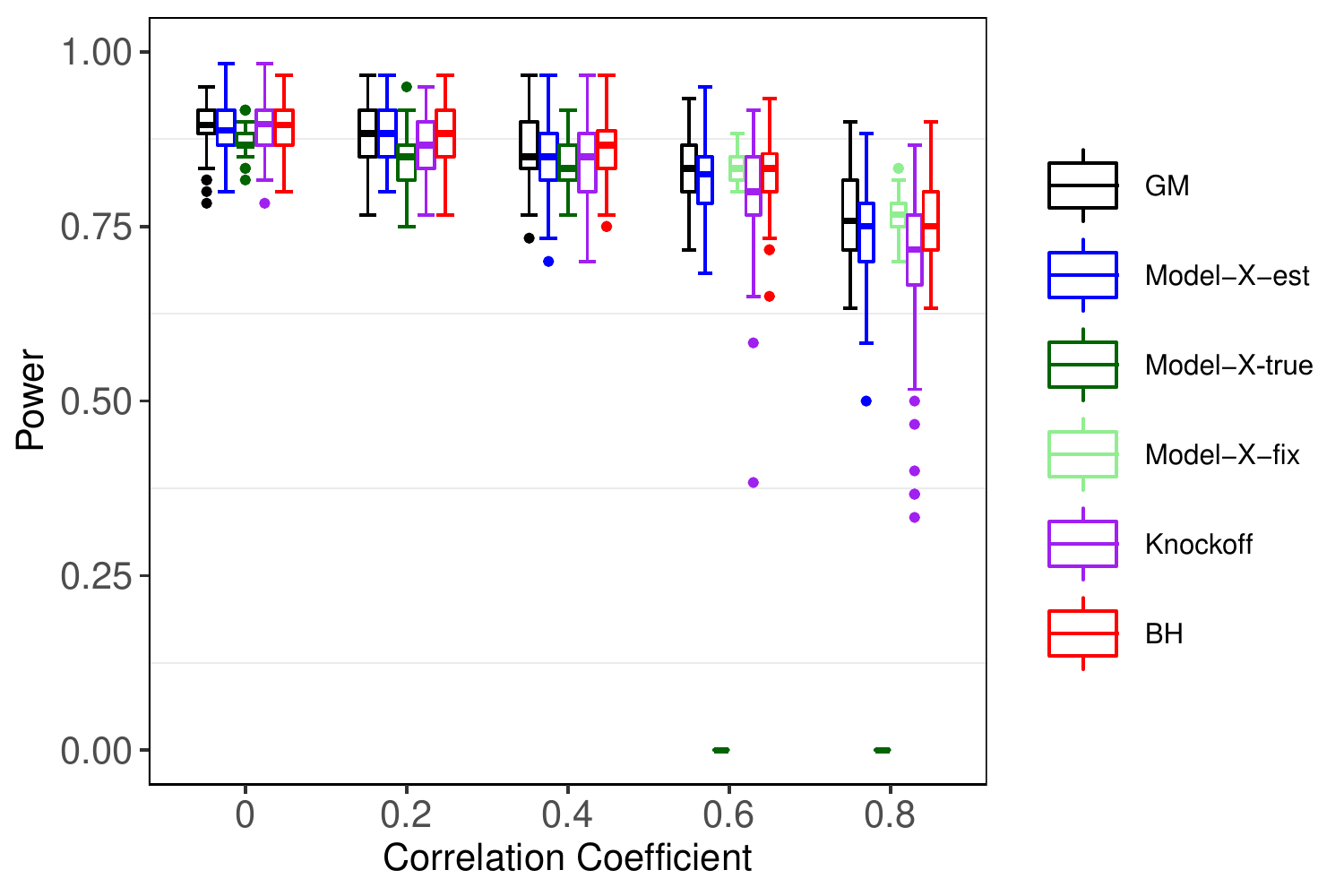} \\
{\hspace{-30pt}\footnotesize (c1) Constant partial correlation  }& {\hspace{-30pt}\footnotesize (c2) Constant partial correlation } \\
    \includegraphics[width=0.45\textwidth]{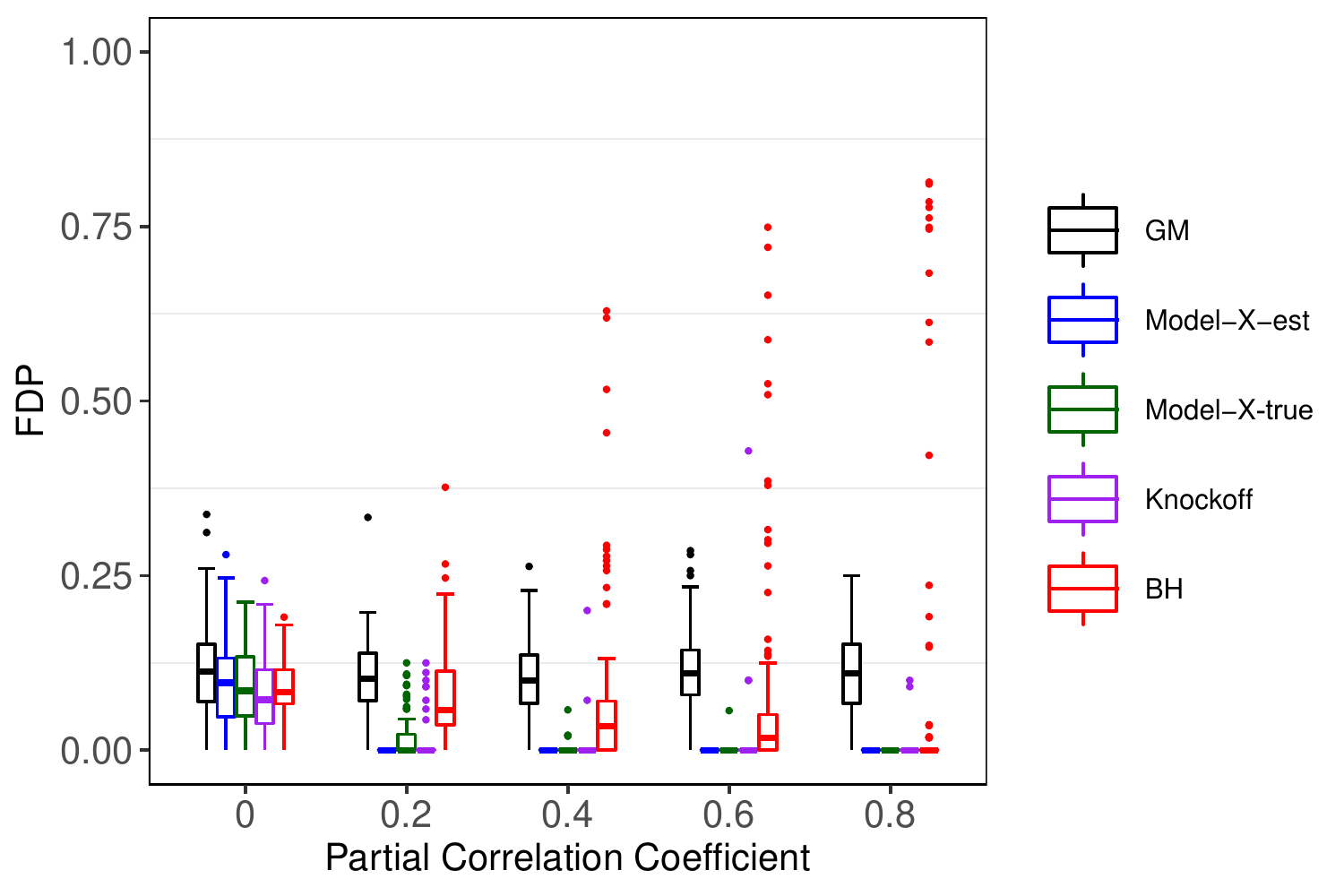}&\includegraphics[width=0.45\textwidth]{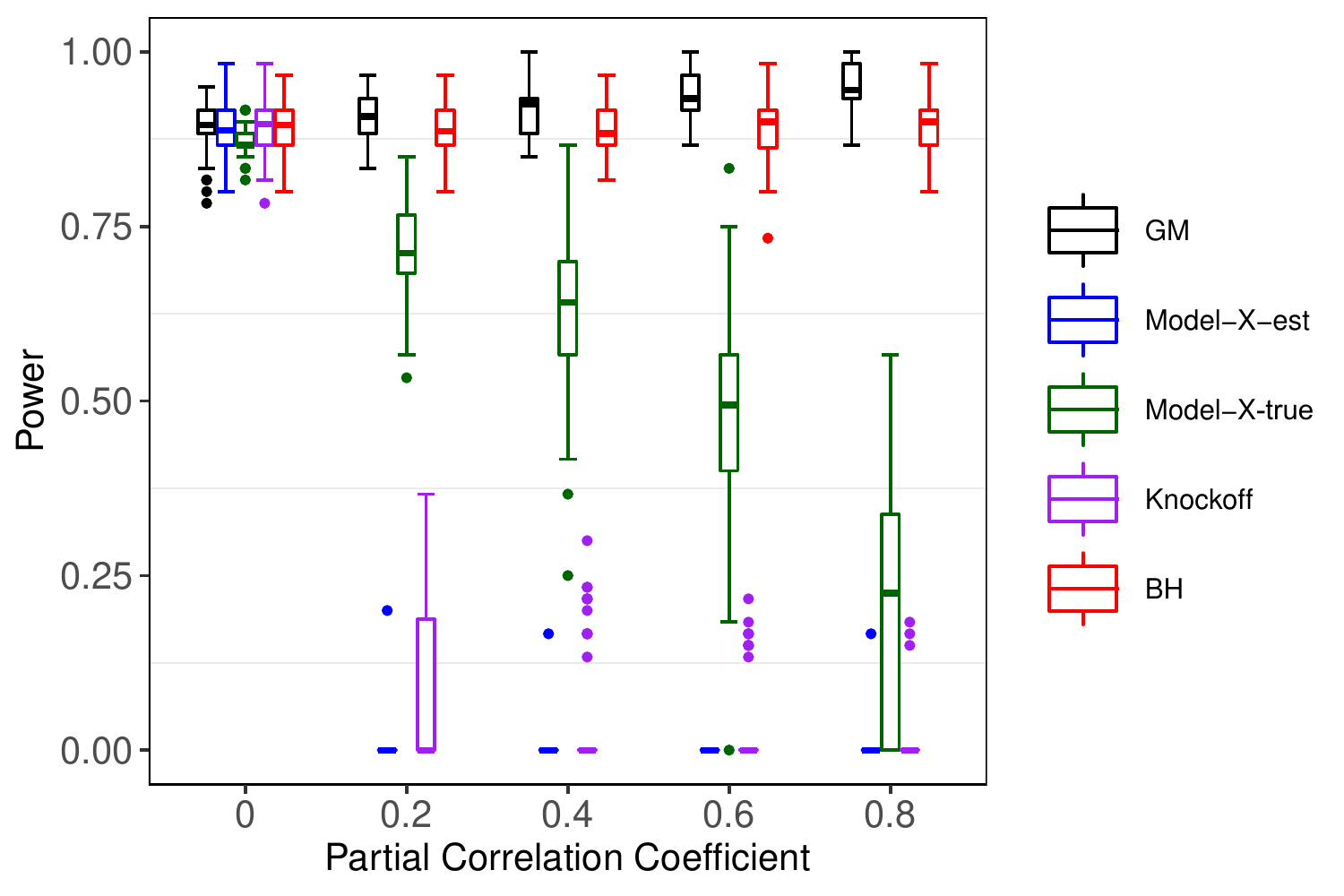}
\end{tabular}
 \caption{The plots of the FDPs and the proportions of true rejections for the low-dimensional case. %
 The upper and lower hinges of each box correspond to the first and third quartiles. The middle bar is the 
 mean of the FDPs(left panel) and the proportions of true rejections (right panel). The interval is calculated as the mean $\pm$ one  standard deviation of the FDPs. 
 }\label{fig:ld}
\end{figure}

{\textbf{(ii) Constant positive correlation.}} Here we let $\sigma_{ij}=\rho^{\mathbbm{1}(i\neq j)}$, with $\mathbbm{1}(\cdot)$ being an indicator function. As shown in Figure \ref{fig:ld} (b1-b2), correlations among the $\bbeta$'s are small and the $t$-statistics are only weakly correlated because of  low partial correlations. Hence, the BH method controls the FDR well and also has a good power. %
Model-X-est and  the GM method come as the close second with the results very similar to those of the BH method. %

When $\rho \geq 0.5$, the power of  Model-X-true drops to near-zero due to a numerical issue in the original Model-X software.
The Model-X procedure generates knockoffs from $N(\mu, V)$, where $V$ is obtained via solving a semi-definite programming problem. When  $\rho\ge 0.5$, $V$ is nearly a rank-one matrix, which results in high collinearity among the generated knockoffs and significantly reduces the power. Dongming Huang and Lucas Jensen (personal communication) suggested a simple remedy, which is to project each variable onto the orthogonal space of the first principal  component of $\Sigma$ and add a small perturbation as
\begin{equation*}
    \bx^{new}_j = \bx_j - \frac{1}{\delta}\vX \Sigma^{-1}\mathbf{1}_p + \frac{1}{\sqrt{\delta}}\vz_j,
\end{equation*}
where $\delta = (\rho - \rho^2 \mathbf{1}_p^T\Sigma^{-1}\mathbf{1}_p)^{-1}$ and $\vz_j\stackrel{iid}{\sim} N(0,I_n)$. Then, $\vX^{new} =(\vx_1^{new}, \dots,\vx_p^{new})$ serves a new design matrix for constructing knockoffs via the standard Model-X procedure with the covariance matrix  $I_p$.
This modification substantially improves the power  (as shown in Figure~\ref{fig:ld}(b1-b2) and Figure~\ref{fig:hd}(b1-b2) for large-$\rho$ cases), but the trick cannot be extended to general cases. %
In general, power decreases as $\rho$ increases for all the methods. GM appears to be affected the least.

{\textbf{(iii) Constant partial correlation.}}
In this setting, the precision matrix $Q=\bSigma^{-1}$ has  constant off-diagonal elements, i.e., $q_{ij}=\tau^{1(i\neq j)}$.  Correlations among the $\vX$ can be very small when $\tau$ is large. For example, when $\tau=0.6$, the off-diagonal entries of  $\bSigma$ is around $-0.0083$ and the diagonal entry is $2.4917$. The results are reported in Figure \ref{fig:ld}(c1-c2). Both the GM and BH methods work well in terms of the FDR and power, while GM is slightly more powerful. When $\tau>0$, both the knockoff filter and Model-X-est  are extremely conservative. For instance, when $\tau=0.6$, both the FDR and power of these two methods drop to zero. %
Mode-X-true works better, but its power decreases much faster with respect to $\rho$ than the GM and BH methods.

\subsection{High-dimensional case ($p\geq n$)}\label{sec:sim:high:p}

We consider the high-dimensional case with $p=1000$ and $n=300$, in which Algorithm \ref{alg:gm:lasso} is used for the GM method. Similar to \cite{candes2018panning}, we implement the marginal BH (henceforth, BH-ma), which is based on the p-values of the marginal regression between $\vy$ and $\vx_j$, $j=1,\cdots, p$. We also consider the data splitting BH (henceforth, BH-ds), which employs Lasso to select variables based on the first half of the data, and applies the BH method to the p-values of the OLS estimates obtained from the second half of the data for the selected variables.
For the Model-X knockoff, we use both the Gaussian knockoffs with a known covariance matrix and the second-order knockoffs with estimated covariance matrix as implemented in the R package \textit{knockoff}. Note that the original knockoff method can not be applied  when $p\geq n$. 

\smallskip

{\textbf{(i) Power decay correlation.}} %
As shown in Figure \ref{fig:hd} (a1), FDPs of  GM  are around $0.1$ in all cases.  The FDRs of Model-X-est, Model-X, BH-ds  are around 0.1 when $\kappa\leq 0.6$. When $\kappa = 0.8$,  the FDR of Model-X-est could be as large as $0.23$. This inflation is also observed in \cite{candes2018panning}. The FDR of BH-ms is also inflated when $\kappa\geq 0.4$. The FDR of the GM method is controlled at the designated level well. In Figure \ref{fig:hd} (a2), it is shown that GM and two Model-X methods have comparable powers, whereas the two BH methods do not perform well.

\begin{figure}[hp]
\centering 
\begin{tabular}{cc}
    {\hspace{-30pt}\footnotesize (a1) Power decay auto-correlation }& {\hspace{-30pt}\footnotesize (a2) Power decay auto-correlation} \\
    \includegraphics[width=0.45\textwidth]{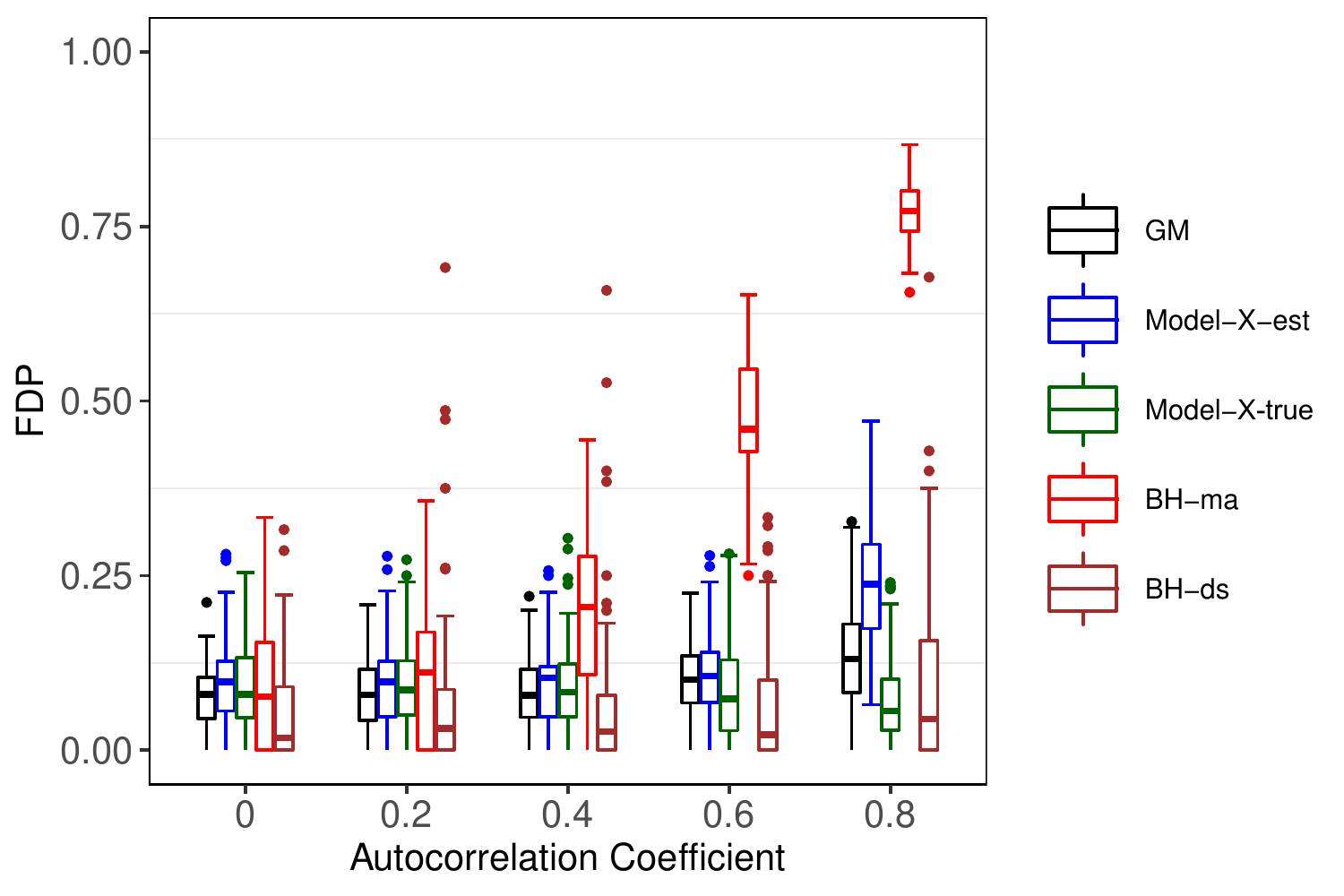}&\includegraphics[width=0.45\textwidth]{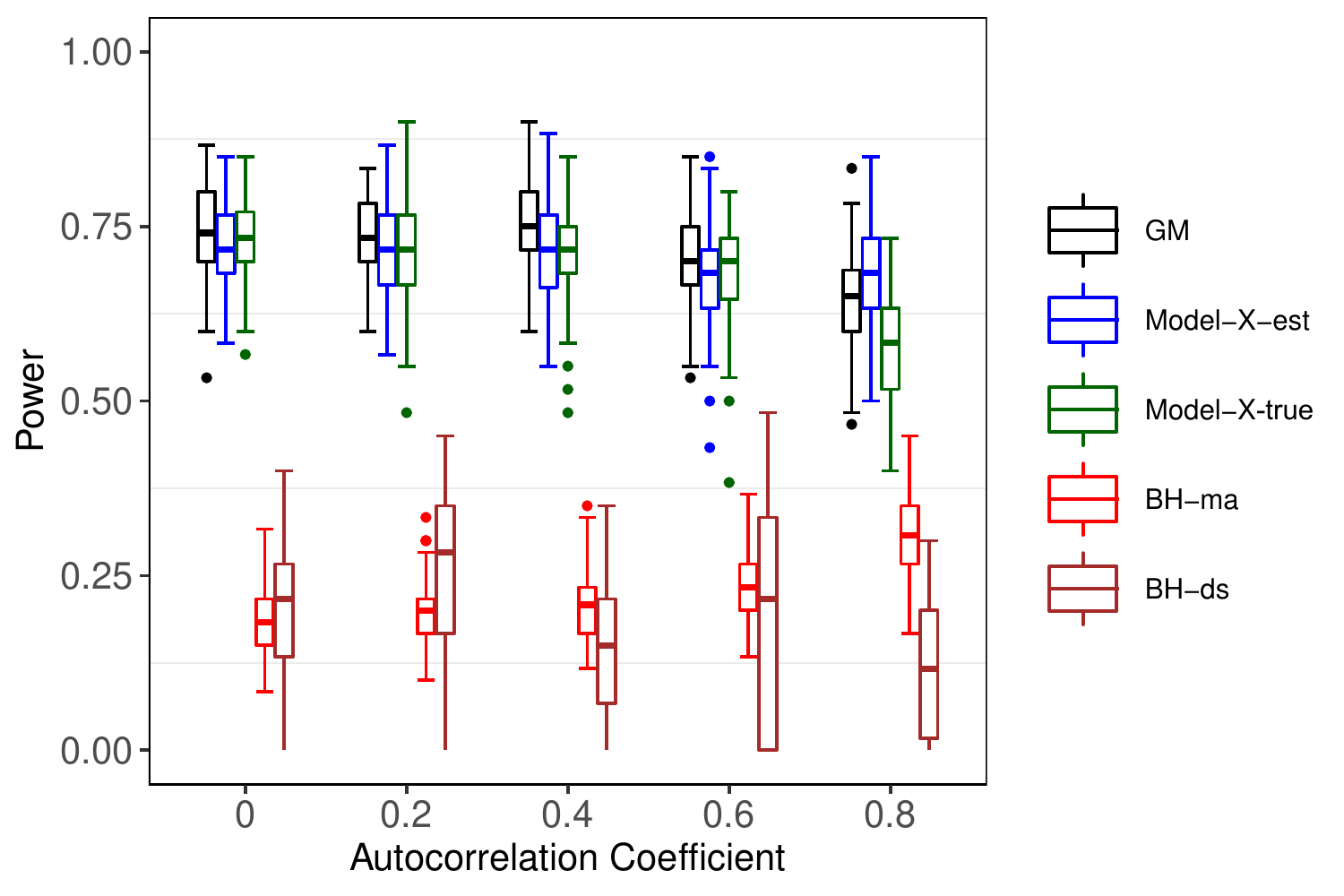} \\ 
    {\hspace{-30pt}\footnotesize (b1) Constant positive correlation }& {\hspace{-30pt}\footnotesize (b2) Constant positive correlation } \\
    \includegraphics[width=0.45\textwidth]{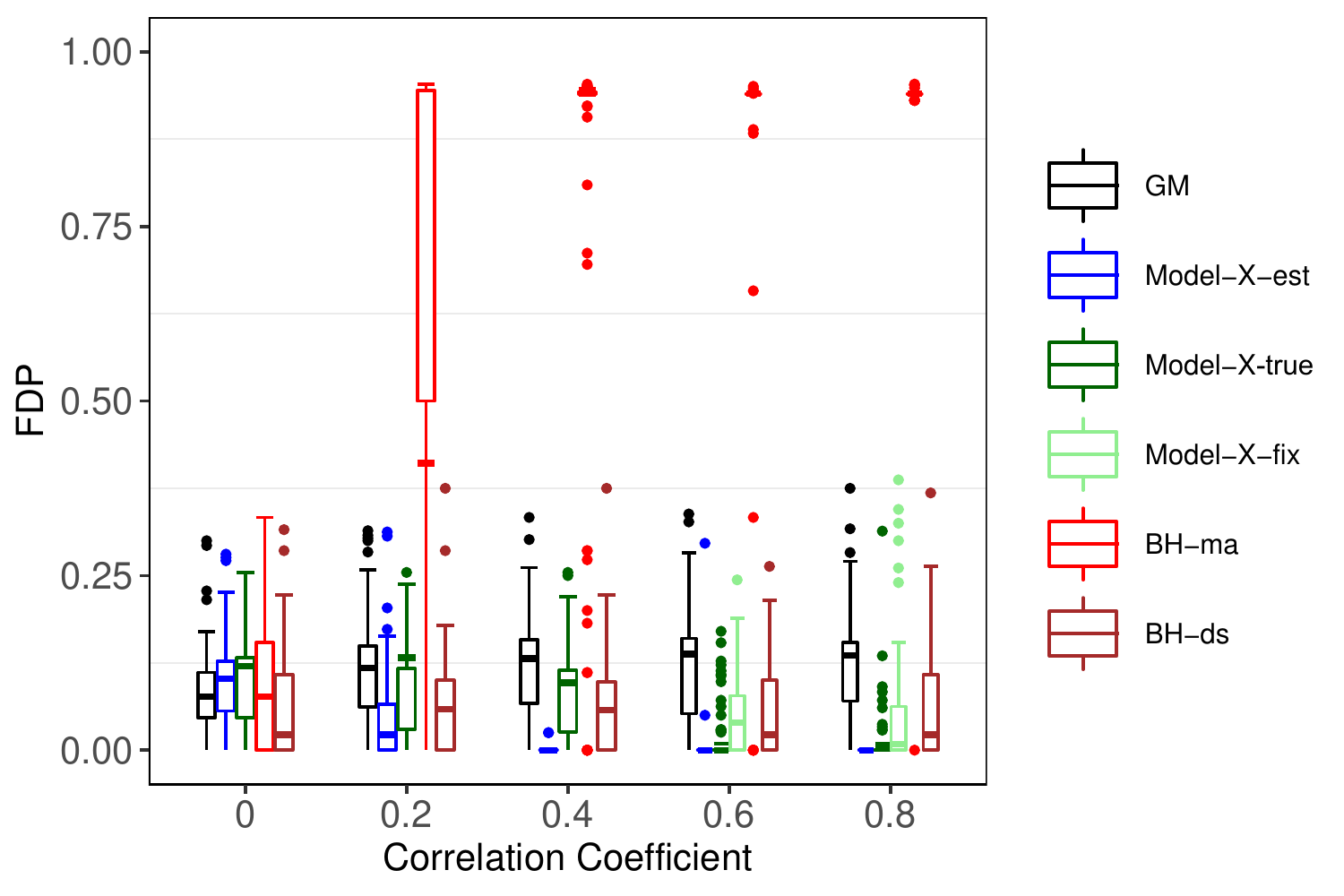}&\includegraphics[width=0.45\textwidth]{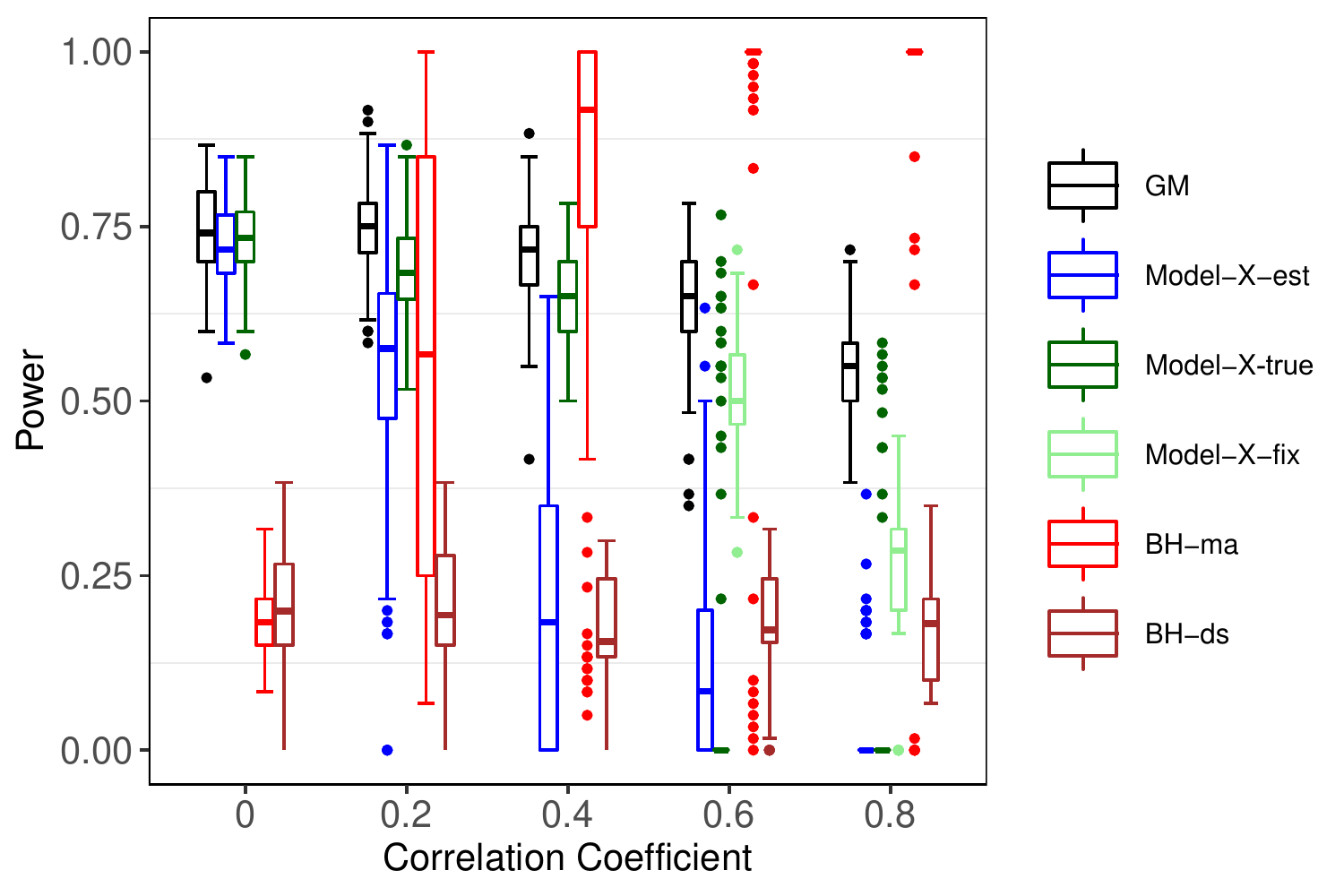} \\
    {\hspace{-30pt}\footnotesize (c1) Constant partial correlation  }& {\hspace{-30pt}\footnotesize (c2) Constant partial correlation } \\
    \includegraphics[width=0.45\textwidth]{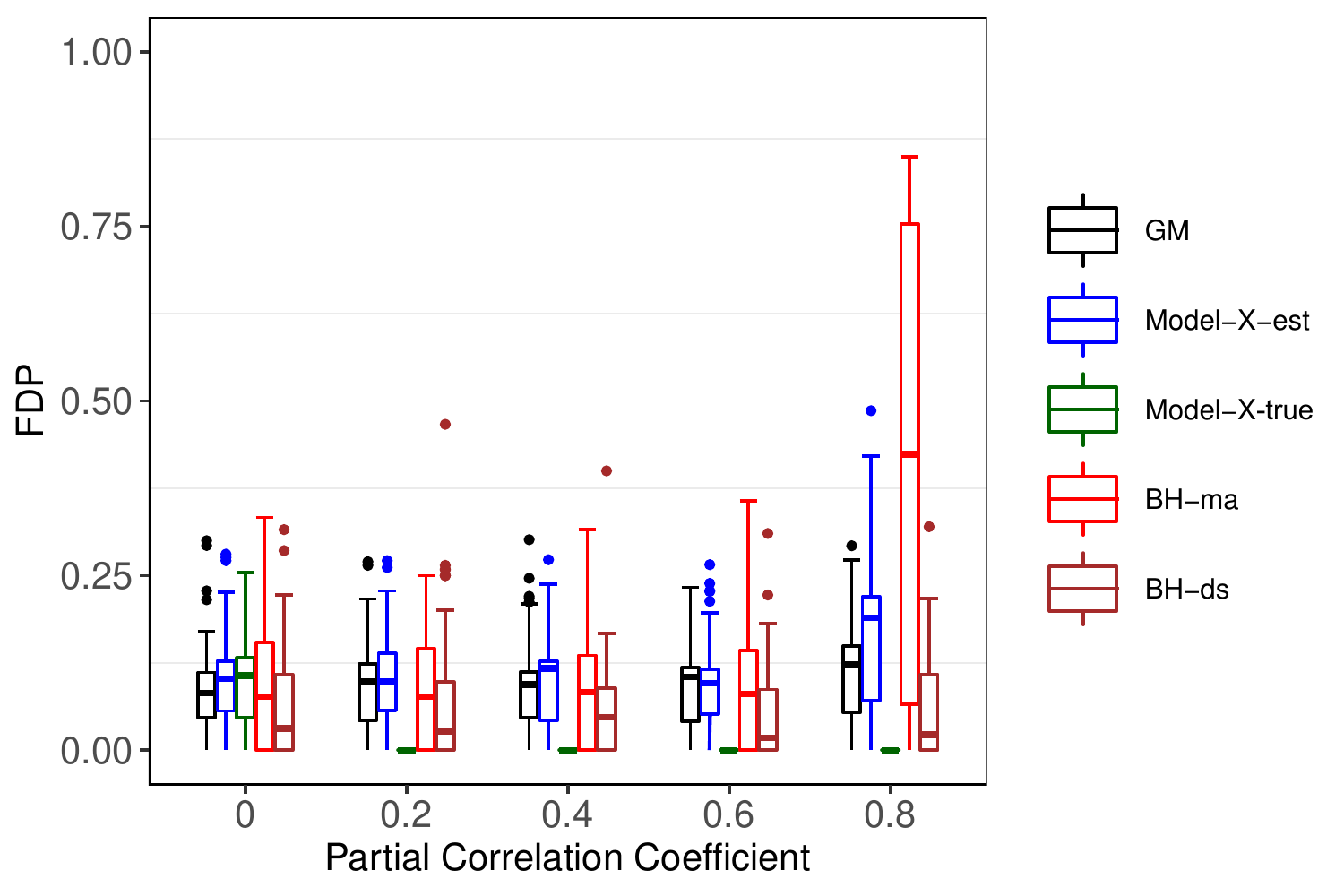}&\includegraphics[width=0.45\textwidth]{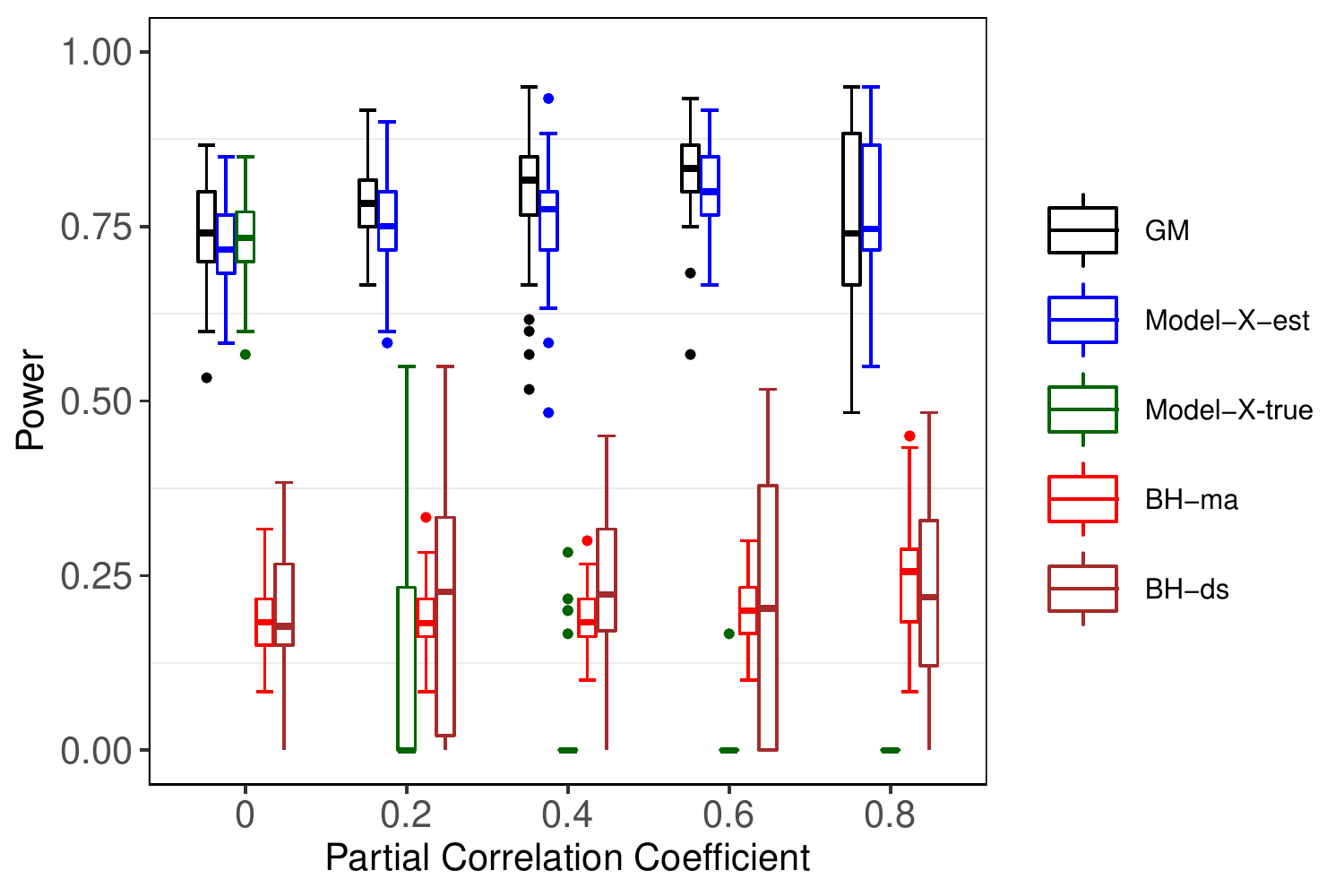}
\end{tabular}
 \caption{The plots of the FDPs and proportions of true rejections for the high-dimensional settings. Notations are the same as those in Figure~\ref{fig:ld}. Note that Model-X-fix, which fixed a previous numerical error of the knockoff package, was applied only to case (b) for $\rho\geq 0.6$.
 }\label{fig:hd}
\end{figure}
{\textbf{(ii) Constant positive correlation.}} As shown in Figure \ref{fig:hd} (b1), GM, BH-ds,  Model-X-true, and Model-X-fix all control their FDR properly around $0.1$, the designated level. Model-X-est is overly conservative with observed FDPs being close to zero when $\rho\geq 0.2$, whereas BH-ma fails to control the FDR especially when $\rho\geq 0.2$. %
When $\rho$ is larger than $0.5$, we apply the same algorithmic modification as in the low-dimensional setting for Model-X, denoted as Model-X-fix, which improves the power.
Figure \ref{fig:hd} (b2) shows proportions of true positives and the power. 
The power of Model-X-est knockoff  decreases rapidly as $\rho$ increases. Model-X-fix has a better power than Model-X-est. The power of BH-ds  is low and does not change much with respect to $\rho$. The powers of GM is stable throughout all $\rho$ values, decreasing  moderately  as $\rho$ increases. 

{\textbf{(iii) Constant partial correlation.}} As shown in Figure \ref{fig:hd} (c1), the FDRs of all the methods except BH-ma are controlled well.   
Both  GM and Model-X-est perform
better than all other methods, with GM having a slight edge.

\subsection{Simulation studies for non-Gaussian designs}

(i) \textbf{T-distribution.} We generate the features $\vx_j$ from a centered multivariate T-distribution with 3 degrees of freedom and the covariance matrix $\Sigma$, which is set to be autoregressive (Toeplitz), i.e. $\sigma_{ij} = \kappa^{|i-j|}$ for $\kappa = 0,0.2,\dots, 0.8$. As shown in Figure \ref{fig:ng} (a1,a2), the GM method controls the FDR at the designated level for all the settings and maintains the highest power among all the methods that control the FDR at the designated level. The FDR of Model-X knockoff is slightly inflated when the correlation is high. The BH method with data splitting shows a large variation in FDRs and has the lowest power. The marginal BH method can not control the FDR at the designated level.

(ii) \textbf{Bi-modal distribution.} The features are  generated from a Gaussian mixture distribution with two components: one centered at $-0.5 \mathbf{1}_p$ and the other at $0.5 \mathbf{1}_p$. The covariance matrices of the two components are the same and are autoregressive as in the previous case. As shown in Figure \ref{fig:ng} (b1,b2), GM  controls the FDR at the designated level for all the settings and has the highest power among all the methods that control the FDR at the designated level. Being conservative with a small FDR, the Model-X knockoff has a lower average power than the GM method. Additionally, the variation of the proportion of true rejections is much higher than that of the GM method. The BH method with data splitting shows a valid FDR control, but it has the lowest power among all these methods. The marginal BH method can not control FDR at the designated level.

\begin{figure}[hp]
\centering 
\begin{tabular}{cc}
    {\hspace{-30pt}\footnotesize (a1) T-distributed design }& {\hspace{-30pt}\footnotesize (a2) T-distributed design} \\
    \includegraphics[width=0.45\textwidth]{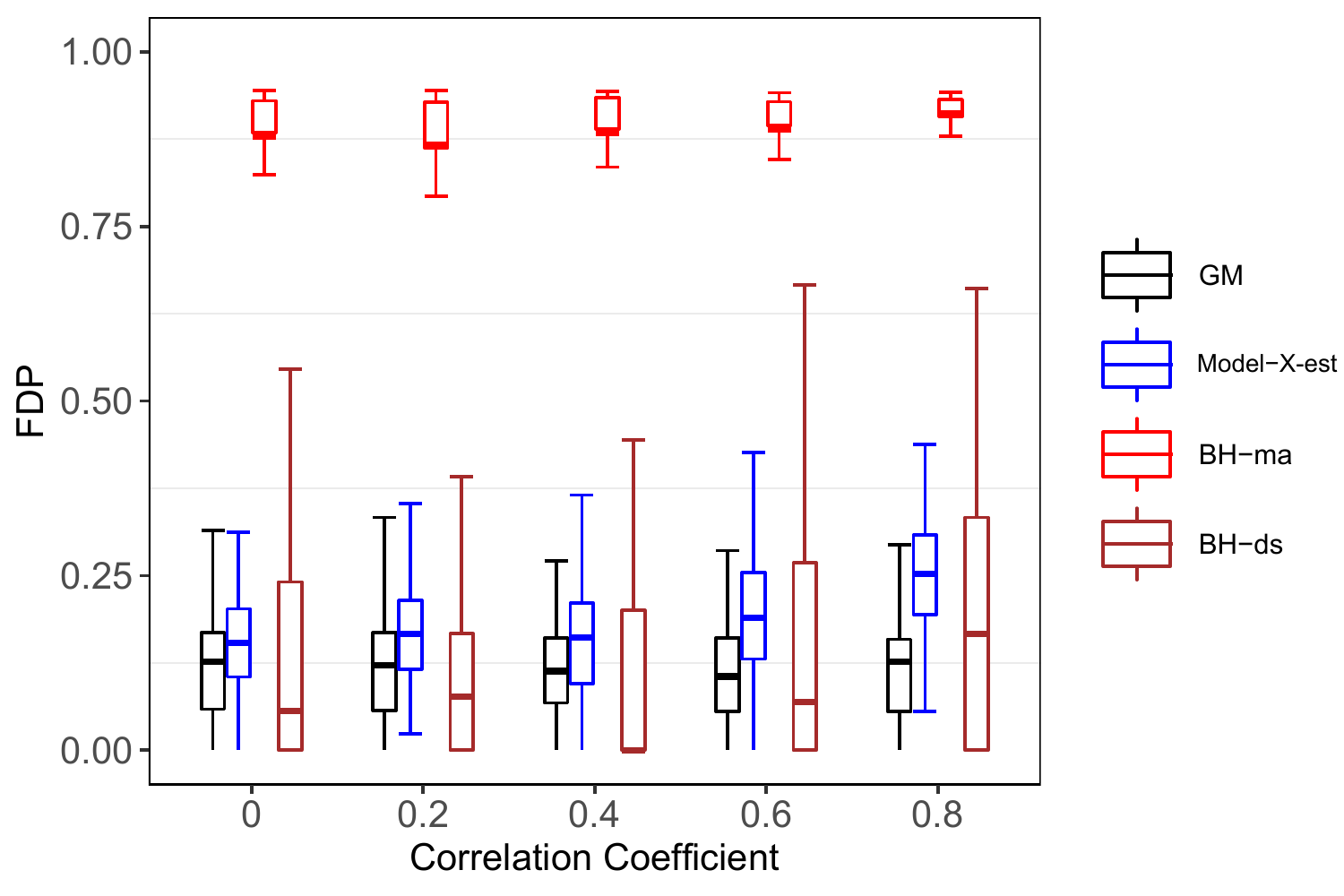}&\includegraphics[width=0.45\textwidth]{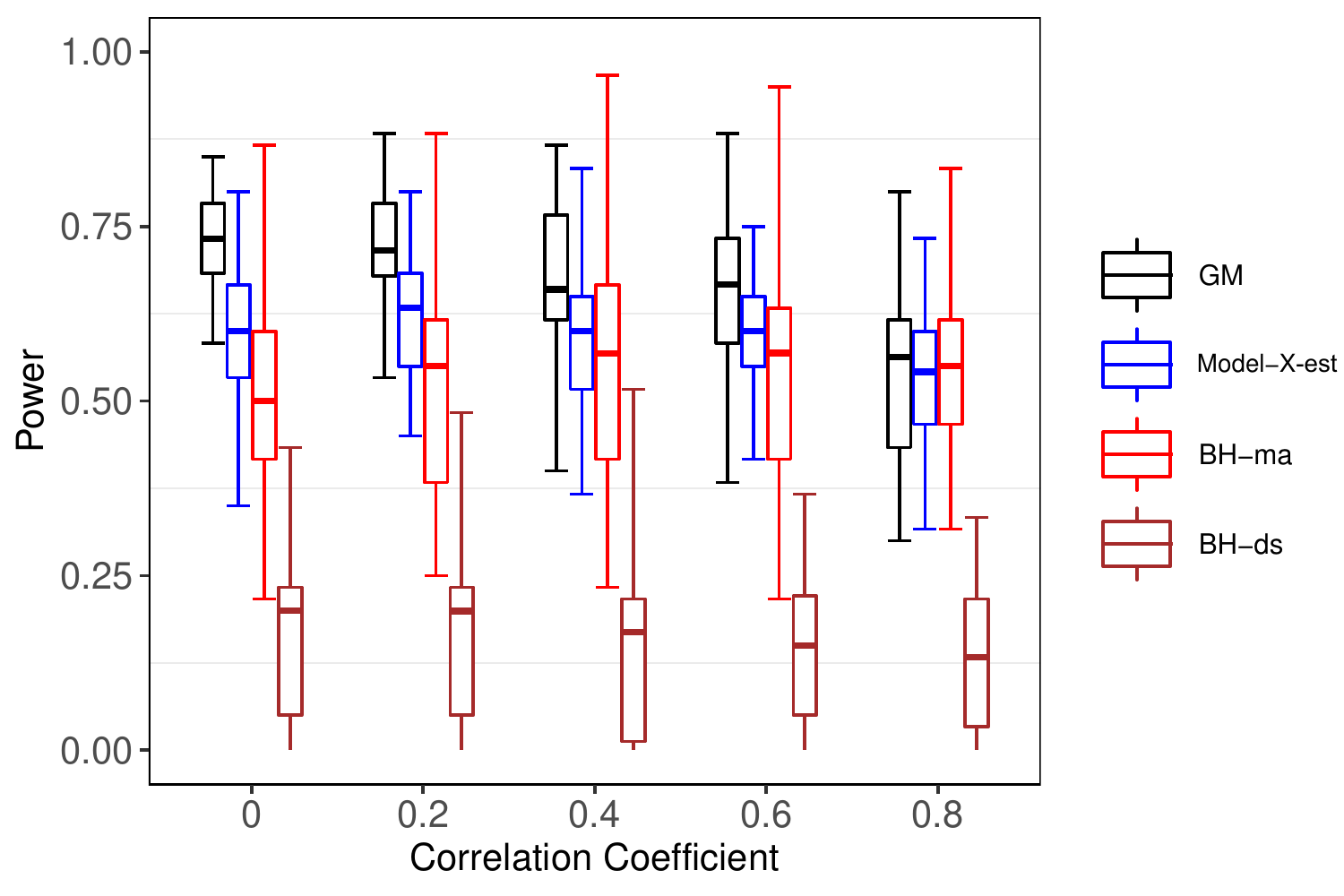} \\ 
    {\hspace{-30pt}\footnotesize (b1) Bimodal design }& {\hspace{-30pt}\footnotesize (b2) Bimodal design } \\
    \includegraphics[width=0.45\textwidth]{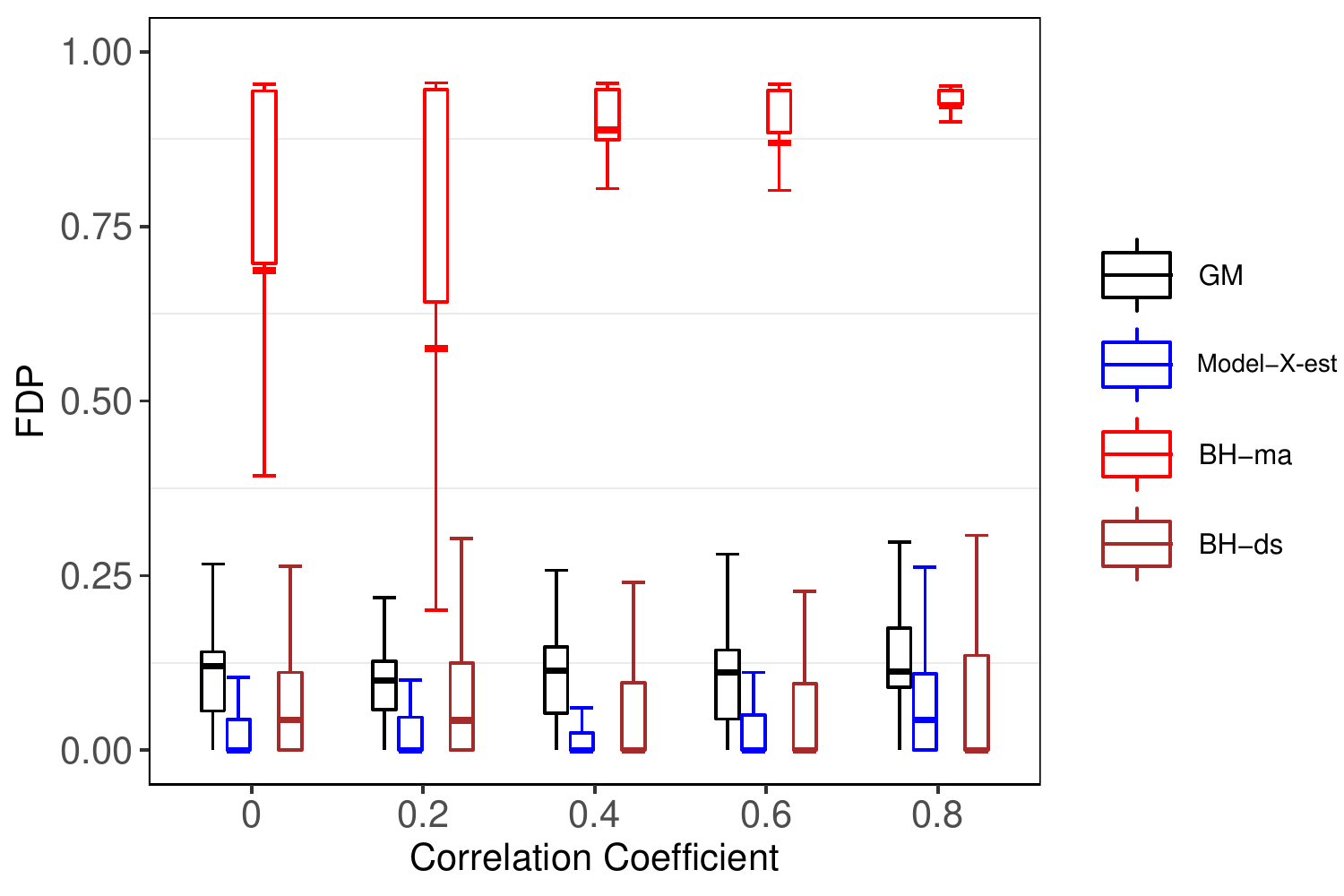}&\includegraphics[width=0.45\textwidth]{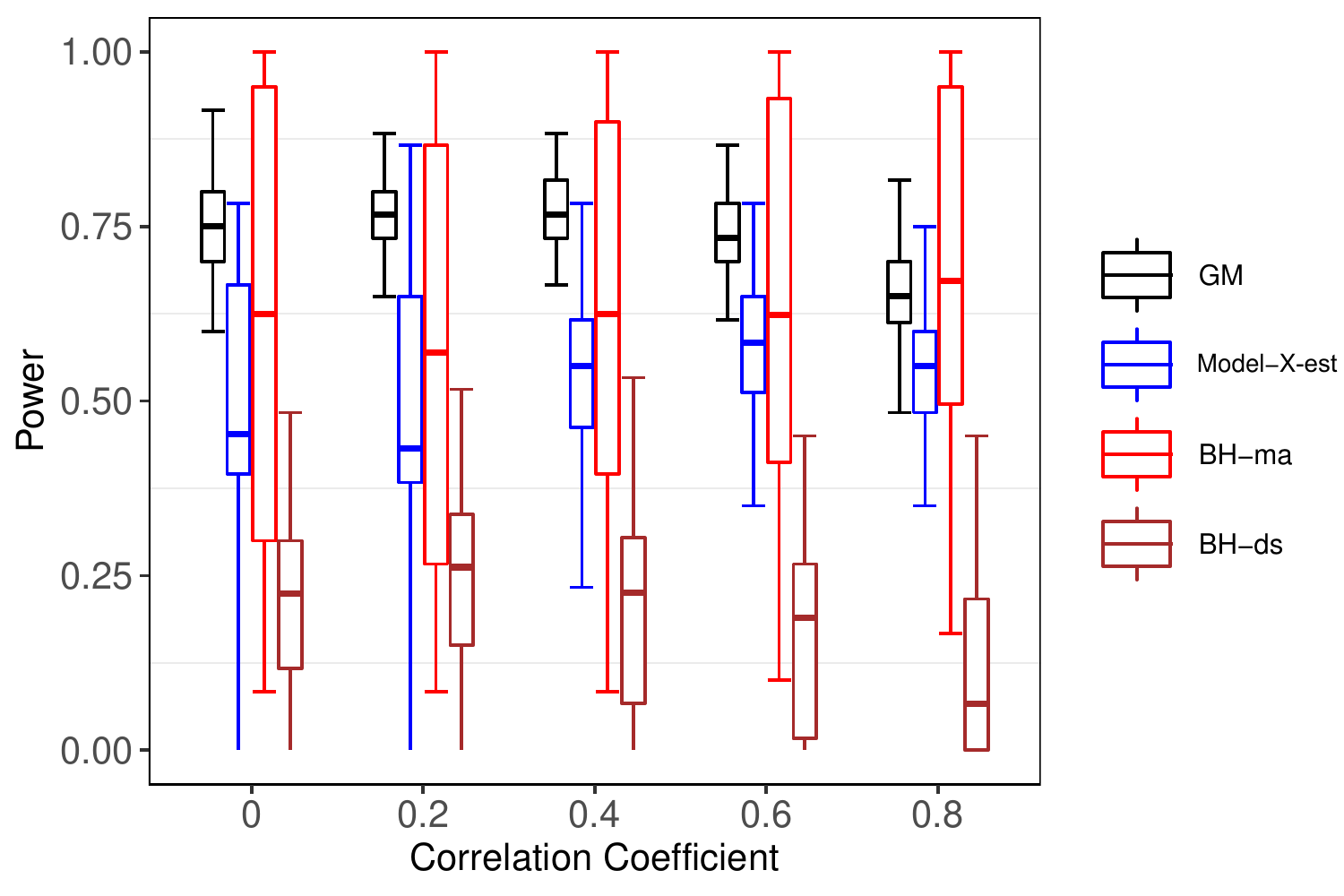} 
\end{tabular}
 \caption{The plots of the FDPs and the proportions of true rejections for non-Gaussian design settings. Notations are the same as those in Figure~\ref{fig:ld}.
 }\label{fig:ng}
\end{figure}

\subsection{Simulation studies based on a population genetics data set}
Genome wide association studies (GWAS) have become an attractive tool for genetic research. In these studies,  researchers examine a genome-wide set (tens of thousands to millions) of genetic variants of a group of individuals randomly selected from a target population to see if any variants are associated with a phenotype of interest.  Genetic variants are usually  in the form of  single nucleotide polymorphisms (SNPs), and are often used as covariates in a linear or logistic regression model, with a  main goal being to select relevant variants. Since most SNPs  take on only three values,  $\{0,1,2\}$ (representing 0, 1, or 2 minor allele mutations, respectively), the design matrix of such a SNP-based regression model is clearly non-Gaussian.

\begin{figure}[ht!]
\centering 
\begin{tabular}{ccc}
    \includegraphics[width=0.5\textwidth]{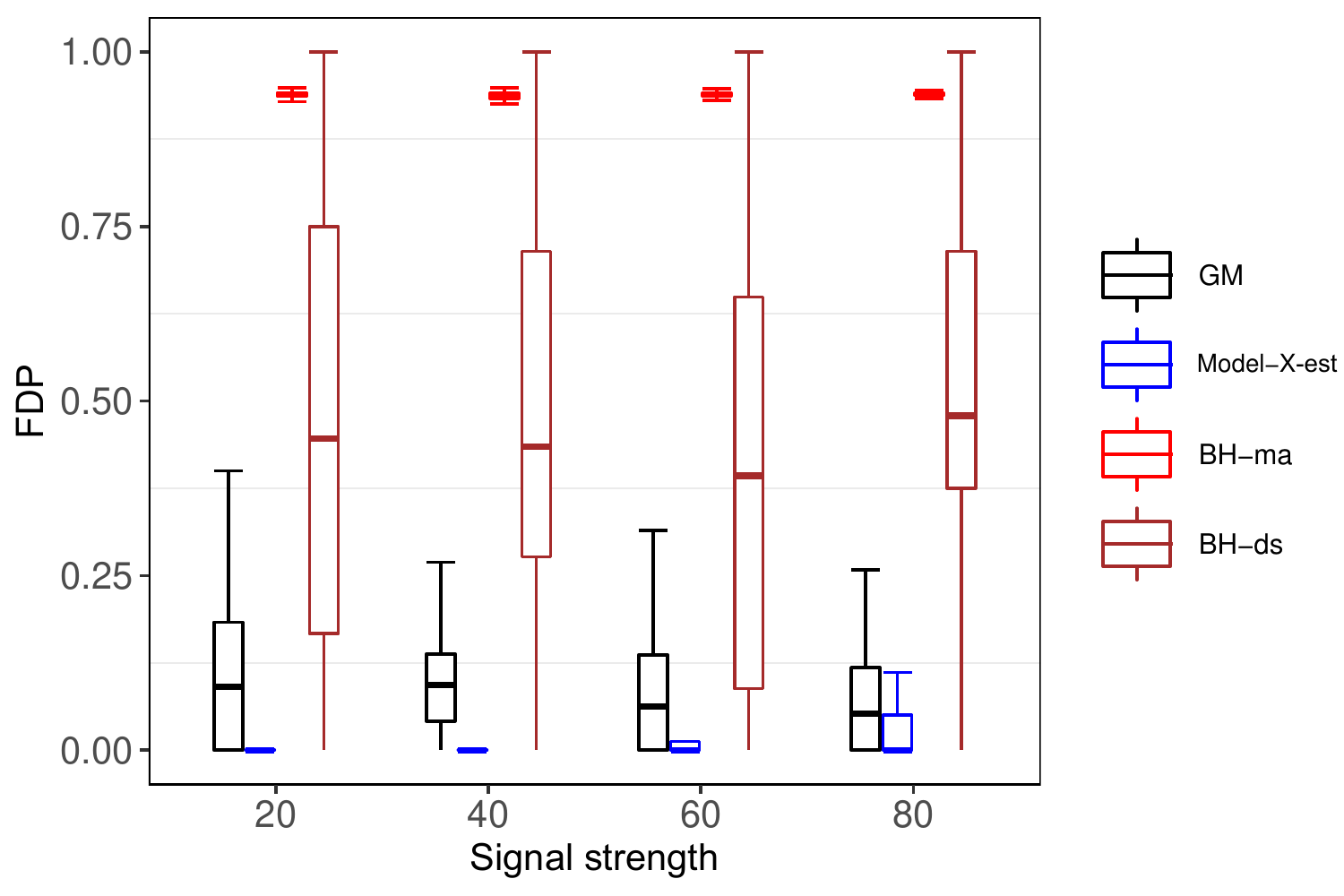}&\includegraphics[width=0.5\textwidth]{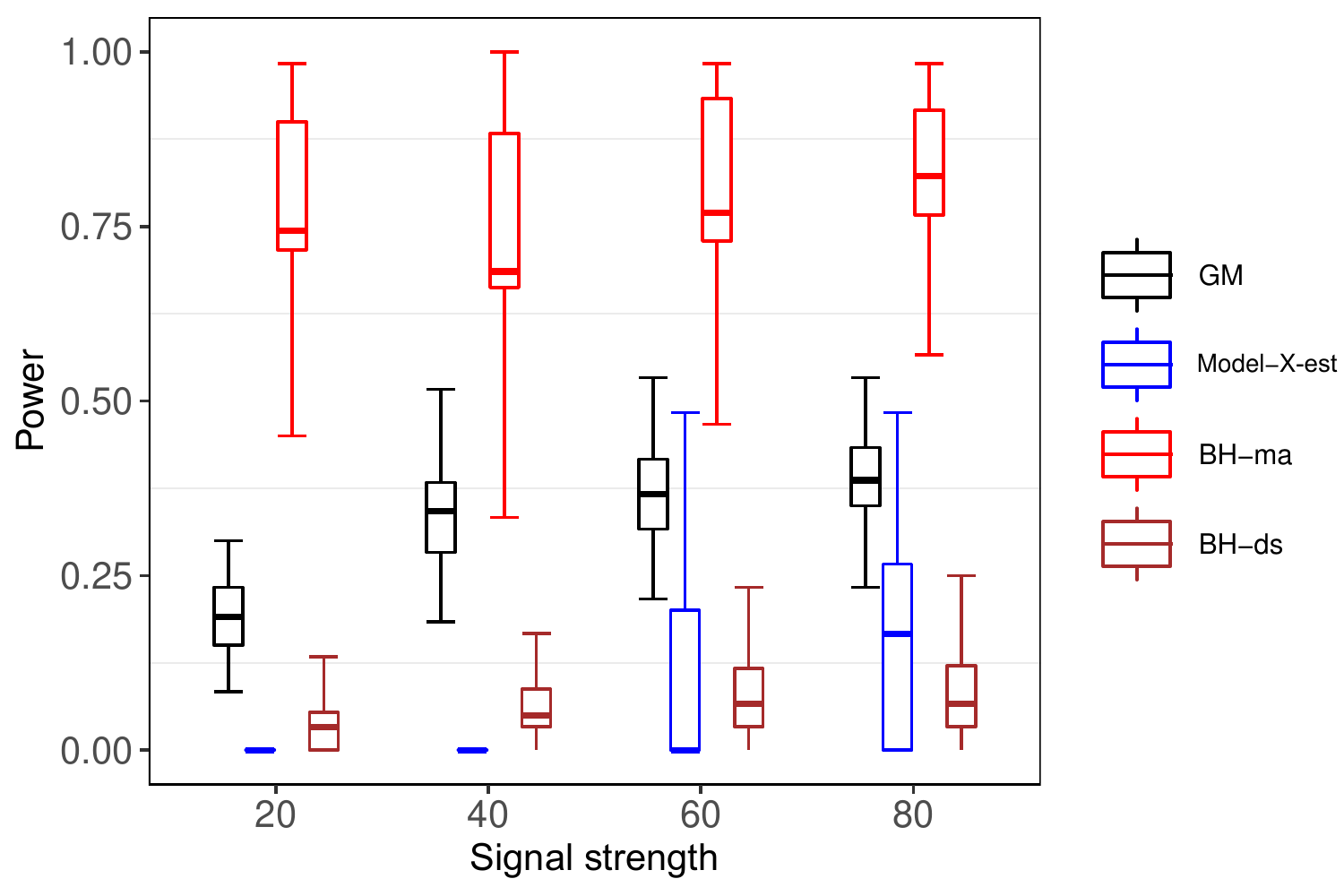}
\end{tabular}
 \caption{The plots of the FDPs and the proportions of true rejections of various methods for the GWAS-based design matrix. %
 }\label{fig:real}
\end{figure}

We consider a panel of 292 tomato accessions in \cite{bauchet2017use}, which is publicly available at \url{ftp://ftp.solgenomics.net/manuscripts/
Bauchet_2016/}  and includes breeding materials (specimens) characterized by $>11,000$ SNPs. 
 Here we are interested in examining the FDR and power of GM, BH, and Model-X knockoff  by using this real-data set to create realistic design matrices. Specifically, we randomly select $1000$ SNPs as $\vX$ and randomly generate  $60$ nonzero regression coefficients from  $N(0, c^2/n)$, where $c$ ranges from $20$ to $80$ representing signal strength.  The response variable $\vy$ is generated from Eq (\ref{eqn:lm}) with a standard Gaussian noise. We set the target FDR level as $10\%$ and calculate the FDR and power  based on $100$ replications.

As shown in Figure \ref{fig:real}, we observe that GM controls the FDR at the designated level properly. Model-X is very conservative with its FDR very close to zero. Consequently, the power of  Model-X is much lower than that of  GM. The BH methods (both the marginal regression and the data splitting version) fail to control the FDR  because of high correlations among the randomly selected SNPs. As a consequence, the high power of BH-ma is not scientifically meaningful in this case.

\subsection{Empirical results on estimating the expected number of FDs }
We evaluate the performance of the GM method for estimating 
the expected number of FDs in a top-$k$ list, denoted as $\EE[FD(k)]$,
and the coverage frequency of its bootstrap confidence interval  proposed in Section~\ref{sec:var:fdp}. 
We consider both low-dimensional  $(n=1000, p=300)$ and high-dimensional settings $(n=300, p=1000)$, in which we randomly generate $60$ nonzero coefficients independently from $N(0,(20/\sqrt{n})^2)$.
We consider three types of design matrices similar to those described in Sections \ref{sec:sim:low:p} and \ref{sec:sim:high:p}: autoregressive (Toeplitz), constant correlation, and constant partial correlation. The response $\by$ is generated according to Eq (\ref{eqn:lm}) with $\sigma=1$.  We fix $\kappa=\rho=\tau=0.2$ for these settings, and repeat $100$ times for each setting.  

In each replication, we obtain an estimate of $\EE[FD(k)]$ with $k\in(50, 70)$. For the $r$th replication ($1\leq r \leq 100$), we calculate  $\widehat{FD}^{(r)}(k)$  following (\ref{eq:est:number:fdr}), and 
record the underlying true number of false discoveries as  $FD^{(r)}(k)$.
We use the sample average  $\tilde{\EE}[FD(k)]:= \frac{1}{100}\sum_{r=1}^{100} FD^{(r)}(k)$ as an approximation to $\EE[FD(k)]$ (note that each $FD^{(r)}(k)$ is the true number of FDs for replication $r$ and is only available in simulations), and use
$\tilde{\EE}[\widehat{FD}(k)] := \frac{1}{100}\sum_{r=1}^{100} \widehat{FD}^{(r)}(k)$ as an approximation of $\EE[\widehat{FD}(k)]$, which should ideally track the value of $\EE[FD(k)]$.

\begin{figure}[hp!]
\centering 
\begin{tabular}{cc}
   {\hspace{-30pt}\footnotesize (a1) Power decay auto-correlation }& {\hspace{-30pt}\footnotesize (a2) Power decay auto-correlation} \\
    \includegraphics[width=0.5\textwidth]{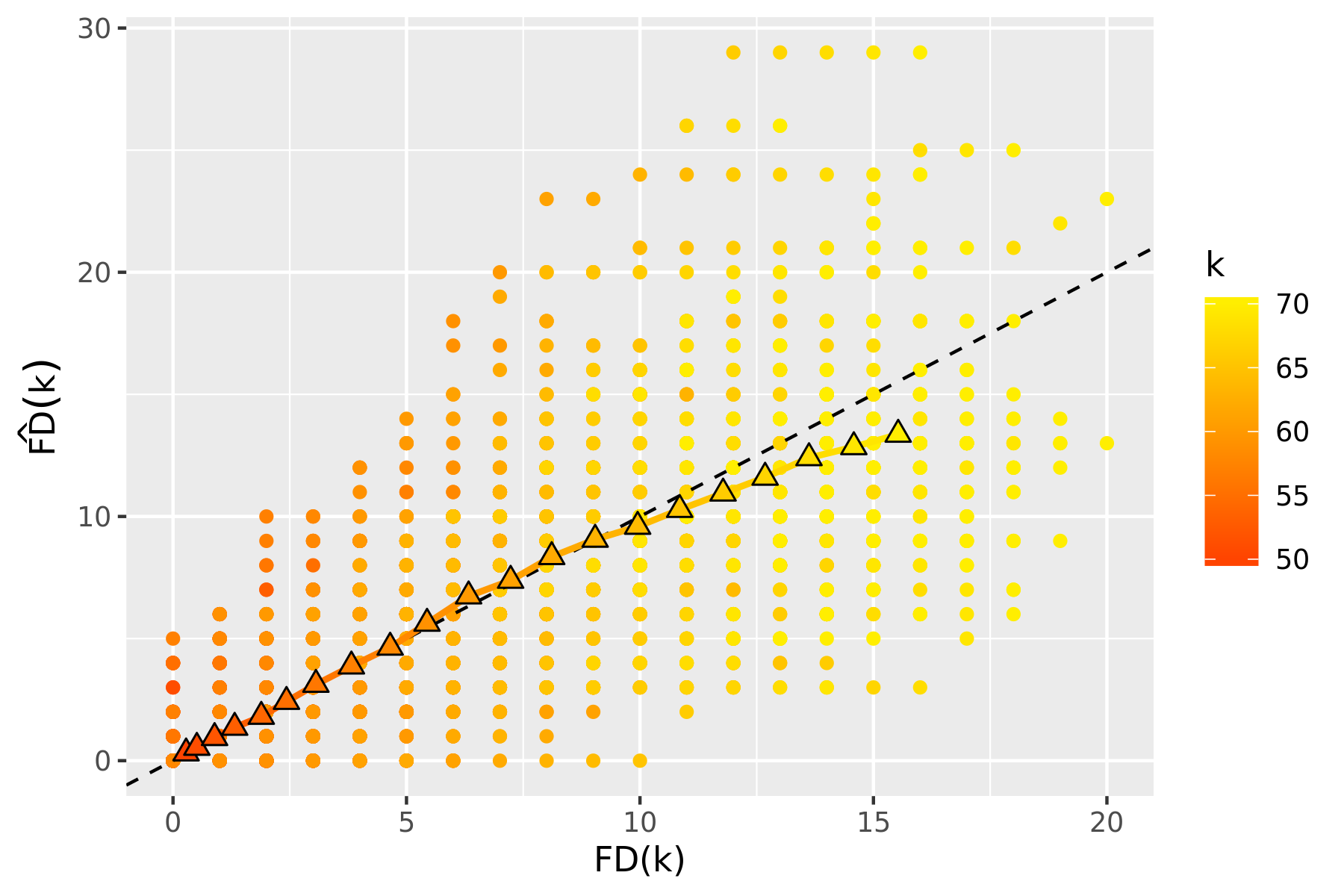}&\includegraphics[width=0.5\textwidth]{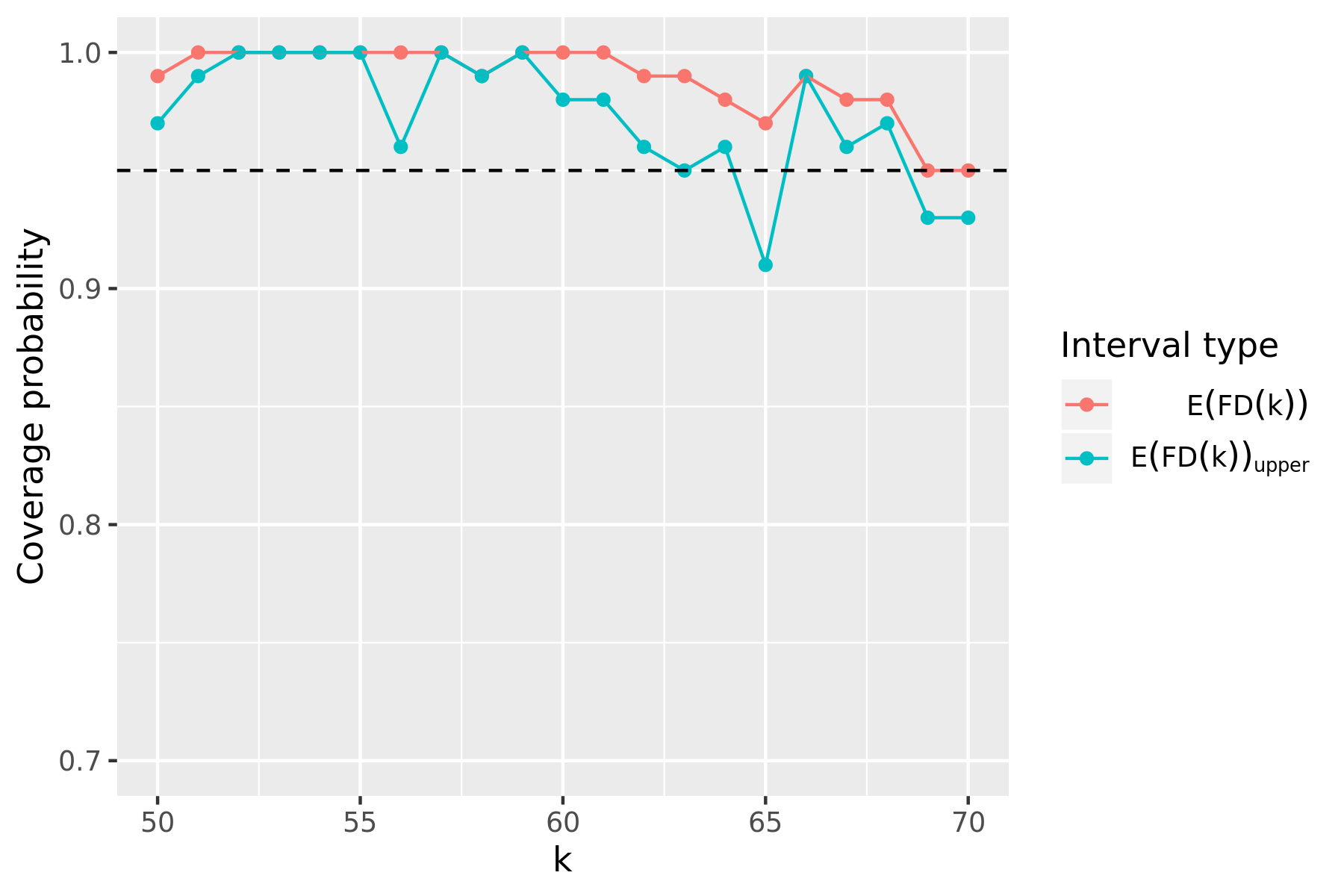} \\ 
    {\hspace{-30pt}\footnotesize (b1) Constant positive correlation }& {\hspace{-30pt}\footnotesize (b2) Constant positive correlation } \\
    \includegraphics[width=0.5\textwidth]{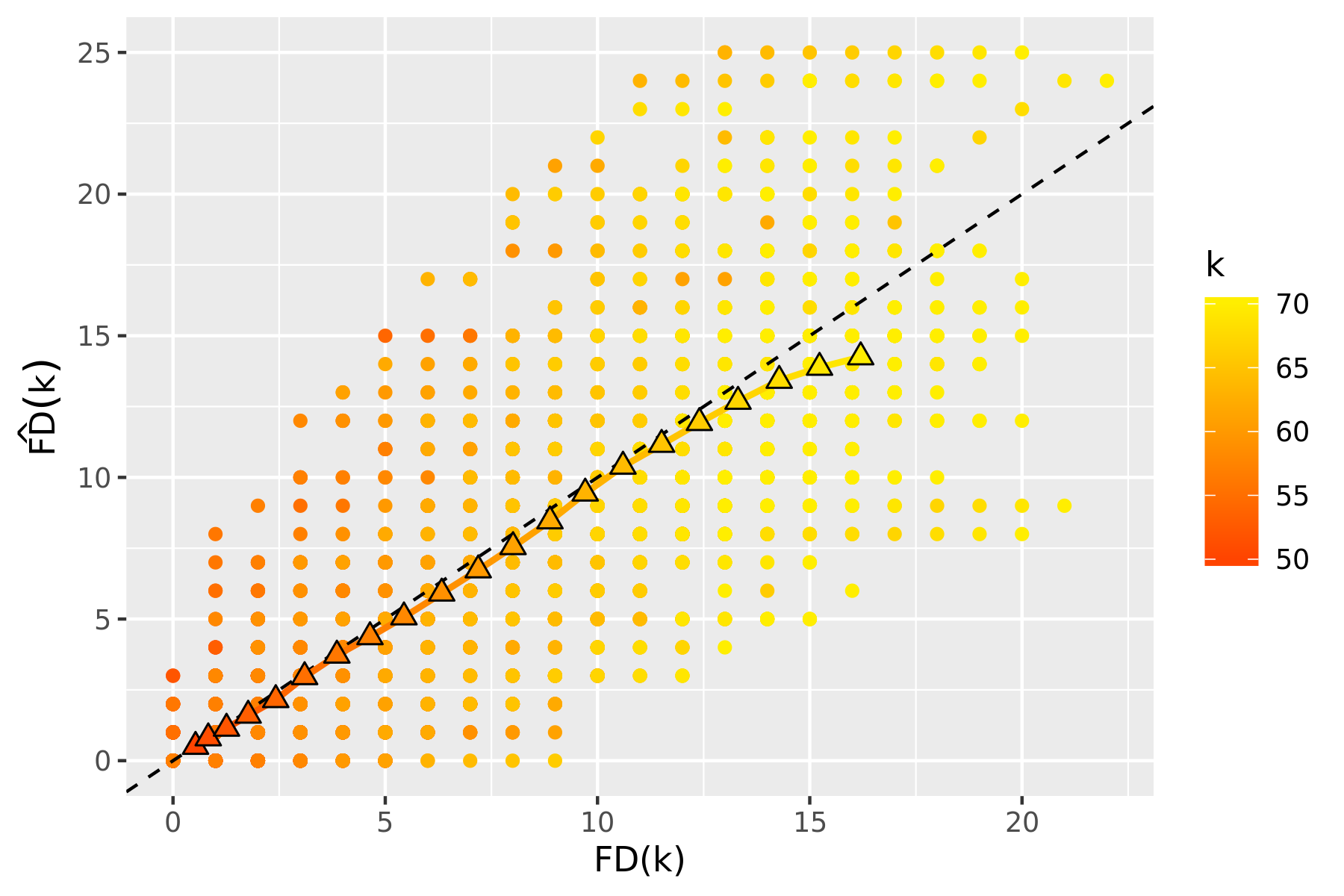}&\includegraphics[width=0.5\textwidth]{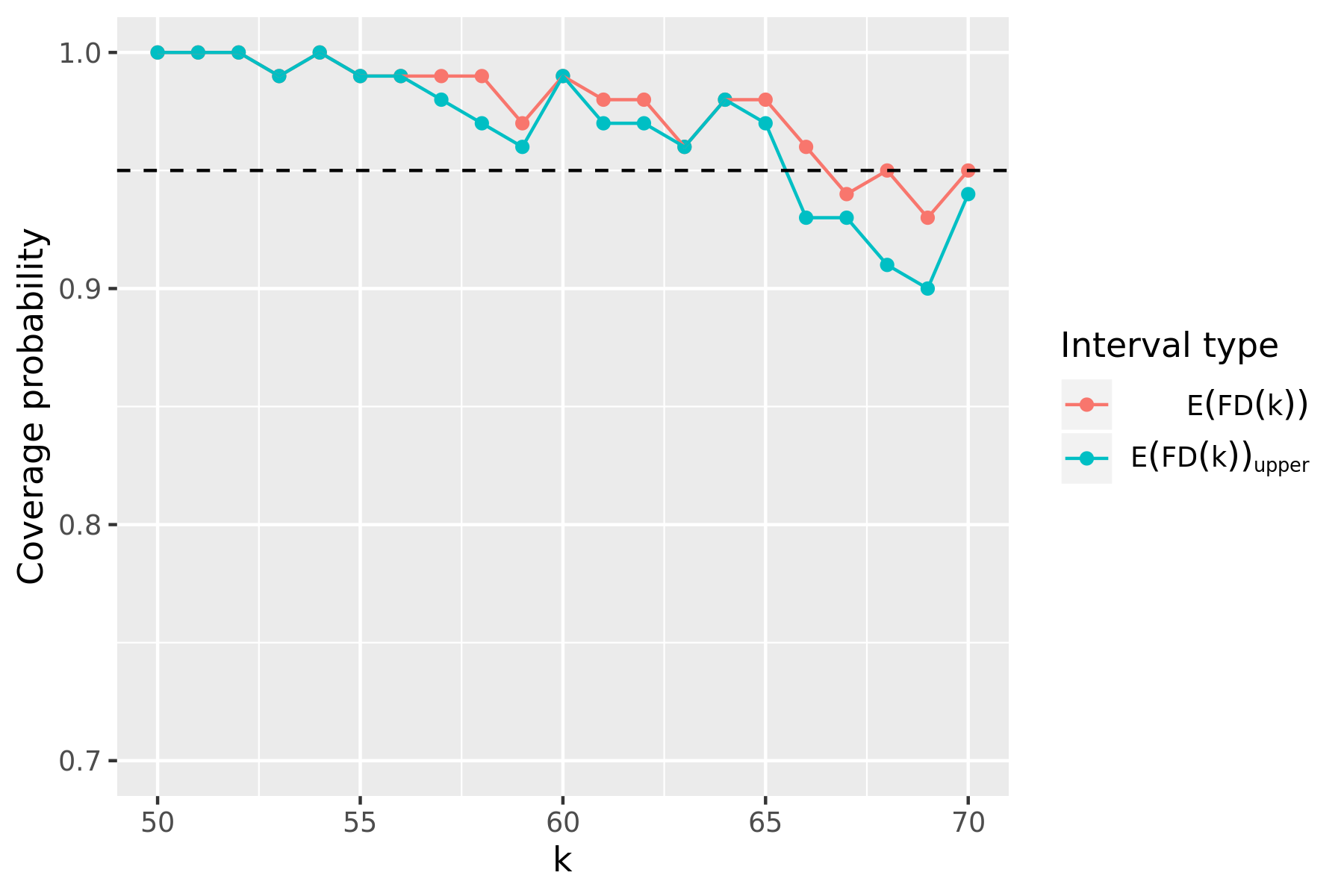}\\
    {\hspace{-30pt}\footnotesize (c1) Constant partial correlation  }& {\hspace{-30pt}\footnotesize (c2) Constant partial correlation } \\
    \includegraphics[width=0.5\textwidth]{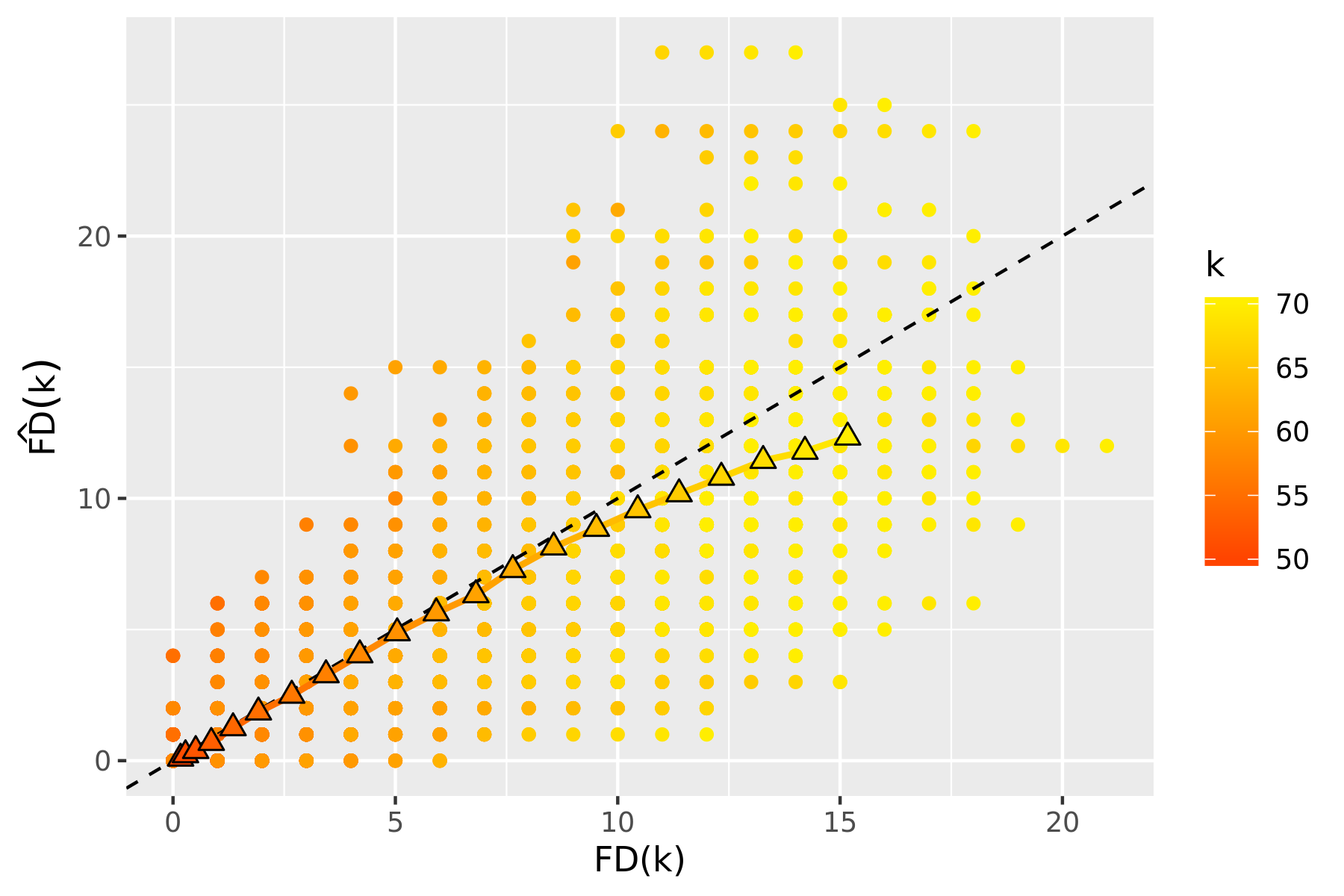}&\includegraphics[width=0.5\textwidth]{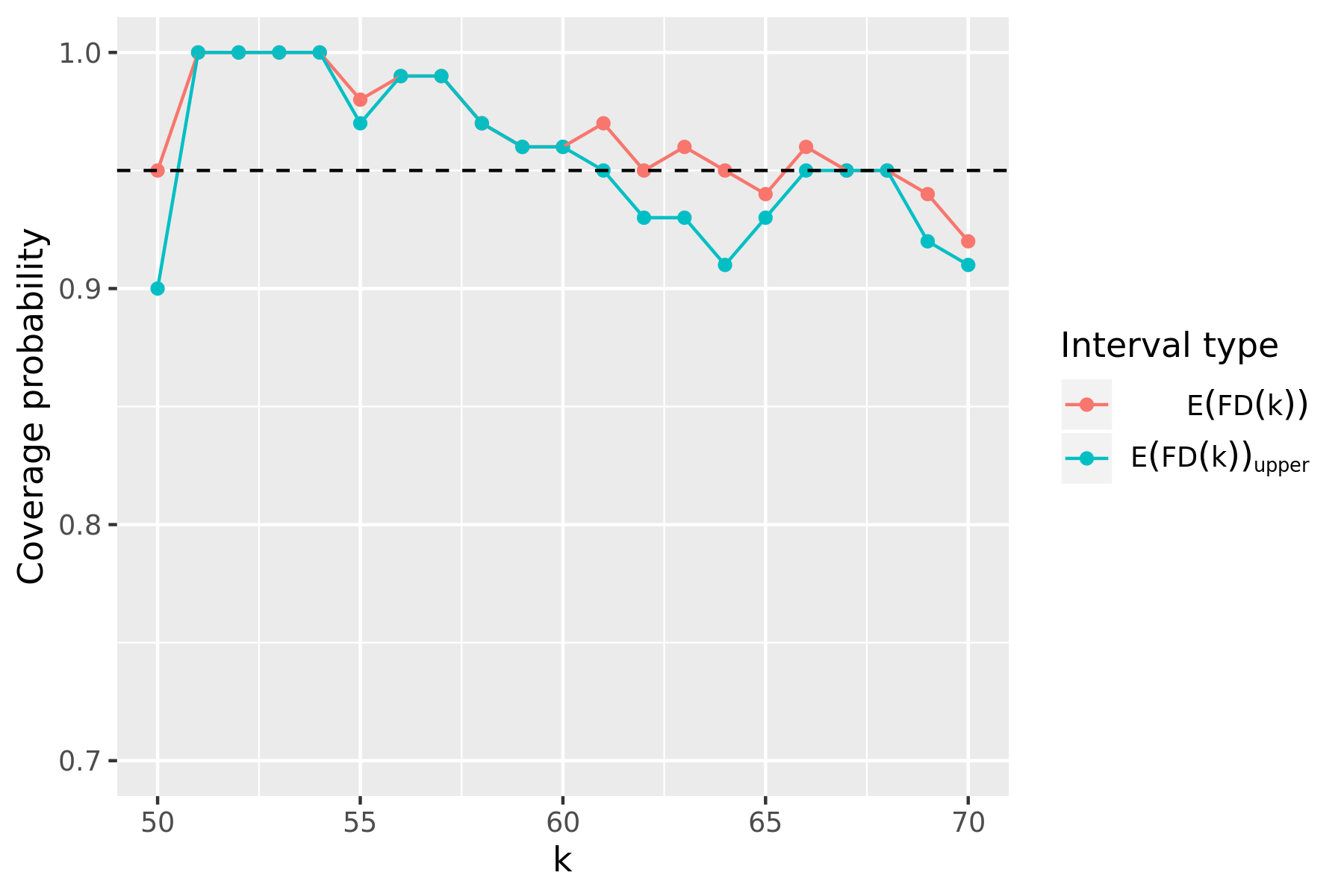}
\end{tabular}
 \caption{The low-dimensional case with $n$=1000, and $p$=300. Left panels: scatter plots of  $\widehat{FD}(k)$ versus $FD(k)$ for $k\in(50,70)$, under three correlation structures of the design matrix, respectively. The triangles are $\tilde{\EE}[FD(k)]$ versus $\tilde{\EE}[\widehat{FD}(k)]$ in $100$ replications, and the dashed line is $x=y$. Right panels:  the coverage frequencies of the  $95\%$ bootstrap confidence interval and confidence upper bound of  $\EE[FD(k)]$ (approximated by $\tilde{\EE}[FD(k)]$), respectively, for $k$ ranging from $50$ to $70$. The color shading represents different values of $k$.}\label{fig:bootld}
\end{figure}

\begin{figure}[hp!]
\centering 
\begin{tabular}{cc}
   {\hspace{-30pt}\footnotesize (a1) Power decay auto-correlation }& {\hspace{-30pt}\footnotesize (a2) Power decay auto-correlation} \\
    \includegraphics[width=0.5\textwidth]{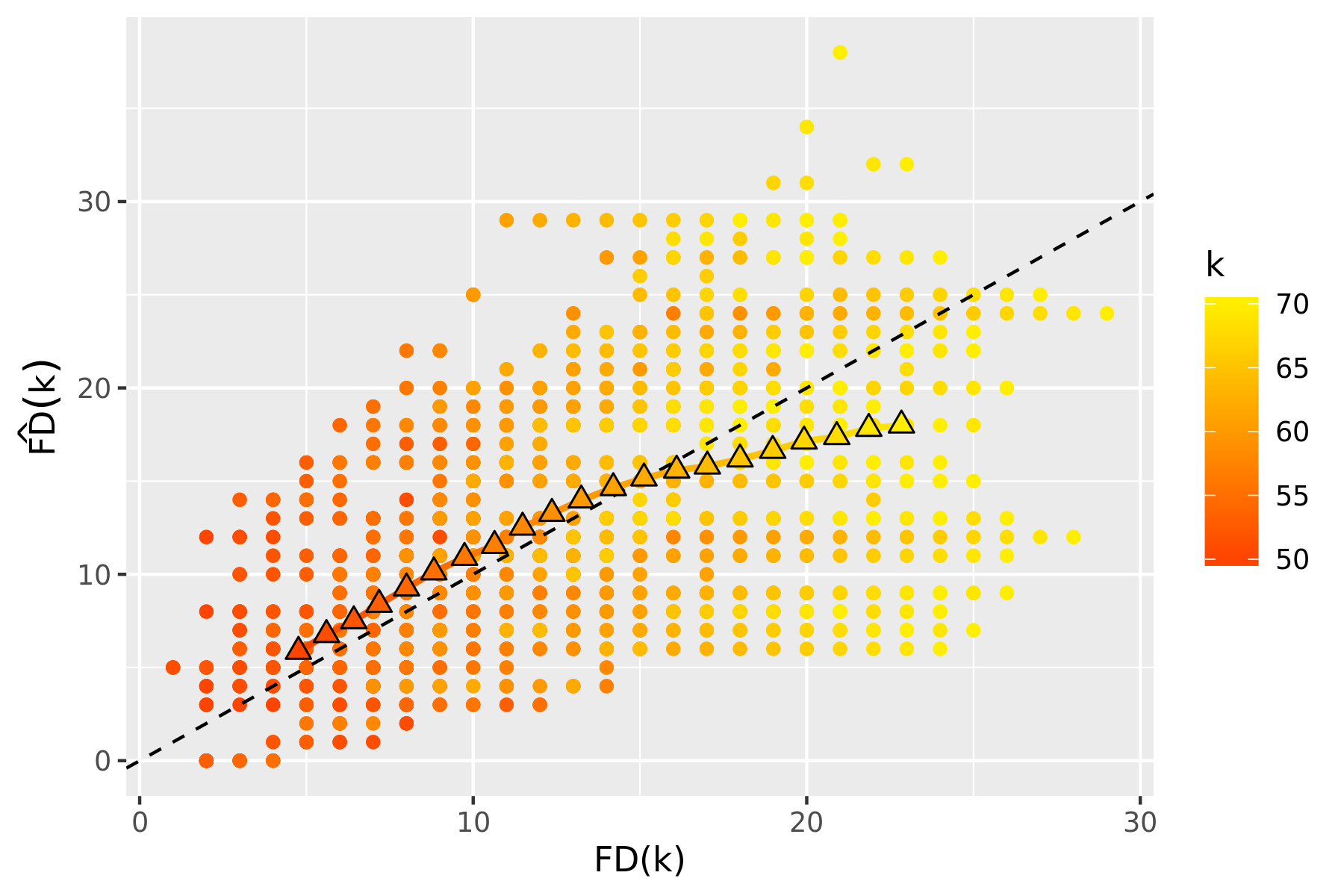}&\includegraphics[width=0.5\textwidth]{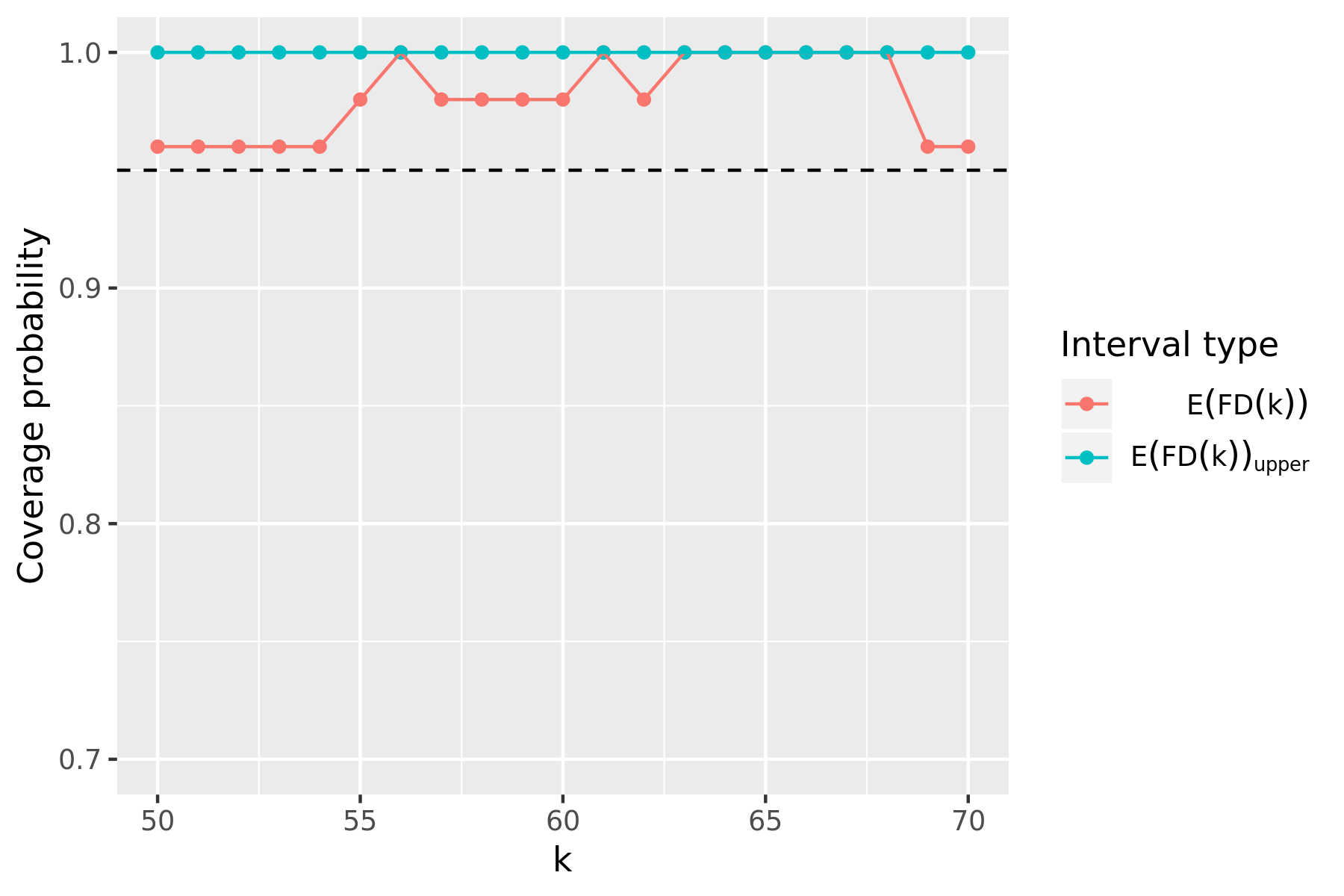} \\ 
    {\hspace{-30pt}\footnotesize (b1) Constant positive correlation }& {\hspace{-30pt}\footnotesize (b2) Constant positive correlation } \\
    \includegraphics[width=0.5\textwidth]{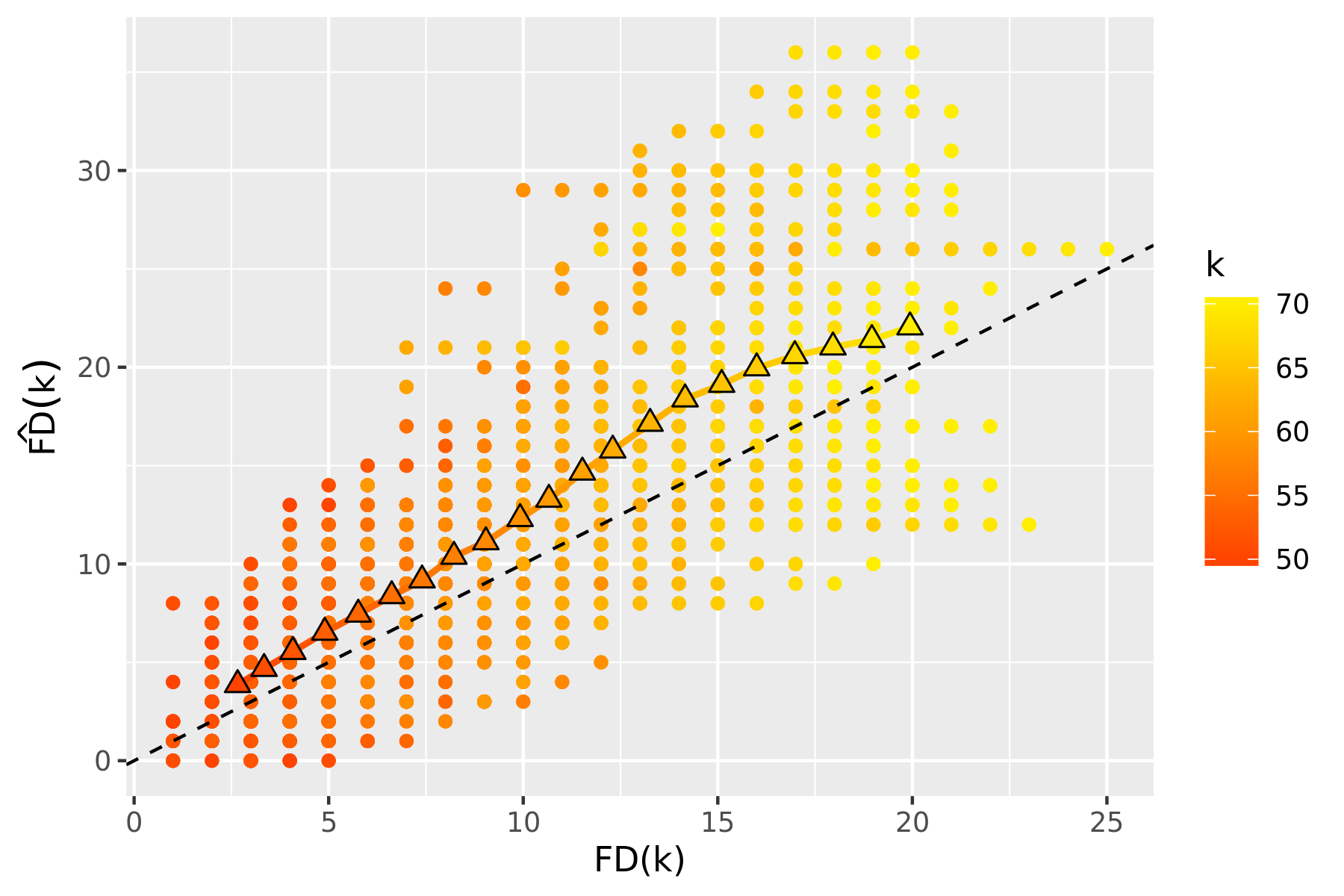}&\includegraphics[width=0.5\textwidth]{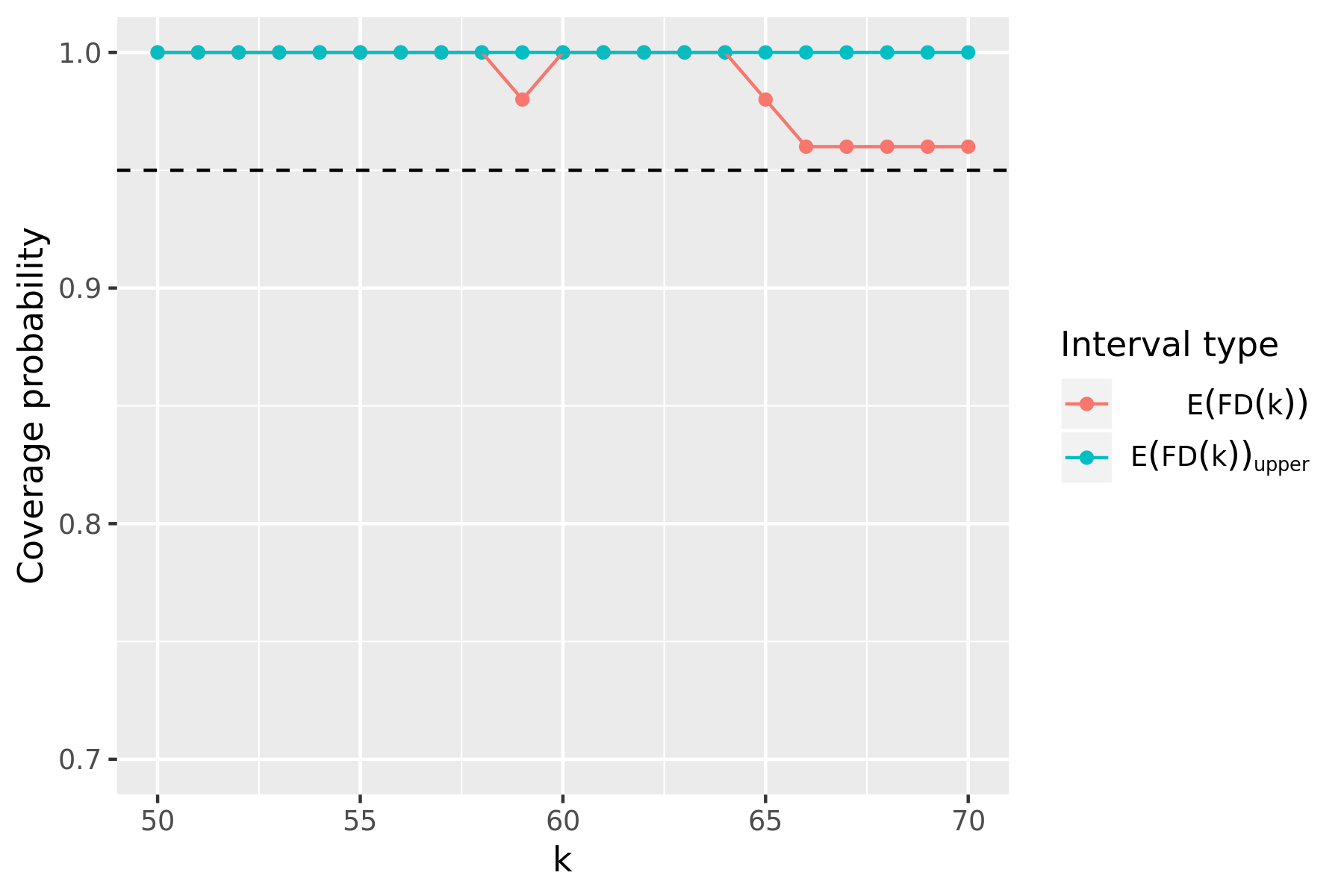}\\
    {\hspace{-30pt}\footnotesize (c1) Constant partial correlation  }& {\hspace{-30pt}\footnotesize (c2) Constant partial correlation } \\
    \includegraphics[width=0.5\textwidth]{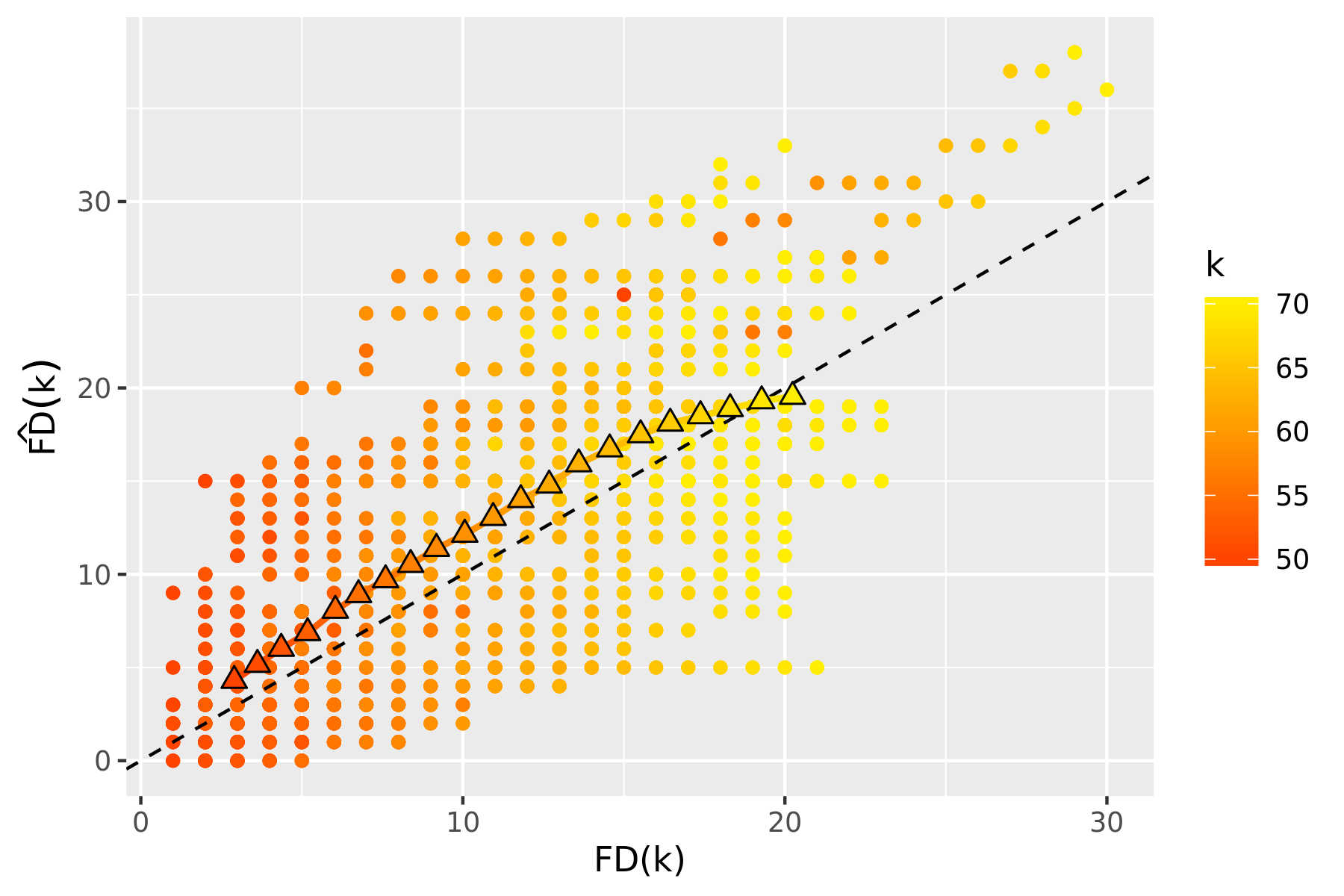}&\includegraphics[width=0.5\textwidth]{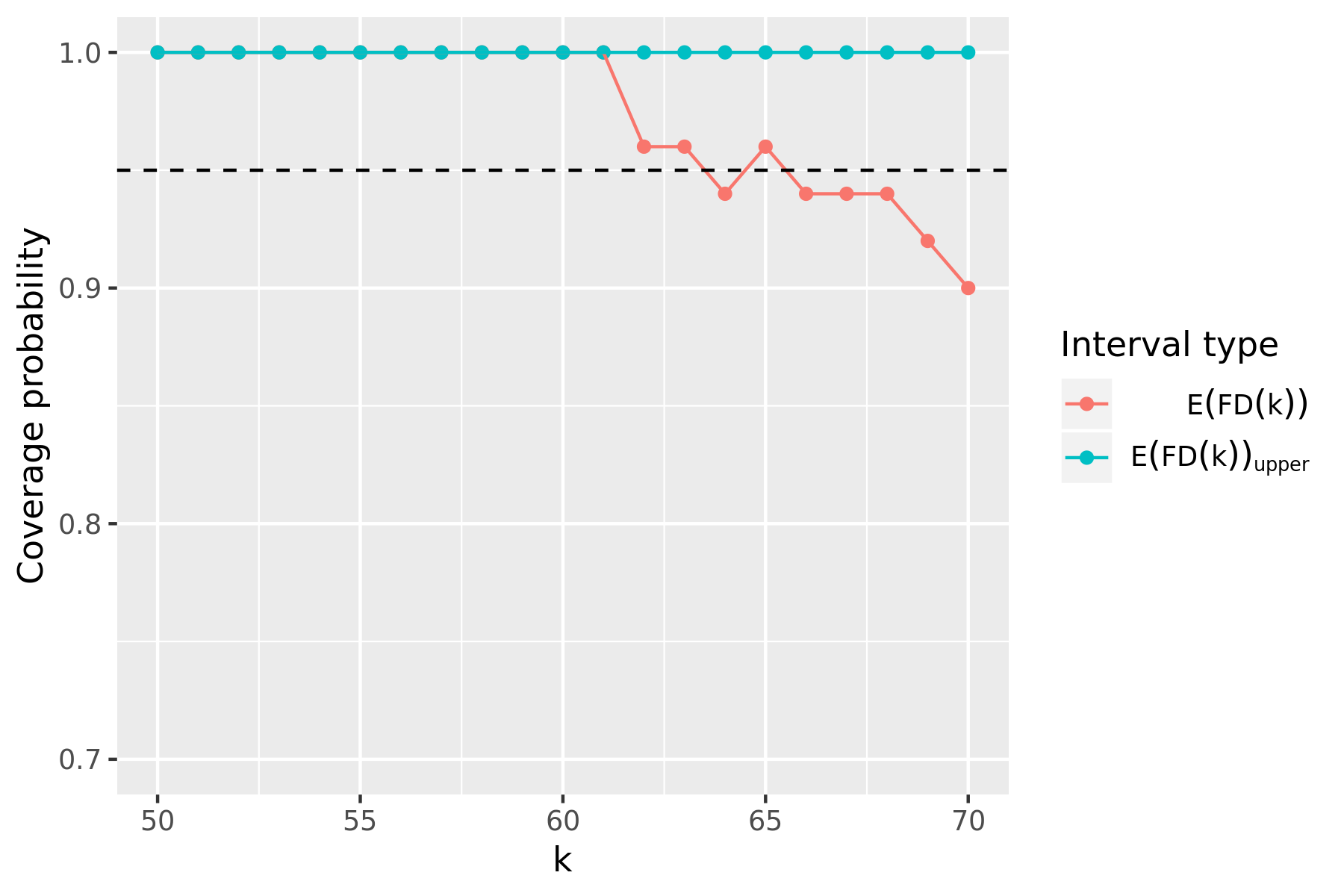}
\end{tabular}
     \caption{The high-dimensional settings with (n=300, p=1000). Details and notations are the same as those in Figure~\ref{fig:bootld}.}
 \label{fig:boothd}
\end{figure}

By using Algorithm \ref{alg:bootstrap} with $B=200$, we calculate the empirical coverage frequency of the proposed bootstrap confidence interval for $\EE[FD(k)]$ as
\begin{equation*}
\frac{1}{100}\#\{r: \tilde{\EE}[FD(k)] \in [ \widehat{FD}_{(\alpha/2)}(k),\widehat{FD}_{(1-\alpha/2)}(k)]\}
\end{equation*}
where %
$\widehat{FD}_{(\alpha/2)}^{(r)}(k)$ and $\widehat{FD}_{(1-\alpha/2)}^{(r)}(k)$ are  the $\alpha/2$ and $(1-\alpha/2)$  quantiles of the  bootstrap distribution of $\widehat{FD}(k)$ in the $r$th replication.  In addition, we evaluate the empirical coverage probability of the bootstrap $(1-\alpha)$100\% upper bound
\begin{equation*}
\frac{1}{100}\#\{r: \tilde{\EE}[FD(k)] \in [ 0 ,\widehat{FD}_{(1-\alpha)}(k)]\}.
\end{equation*}
We set  $\alpha=0.05$ for all simulation cases.

In Figures \ref{fig:bootld} (a1,b1,c1), we plot $\widehat{FD}(k)$ against $FD(k)$, the true number of false discoveries with $k$ ranging from 50 to 70. 
The closer the point approaches the diagonal line, the closer the two numbers are. As $k$ increases, the estimate tends to be more dispersed as anticipated since the error between  $\widehat{FD}(k)$ and $FD(k)$ increases as $k$ increases.
In fact, $FD(k)$ is also more variable around its mean $\EE(FD(k))$ as $k$ increases.  Figure \ref{fig:bootld} (a2,b2,c2) report the coverage probabilities of the confidence interval and upper confidence interval for  $\EE[FD(k)]$ in all the settings. The coverage probabilities are above $0.95$ when $k$ is smaller than $60$. When $k$ is larger than $60$, the coverage probabilities drop occasionally below $0.95$ by a small amount. This is expected since the variation of $\widehat{FD}(k)$ increases as $k$ increases.%

In Figure \ref{fig:boothd}, we show the simulation results for the high-dimensional settings. For the Toeplitz structure, the scatter plot in Figure \ref{fig:boothd} (a1)  shows that our proposed estimate tends to underestimate its target slightly when $k$ is large. For the constant correlation structure and constant partial correlation structure (Figure \ref{fig:boothd} (b1, c1)), however, the proposed estimate slightly overestimate its target . Similar to the low-dimensional case, the points become more dispersed from the diagonal line as $k$ increases.  In Figures \ref{fig:boothd} (a2,b2,c2), we show the coverage of the two-sided confidence  interval and confidence upper bound for $\EE[{FD}(k)]$. The coverage probabilities of the confidence upper bound of $\EE[FD(k)]$ were over $0.95$ for all cases, and that for confidence interval are also mostly above 0.95  except for the constant partial correlation case with $k>65$. This slight under-coverage appears to be a consequence of over-estimation of the number of FDs.

\section{Discussion}\label{sec:disc}
We introduce the GM method for controlling the FDR in high-dimensional  regression models. Intuitively, to test if a variable $X_j$ is truly informative,  GM constructs a pair of variables mirroring each other at $X_j$. The scale of the mirror can be computed explicitly and easily as a function of the design matrix so as to guarantee that the distribution of the mirror statistics is symmetric about zero under the null. 
With this construct, we developed a testing method such that the FDR is controlled at any designated level asymptotically. We have further proposed a way to assess the variance of the number of estimated FDs in a top-$k$ list, 
providing extra information about the reliability of the reported FDP results. %
Through empirical studies we find that  GM is not sensitive to the scale of the mirror, making the method more broadly applicable.

A distinctive advantage of GM is that it can be applied to both low or high-dimensional regression problems without requiring any distributional assumption or knowledge on the design matrix, which is common in many applications. 
Furthermore,  since each Gaussian mirror is constructed marginally for each covariate, it brings two advantages: (i) 
the mirror construction introduces only a small and controllable additive ``disturbance", which not only alleviates possible high correlations among the original covariates, but also gains in power; (ii) the computation in GM can be  parallelized, and the method is easily scalable to handle  large datasets with the deployment of a GPU-based computational infrastructure.
The GM design introduced in this article  is a general framework based on the idea of marginally perturbing the predictor variables one at a time. 
Such a construct can be generalized to handle more complex statistical and machine-learning models such as generalized linear models, index models, additive models, neural networks, etc. In linear models, we choose the scale of the mirror $c_j$ to annihilate the partial correlation between the mirrored variables $X_j^+$ and $X_j^-$ so that the mirror statistic is symmetric under the null. 
A natural generalization for more complex models is to choose $c_j$ to minimize a dependence measure of $X_j^+$ and $X_j^-$ conditional on the remaining variables. Our preliminary results show that a modified GM method  works well for selecting important variables in  neural network models (see Appendix for a detailed example).

\section{Appendix}\label{sec:appendix}
\subsection{Proofs}

{\bf Proof of Lemma \ref{lemma:symmetric}.}
Note that $(U+V, U-V)$ also follows a bivariate normal distribution. Since the correlation between $U$ and $V$ is zero, $\var(U + V) = \var(U-V)$. Furthermore, $E(U+V)=E(U-V)=0$. Therefore, the joint distribution of $(U+V,U-V)$ is the same as the joint distribution of $(U-V, U+V)$. As a result,
\[
P(M > t) = P(|U+V|-|U-V| >t ) = P(|U-V|-|U+V|>t) = P(M<-t).
\]
\qed

{\noindent\bf Proof of Lemma \ref{lemma:olscor}.} Note that $\hat{\bbeta}^j$ is an OLS estimator based on $\vX^j$, then
\begin{multline}\label{eq:cov}
\frac{1}{n}\cov (\hat{\bbeta}^j\,|\, \vX^j) = \\ \sigma^2 \begin{bmatrix}
    1+\vzj + 2\rhoxjzj  & 1-\vzj & (\rhoxjxnj + \rhozjxnj)^\top \\
    1- \vzj & 1+\vzj - 2\rhoxjzj & (\rhoxjxnj - \rhozjxnj)^\top \\
    \rhoxjxnj + \rhozjxnj & \rhoxjxnj - \rhozjxnj & \sigmaxnj
  \end{bmatrix}^{-1},
\end{multline}
where  $\vzj =\frac{1}{n} c_j^2 \vz_j^\top \vz_j$, $\rhoxjzj = \frac{1}{n}c_j \vx_j^\top \vz_j$, $(\rho^{(\vx_j,\vX_{-j})})^\top = \frac{1}{n}\vx_j^\top \vX_{-j}$, $(\rhozjxnj)^\top = \frac{1}{n}c_j \vz_j^\top \vX_{-j}$, and $\sigmaxnj = \frac{1}{n} \vX_{-j}^\top \vX_{-j}$.
From (\ref{eq:cov}), we obtain the covariance of $(\hbjplus , \hbjminus )$  as
\begin{align*}
  & \frac{1}{n}  \cov(\hbjplus , \hbjminus \mid \vX^j) \nonumber \\
  = & K^{-1}\sigma^2\left( 1 - \vzj - (\rhoxjxnj)^\top \sigmaxnj^{-1} \rhoxjxnj + (\rhozjxnj)^\top\sigmaxnj^{-1} \rhozjxnj \right),
\end{align*}
where %
\begin{eqnarray*}
K^{-1} &=&   (1 + \vzj)^2 - 4(\rhoxjzj)^2 -  ((\rhoxjxnj)^\top\sigmaxnj^{-1} \rhoxjxnj + (\rhozjxnj)^\top\sigmaxnj^{-1} \rhozjxnj)^2   \\
& & - ( 1- \vzj - (\rhoxjxnj)^\top\sigmaxnj^{-1} \rhoxjxnj + (\rhozjxnj)^\top\sigmaxnj^{-1} \rhozjxnj)^2. 
\end{eqnarray*}
To guarantee that $Cov(\hat{\beta}_j^+,\hat{\beta}_j^-|\vX^j)=0$, we choose $c_j$ such that 
\[
\vzj - (\rhozjxnj)^\top\sigmaxnj^{-1}\rhozjxnj = 1 - (\rhoxjxnj)^\top \sigmaxnj^{-1} \rhoxjxnj.
\]
This choice is  equivalent to choosing $c_j$ such that
\begin{equation*}
    c_j^2 \vz_j^\top(\vI_n - \vX_{-j}(\vX_{-j}^\top \vX_{-j})^{-1} \vX_{-j}^\top) \vz_j =  \vx_j^\top(\vI_n - \vX_{-j}(X_{-j}^\top \vX_{-j})^{-1} \vX_{-j}^\top) \vx_j,
\end{equation*}
\qed

\vspace{0.1in}

{\noindent\bf Proof of Lemma \ref{lem:constraints}. }
Let $A_0$ represent the ``inactive'' constraint defined in (\ref{eq:constraint}) and $\veta = (\veta_1, \veta_2)$. We have
\begin{equation*}
    A_0(\cS, \bs)\veta = A_0(\cS,\bs)(\veta_1, \veta_2) = A_0(\cS, \bs)(\veta_1+\veta_2, \veta_1-\veta_2) \begin{pmatrix}
    1/2 & 1/2\\
    1/2 & -1/2
    \end{pmatrix}.
\end{equation*}
By the definition of $\tilde{\vz}_j = (I - P_{\cS})\vz_j$, where $c_j$ is defined in (\ref{eq:lassocj}), we have $A_0(\cS, \bs)(\veta_1+\veta_2) = 0$. Then, we have
\begin{equation*}
    A_0(\cS,\bs) \veta = A_0(\cS, \bs)(\frac{\veta_1-\veta_2}{2})(1,-1).
\end{equation*}
Setting $\ba_0 = A_0(\cS, \bs)(\frac{\veta_1-\veta_2}{2})$, we have $A_0(\cS,\bs)\veta = \ba_0(1,-1)$. Similarly, we have $A_1(\cS, \bs)(\veta_1-\veta_2) = 0$. Letting $\ba_1 = A_1(\cS, \bs)(\frac{\veta_1+\veta_2}{2})$, we have $A_1(\cS,\bs)\veta = \ba_1(1,-1)$.
\qed

\vspace{0.2in}

{\noindent\bf Proof of Lemma \ref{le:exact:constraint}.} We first derive the  constraints $\{A_0(\cS,\bs)\by \le \vb_0(\cS,\bs) \}$ and $\{A_1(\cS,\bs) \by \le \vb_1(\cS,\bs)\}$ in (\ref{eq:constraint}) separately. Plugging in $\by= P_{\veta} \by + (I- P_{\veta})\by$ and
$\br = (I - P_{\veta})\by$, we have
\begin{align*}
 A_0(\cS, \bs) \veta (\veta^\top\veta)^{-1}\veta^\top\by + A_0(\cS, \bs) \br &= \ba_0(1,-1)\diag(\alpha,\alpha) (\hbjplus, \hbjminus)^\top  + A_0(\cS, \bs) \br \\
 & = \alpha \ba_0(\hbjplus- \hbjminus)  + A_0(\cS, \bs) \br < \vb_0(\cS, \bs).
\end{align*} 
Then, $\{A_0(\cS,\bs)\by \le \vb_0(\cS,\bs) \}$ can be written as $\{\cV_0^{L}(\br) \le \hbjplus-\hbjminus \le \cV_0^{U}(\br),\, \cV_0^{N}(\br)>0 \}$, where 
\begin{equation*}
\begin{aligned}
    \cV_0^{L}(\br)  := \max_{j:a_{0j} <0}
    \frac{b_{0j} - (A_0\br)_j}{\alpha a_{0j}}, \, \cV_0^{U}(\br)  := \min_{j:a_{0j} >0}\, \frac{ b_{0j} - (A_0\br)_{j}}{\alpha a_{0j}}, \,
    \cV_0^{N}(\br)  := \min_{j: a_{0j} = 0}
     \vb_{0j} -  (A_0\br)_{j}.
\end{aligned}
\end{equation*}
Similarly, we have  
\begin{align*}
 A_1(\cS, \bs) \veta (\veta^\top\veta)^{-1}\veta^\top\by + A_1(\cS, \bs) \br &= \ba_1(1,1)\diag(\alpha,\alpha) (\hbjplus, \hbjminus)^\top  + A_1(\cS, \bs) \br \\
 & = \alpha \ba_1(\hbjplus + \hbjminus)  + A_1(\cS, \bs) \br < \vb_1(\cS, \bs).
\end{align*} 
Then, $\{A_1(\cS,\bs) \by \le \vb_1(\cS,\bs)\}$ can be written as $\{\cV_1^{L}(\br) \le \hbjplus+\hbjminus \le \cV_1^{U}(\br),\,\cV_1^{N}(\br)>0 \}$, where 
\begin{equation*}
\begin{aligned}
    \cV_1^{L}(\br)  := \max_{j: a_{1j} <0}
    \frac{b_{1j} -  (A_1\br)_{j}}{\alpha a_{1j}}, \,
    \cV_1^{U}(\br)  := \min_{j: a_{1j} >0}
    \frac{ b_{1j} - (A_1\br)_{j}}{ \alpha a_{1j}}, \,
    \cV_1^{N}(\br)  := \min_{j: a_{1j} = 0}
     b_{1j} - (A_1\br)_{j}.
\end{aligned}
\end{equation*}
Thus, we have 
\begin{align*}
   & \{A\by \leq \vb\} = \{A_0(\cS,\bs)\by \le \vb_0(\cS,\bs) \} \cap \{A_1(\cS,\bs) \by \le \vb_1(\cS,\bs)\} \\
    =&   \{\cV_0^{L}(\br) \le \hbjplus-\hbjminus \le \cV_0^{U}(\br),\, \cV_0^{N}(\br)>0 \}   \cap \{\cV_1^{L}(\br) \le \hbjplus+\hbjminus \le \cV_1^{U}(\br),\,\cV_1^{N}(\br)>0 \}.
\end{align*}
\qed

\vspace{0.1in}

{\noindent\bf Proof of Theorem \ref{thm:trun:dist}.} 
By Lemma \ref{le:exact:constraint}, we know that $\{A\by\le \vb\}$ imposes a rectangle region in the subspace of $(\hbjplus,\hbjminus)$ with the boundaries being parallel to $(\hbjplus+\hbjminus, \hbjplus-\hbjminus)$.  The truncation points of the region are only related to $\br$, which are independent of $(\hbjplus+\hbjminus, \hbjplus-\hbjminus)$. Thus, through a transformation using the cumulative distribution function of a truncated normal distribution, we have 
\begin{align*}
F_{0,\sigma^2}^{[\cV_1^{L}(\br), \cV_1^{U}(\br)]} (\hbjplus+\hbjminus )\,\mid\, \hat{\cS} = \cS, \hat{\bs} = \bs \sim  Unif(0,1), \\
F_{0,\sigma^2}^{[\cV_0^{L}(\br), \cV_0^{U}(\br)]} (\hbjplus - \hbjminus )\,\mid\, \hat{\cS} = \cS, \hat{\bs} = \bs \sim  Unif(0,1),
\end{align*}
where $F_{\mu, \sigma}^{[a, b]}(x)$ is the cumulative distribution function of the truncated normal distribution defined in (\ref{eq:truncatednormal}).\qed

\vspace{0.2in}

{\noindent\bf Proof of Lemma \ref{lemma:suff:ols}.}
Let $\eta_j = \hbjplus + \hbjminus$ and $\theta_j = \hbjplus-\hbjminus$, and we have
\begin{align}
&\sum_{j,k\in\cS_0}\cov\big( 1(M_j\ge t),1(M_k\ge t) \big)\nonumber\\
 =& \sum_{j,k\in \cS_0}E[cov(1(|\eta_j| - |\theta_j|>t), 1(|\eta_k|-|\theta_k|>t)|\theta_j, \theta_k)] \nonumber\\
 &+\sum_{j,k\in\cS_0} cov(P(|\eta_j| - |\theta_j|>t|\mid\theta_j,\theta_k) P(|\eta_k|-|\theta_k|>t\mid\theta_j, \theta_k)) \nonumber\\
 &= I_1 + I_2.
 \label{eq:covI}
\end{align}
We define the maximum correlation coefficient between $\eta_j$ and $\eta_k$ as
\begin{equation*}
\rho_{\max}(\eta_j, \eta_k ) = \sup_{f_1\in \mathcal{F}_j, f_2\in \mathcal{F}_k}E[f_1(\eta_j)f_2(\eta_k)],
\end{equation*}
where $\mathcal{F}_j = \{f_1: E[f_1(\eta_j)]=0, \var(f_1(\eta_j))=1\}$ and  $\mathcal{F}_k = \{f_1: E[f_1(\eta_k)]=0, \var(f_1(\eta_k))=1\}$. For given $\theta_j$ and $\theta_k$, we have $\var[1(|\eta_j| - |\theta_j|>t)]\leq 1$ and $\var[1(|\eta_k| - |\theta_k|>t)]\leq 1$. Then, we can bound each term in $I_1$ as
\begin{align*}
&cov(1(|\eta_j| - |\theta_j|>t), 1(|\eta_k|-|\theta_k|>t)\mid\theta_j, \theta_k) \leq cor(1(|\eta_j| - |\theta_j|>t), 1(|\eta_k|-|\theta_k|>t)\mid \theta_j, \theta_k).
\end{align*}
Since the maximal correlation between a pair of bivariate normal random variables is just the absolute value of their ordinary correlation and the definition of $\tilde{z}_j$ in (\ref{eq:lassocj}) makes $\theta_j, \theta_k$ independent of $\eta_j,\eta_k$, we have 
\begin{equation}\label{eq:I1}
I_1 \leq \sum_{j,k \in \mathcal{S}_0}\rho_{\max}(\eta_j, \eta_k ) \leq \sum_{j,k\in \cS_0}|\rho(\eta_j, \eta_k)| \leq \sum_{j,k\in \cS_0} \frac{|\Omega_{jk}|} {\sqrt{\Omega_{jj}}\sqrt{\Omega_{kk}}}.
\end{equation}

Let $\nu = \lfloor \frac{p}{n-p} \rfloor$, we construct the block design $[\tilde{z}_1 = \dots = \tilde{z}_{\nu}]$, $[\tilde{z}_{\nu_1}=\dots=\tilde{z}_{2\nu}]$, \dots, $[\tilde{z}_{(n-p-1)\nu}=\dots=\tilde{z}_p]$ and the vectors in different blocks are orthogonal to each other. If $\tilde{z}_j, \tilde{z}_k$ are from different blocks, we have  $cov(P(|\eta_j| - |\theta_j|>t|\mid\theta_j,\theta_k) P(|\eta_k|-|\theta_k|>t\mid\theta_j, \theta_k))=0$. Then, 
\begin{equation}\label{eq:I2}
I_2 = O(\frac{|\cS_0|^2 p}{n-p}).
\end{equation}
Combining (\ref{eq:I1}) and (\ref{eq:I2}), we find that there exists a $\alpha_1\in(0,2)$ such that
\begin{equation*}
\sum_{j,k}\cov\big(1(M_j\ge t), 1(M_k\ge t) \big) \le C_1 |\mathcal{S}_0|^{\alpha_1}, \; \forall \;  t,\nonumber
\end{equation*}
which concludes the proof.
\qed

\vspace{0.2in}

{\noindent\bf Proof of lemma \ref{lemma:suff:lasso}.}
Let $\Omega^0 = \left(X_{[\widehat{\cS}\cap S_0]}^T X_{[\widehat{\cS}\cap S_0]}\right)^{-1}$.  We define
\begin{equation*}
R(\widehat{\cS}\cap S_0) = \sum_{j,k\in \widehat{\cS}\cap S_0} \frac{|\Omega^0_{jk}|} {\sqrt{\Omega^0_{jj}}\sqrt{\Omega^0_{kk}}}.
\end{equation*}
First, we will show that $R(\widehat{\cS}\cap S_0)  \le O_p((\widehat{\cS}\cap S_0)^{3/2})$.  We denote the m-sparse maximum and minimum eigenvalues of the covariance matrix $\Sigma$ as
\begin{equation*}
\lambda_{\min}(m) = \min_{\beta:||\beta||_{0}\leq m}\frac{\beta^\intercal \Sigma\beta}{\beta^\intercal\beta}\ \  \textnormal{and}\ \  \lambda_{\max}(m) = \max_{\beta:||\beta||_{0}\leq m}\frac{\beta^\intercal \Sigma\beta}{\beta^\intercal\beta},
\end{equation*}
and the  m-sparse maximum and minimum eigenvalues of the sample covariance matrix $\hat{\Sigma}$ as
\begin{equation*}
\hat{\lambda}_{\min}(m) = \min_{\beta:||\beta||_{0}\leq m}\frac{\beta^\intercal \hat{\Sigma}\beta}{\beta^\intercal\beta}\ \  \textnormal{and}\ \  \hat{\lambda}_{\max}(m) = \max_{\beta:||\beta||_{0}\leq m}\frac{\beta^\intercal \hat{\Sigma}\beta}{\beta^\intercal\beta}.
\end{equation*}
By Theorem  6.5 in \cite{wainwright2019high} and the Cauchy interlacing Theorem, we have 
\begin{align*}
    &\widehat{\lambda}_{\min}(p_1\log n)\geq \lambda_{\min}(p_1\log n)-o_p(1) \geq \lambda_{\min}(p)-o_p(1)\geq 1/C-o_p(1),\\
&\widehat{\lambda}_{\max}(p_1+\min\{n, p\})\leq C^\prime\lambda_{\max}(p_1+n)+o_p(1) \leq C^\prime\lambda_{\max}(p)+o_p(1)\leq C^\prime C+o_p(1).
\end{align*}
Thus, the conditions in Lemma 5 in \citet{meinshausen2009lasso} hold with probability approaching 1, which further implies that $|\hat{\cS}| = O(|\cS_1|)$.  We have
\begin{equation*}
1/C - o_p(1) \leq  \lambda_{\min}(\Omega^0) \leq \lambda_{\max}(\Omega^0)\leq  C + o_p(1).
\end{equation*}
Then we have 
\begin{equation}\label{eq:Requal}
R(\widehat{\cS}\cap S_0) \leq \frac{||\Omega^0||_1}{\lambda_{\min}(\Omega^0)} \leq \frac{\sqrt{|\widehat{\cS}\cap S_0|}||\Omega^0||_2}{\lambda_{\min}(\Omega^0)} = O_p((\widehat{\cS}\cap S_0)^{3/2}),
\end{equation}
where the first inequality follows from the fact that for any positive definite matrix $A\in \mathbbm{R}^{m\times m}$, $\lambda_{\min}(A)\leq A_{ii}\leq \lambda_{\max}(A)$ for $i\in[m]$. The second inequality follows from the Cauchy-Schwartz inequality. The third equality is based on the following fact:
\begin{equation*}
   \sum_{i,j}A_{ij}^2 = \text{tr}(A^\intercal A) = \sum_{i = 1}^m\lambda_{i}^2(A).
\end{equation*}

Second, we characterize the sum of all pairwise correlations of the $M_j$'s. Let $\eta_j = \hbjplus+\hbjminus$, $\theta_j=\hbjplus-\hbjminus$, where $\hbjplus$ and $\hbjminus$ are defined in (\ref{eq:post-beta}).
We define $\tilde{\eta}_j = F^{-1}_{0,\sigma^2}F_{0,\sigma^2}^{[\cV_1^{-}(\br), \cV_1^{+}(\br)]} ( \eta_j )$ and $\tilde{\theta}_j = F^{-1}_{0,\sigma^2}F_{0,\sigma^2}^{[\cV_0^{-}(\br), \cV_0^{+}(\br)]} ( \theta_j  )$. Then, we  have
\begin{align*}
&\sum_{j,k\in\widehat{\cS}\cap S_0}\cov\big( 1(M_j\ge t),1(M_k\ge t) \big)\nonumber\\
 =& \sum_{j,k\in\widehat{\cS}\cap S_0}E[cov(1(|\tilde{\eta}_j| - |\tilde{\theta}_j|>t), 1(|\tilde{\eta}_k|-|\tilde{\theta}_k|>t)\mid \theta_j, \theta_k)] \nonumber\\
 &+\sum_{j,k\in\widehat{\cS}\cap S_0} cov(P(|\tilde{\eta}_j| - |\tilde{\theta}_j|>t|\mid \theta_j,\theta_k) P(|\tilde{\eta}_k|-|\tilde{\theta}_k|>t\mid\theta_j, \theta_k)) \nonumber\\
 =& I_1 + I_2.
\end{align*}
Here $\tilde{\eta}_j$ and $\tilde{\eta}_k$ are truncated Gaussian with truncation intervals defined in (\ref{eq:v}).   We define the maximum correlation coefficient between $\tilde{\eta}_j$ and $\tilde{\eta}_k$ as
\begin{equation*}
\rho_{\max}(\eta_j, \eta_k ) = \sup_{f_1\in \mathcal{F}_j, f_2\in \mathcal{F}_k}E[f_1(\tilde{\eta}_j)f_2(\tilde{\eta}_k)]
\end{equation*}
where $\mathcal{F}_j = \{f_1: E[f_1(\tilde{\eta}_j)]=0, \var(f_1(\tilde{\eta}_j))=1\}$ and  $\mathcal{F}_k = \{f_1: E[f_1(\tilde{\eta}_k)]=0, \var(f_1(\tilde{\eta}_k))=1\}$. For given $\theta_j$ and $\theta_k$, we have $\var[1(|\tilde{\eta}_j| - |\tilde{\theta}_j|>t)]\leq 1$ and $\var[1(|\tilde{\eta}_k| - |\tilde{\theta}_k|>t)]\leq 1$. Then, similar to the proof of Lemma~\ref{lemma:suff:ols}, we can bound each term in $I_1$ by the corresponding maximum correlation so that
\begin{equation*}
I_1 \leq \sum_{j,k \in \hat{\cS}\setminus\cS_1}\rho_{\max}(\eta_j, \eta_k ) \leq \sum_{j,k\in \hat{\cS}\setminus\cS_1}|\rho(\eta_j, \eta_k)|.
\end{equation*}
Similar to the definition of $\tilde{z}_j$ in (\ref{eq:lassocj}), we choose $z_j$ in the orthogonal space of $X_{\hat{\mathcal{S}}}$ which makes $\theta_j, \theta_k$ independent of $\eta_j,\eta_k$. By condition (28), we have
\begin{equation}\label{eq:I1-2}
I_1 \leq \sum_{j,k\in \hat{\cS}\cap\cS_0} \frac{|\Omega_{jk}|} {\sqrt{\Omega_{jj}}\sqrt{\Omega_{kk}}}
\end{equation}
Since we have $|\widehat{\cS}|=O(|\cS_1|)=o(n)$, we are able to construct $\{\tilde{z}_j\}_{j\in\widehat{\cS}}$ such that they are  orthogonal to each other.
Since $R_{\theta_j,\theta_k}=0$ if $\tilde{z}_j$ and  $\tilde{z}_k$ are orthogonal, we have
\begin{equation}\label{eq:I2-2}
I_2 = 0.
\end{equation}
Combining (\ref{eq:I1-2}), (\ref{eq:I2-2}) , we know that there exists $\alpha_2\in(0,2)$ such that 
\begin{equation*}
\sum_{j,k\in \widehat{\cS}\cap \cS_0}\cov\big( 1(M_j\ge t), 1(M_k\ge t) \big) \le C_1 |\widehat{\cS}\cap \mathcal{S}_0|^{\alpha_2}, \; \forall \;  t.\nonumber
\end{equation*}
This concludes the proof.
\qed

\vspace{0.2in}

{\noindent\bf Proof of Lemma \ref{lem:slln}.} 
For any $\epsilon \in (0, 1)$ and $N_{\epsilon} =\lceil2/\epsilon\rceil$, let $-\infty = a^{p}_0 < a^{p}_1 < \cdots < a^{p}_{N_\epsilon} = \infty$ such that $G_0(a^{p}_{k - 1}) - G_0(a^{p}_k) \leq \epsilon /2$ for $k\in[N_\epsilon]$. Such a sequence $\{a_k^p\}_{k=0}^{N_\epsilon}$ exists since $G_0(t)$ is a continuous function with respect to $t\in \mathbbm{R}$.
We have
\begin{equation}
\label{eq:lemma-2-1}
\begin{aligned}
\mathbbm{P}\left(\sup_{t\in\mathbbm{R}}V(t) - G_0(t) > \epsilon\right) & \leq \mathbbm{P}\left(\bigcup_{k = 1}^{N_\epsilon}\sup_{t \in\left[a^{p}_{k - 1},a^{p}_k\right)}V(t) - G_0(t) > \epsilon\right)\\
& \leq \sum_{k = 1}^{N_{\epsilon}}\mathbbm{P}\left(\sup_{t \in\left[a^{p}_{k - 1}, a^{p}_k\right)}V(t) - G_0(t) > \epsilon\right).
\end{aligned}
\end{equation}
Since  $V(t)$ and $G_0(t)$ are both monotonically decreasing functions, for any $k \in [N_\epsilon]$, we have
\begin{equation*}
\sup_{t \in\left[a^{p}_{k - 1}, a^{p}_k\right)}V(t) - G_0(t) \leq V(a^{p}_{k - 1}) - G_0(a^{p}_{k}) \leq V(a^{p}_{k - 1}) - G_0(a^{p}_{k - 1}) + \epsilon/2.
\end{equation*}

For the low-dimensional case, based on Equation \eqref{eq:lemma-2-1}, Assumption \ref{assump}(a), Lemma 6, and the Chebyshev inequality, it follows that
\begin{equation*}
\mathbbm{P}\left(\sup_{t\in\mathbbm{R}}V_{p}(t) - G_0(t) > \epsilon\right) \leq \sum_{k = 1}^{N_{\epsilon}}\mathbbm{P}\left(V(a^{p}_{k - 1}) - G_0(a^{p}_{k - 1}) > \frac{\epsilon}{2}\right)\leq \frac{4CN_\epsilon}{|\cS_0|^{2-\alpha_1}\epsilon^2} \to 0,\ \ \textnormal{as}\ \ p\to\infty.
\end{equation*}
Similarly, we show that
\begin{equation*}
\mathbbm{P}\left(\inf_{t\in\mathbbm{R}}V(t) - G_0(t) < -\epsilon\right) \leq \frac{4CN_\epsilon}{|\cS_0|^{2-\alpha_1}\epsilon^2} \to 0,\ \ \ \ \textnormal{as}\ \ p\to\infty.
\end{equation*}

For the high-dimensional case, based on Equation \eqref{eq:lemma-2-1}, Assumption \ref{assump}(b), Lemma 7, and the Chebyshev inequality, it follows that
\begin{equation*}
\mathbbm{P}\left(\sup_{t\in\mathbbm{R}}V_{p}(t) - G_0(t) > \epsilon\right) \leq \sum_{k = 1}^{N_{\epsilon}}\mathbbm{P}\left(V(a^{p}_{k - 1}) - G_0(a^{p}_{k - 1}) > \frac{\epsilon}{2}\right)\leq \frac{4CN_\epsilon}{|\widehat{\cS}\cap\cS_0|^{2-\alpha_2}\epsilon^2} \to 0,\ \ \textnormal{as}\ \ p\to\infty.
\end{equation*}
Similarly, we show that
\begin{equation*}
\mathbbm{P}\left(\inf_{t\in\mathbbm{R}}V(t) - G_0(t) < -\epsilon\right) \leq \frac{4CN_\epsilon}{|\widehat{\cS}\cap\cS_0|^{2-\alpha_2}\epsilon^2} \to 0,\ \ \ \ \textnormal{as}\ \ p\to\infty.
\end{equation*}

This concludes the proof that $\sup_{t\in\mathbbm{R}} \left|V(t) - G_0(t)\right| \to 0$ in probability. The convergence of $\sup_{t\in\mathbbm{R}} \left|V(t) - G_0(t)\right|$ can be shown similarly using the symmetric assumption of the mirror statistics $M_j$ for $j \in S_0$.
Similarly, we have 
\begin{equation*}
 \sup_t\left|{V'(t)}- G_0(t)\right|\xrightarrow[]{p}0, \quad \sup_t\left|{V^1(t)}- G_1(t)\right|\xrightarrow[]{p}0.
\end{equation*}
\qed

{\noindent\bf Proof of Theorem \ref{thm:fdr}.}
\newcommand{\FDP}{\text{FDP}}
We denote
\begin{equation*}
\FDP^\prime = \frac{V^\prime (t)}{(V^{\prime}(t) + r_p V^1(t))\vee 1/p }.
\end{equation*}
Suppose Assumption 2 holds. By Lemma \ref{lem:slln} and the definition of $\FDP^{\infty}$, we have $\forall \ t>0$,
\begin{equation*}
|\FDP^\prime(t) - \FDP^{\infty}(t)| \xrightarrow[]{p} 0.
\end{equation*}
For any $\epsilon>0$, we choose $t_{q-\epsilon}$ such that $FDP^{\infty}(t_{q-\epsilon})< q-\epsilon$. Then, $\exists  P_0$ such that  $\forall \ p>P_0$,
\[
P(|\FDP^\prime(t_{q-\epsilon})-\FDP^{\infty}(t_{q-\epsilon})|)\geq 1-\epsilon.
\]
Then we have 
\begin{align*}
    P(\tau_q \leq t_{q-\epsilon})\geq P(\FDP^\prime(t_{q-\epsilon})\leq q)
    \geq P(|\FDP^\prime(t_{q-\epsilon})-\FDP^{\infty}(t_{q-\epsilon})|<\epsilon) 
    \geq 1-\epsilon.
\end{align*}
We denote  
\[
\FDPbar(t)= \frac{G_0(t)}{G_0(t) + r_{p} G_1(t)}.
\]
Based on Lemma \ref{lem:slln}, for any fixed $t$, we have 
\begin{equation*}
|\FDP^\prime(t) - \FDPbar(t)| \xrightarrow{p} 0, \quad |\FDP(t) - \FDPbar(t)| \xrightarrow{p} 0.
\end{equation*}
Then we have 
\begin{equation}\label{eq:fdpuniform}
\sup_{0<t\leq t_{q-\epsilon}}|\FDP (t)-\FDPbar(t)| \xrightarrow{p} 0, \quad \sup_{0<t\leq t_{q-\epsilon}}|\FDP^\prime (t)-\FDPbar(t)|\xrightarrow{p} 0.
\end{equation}
Now, we are ready to derive the upper bound for  $\limsup_{p\to\infty} \mathbbm{E}\left[\text{FDP}\left(\tau_q\right)\right]$.
\begin{equation*}
\begin{aligned}
 \limsup_{p\to\infty} \mathbbm{E}\left[\text{FDP}\left(\tau_q\right)\right] = & \limsup_{p\to\infty}\mathbbm{E}\left[\FDP(\tau_q) \mid \tau_q\leq t_{q-\epsilon}\right]P(\tau_q\leq t_{q-\epsilon}) + \epsilon \\
\leq & \limsup_{p\to\infty}\mathbbm{E}\left[|\FDP(\tau_q)-\FDPbar(\tau_q)| \mid \tau_q\leq t_{q-\epsilon}\right]P(\tau_q\leq t_{q-\epsilon})\\
& + \limsup_{p\to\infty}\mathbbm{E}\left[|\FDP^\prime (\tau_q)-\FDPbar(\tau_q)| \mid \tau_q\leq t_{q-\epsilon}\right]P(\tau_q\leq t_{q-\epsilon}) \\
& + \limsup_{p\to\infty}\mathbbm{E}\left[\FDP^\prime (\tau_q) \mid \tau_q\leq t_{q-\epsilon}\right]P(\tau_q\leq t_{q-\epsilon}) + \epsilon\\
\leq & \limsup_{p\to\infty}\mathbbm{E}\left[\sup_{0<t\leq t_{q-\epsilon}}|\FDP (t)-\FDPbar(t)| \right] \\
& +\limsup_{p\to\infty}\mathbbm{E}\left[\sup_{0<t\leq t_{q-\epsilon}}|\FDP^\prime (t)-\FDPbar(t)| \right]\\
& + \limsup_{p\to\infty}\mathbbm{E}\left[FDP^\prime(\tau_q) \right] + \epsilon. \\
\end{aligned}
\end{equation*}
The first two terms are 0 based on \eqref{eq:fdpuniform} and  Lebesgue's dominated convergence theorem (Corollary 6.26 in \citet{klenke2013probability}). For the last term, we have
$\text{FDP}^\dagger_p\left(\tau_q\right)\leq q$ almost surely
based on the definition of $\tau_q$. This concludes the proof of Theorem \ref{thm:fdr}.

\qed

\vspace{0.1in}

{\noindent\bf Proof of Theorem \ref{thm:fd}.}
Suppose Assumption 1 holds. Using the $l_1$-convergence results in \cite{van2014asymptotically}, we have 
\begin{equation*}
    P(M_j >0 ) = P(sign(\hbjplus)= sign(\hbjminus) ) > 1- \frac{2}{p},
\end{equation*}
for $j\in S_1$.
By (\ref{thm:trun:dist}), we have
\begin{equation*}
    \widehat{FD}(k) = \#\{M_j \leq -M_{(k)} \} = \#\{j: j\in \mathcal{S}_0, M_j< -M_{(k)}\} = V(M_{(k)}),
\end{equation*}
where $V(\cdot)$ is defined in Lemma \ref{lem:slln}. Also, $FD(k)$ can be written as 
\begin{equation*}
    FD(k) = \#\{j\in S_0: M_j \geq t\}.
\end{equation*}
The expection of $FD(k)$ is 
\begin{equation*}
    \frac{1}{k}\EE[FD(k)] =\sum_{j\in\mathcal{S}_0} \mathbb{E} 1(M_j\geq M_{(k)})   = \sum_{j\in\mathcal{S}_0} \mathbb{E} 1(M_j\leq -M_{(k)}),
\end{equation*}
where the second equality holds by the symmetric property of $M_j$ for $j\in S_0$. By  Lemma \ref{lem:slln}, conditioned on $\{\mathcal{S}_1 \subset \widehat{\mathcal{S}}\}$,  we have 
\begin{equation*}
      \frac{ V(t) - \sum_{j=1}^p \mathbb{E}1(j\in\mathcal{S}_0, M_j\le -t)}{k} \overset{p}{\to} 0.
\end{equation*}
Since the event $\{\mathcal{S}_1 \subset \widehat{\mathcal{S}}\}$ holding with probability $1-\frac{2}{p}$,  we have 
\begin{equation*}
    \frac{ V(t) - \sum_{j=1}^p \mathbb{E}1(j\in\mathcal{S}_0, M_j\le -t)}{k}   \overset{p}{\to} 0.
\end{equation*}
\qed

{\noindent\bf Proof of Theorem \ref{thm:fdmean}.}
By the same argument as in the proof of Theorem 6, we have that
\begin{equation*}
    \widehat{FD}(k) = \#\{M_j \leq -M_{(k)} \} = \#\{j: j\in \mathcal{S}_0, M_j< -M_{(k)}\},
\end{equation*}
holds with probability at least $1-\frac{2}{p}$. Thus we have that
\begin{equation*}
    \mathbb{E}[\widehat{FD}(k)] = \mathbb{E}1(j\in\mathcal{S}_0, M_j\le -M_{(k)})) = \mathbb{E}1(j\in\mathcal{S}_0, M_j\ge M_{(k)})) = \mathbb{E}[FD(k)],
\end{equation*}
holds with probability at least $1-\frac{2}{p}$. 
\qed

\subsection{Gaussian mirrors for neural networks}
To examine its generalizability, we apply the GM method, termed as NGM,  to select features in neural networks. 
Similar to the GM construction for linear models, for each feature $X_j$ of the input of a neural network model, we create its mirror variables $X_j^+$ and $X_j^-$ the same way as in (\ref{eq:gm_design}), by choosing constant $c_j$ to minimize a nonparametric measure of conditional mutual information (equivalent to partial correlation in Gaussian models). The pair of variables are then connected to all the internal nodes of the network in the same way as $X_j$. The mirror statistic is constructed using influence measures (such as the multiplication of the all weights connecting from the output node to each mirror feature).
To save computing time, we can also create such mirrors for all or a subset of features simultaneously, so as to obtain multiple mirror statistics in one run.
More details can be found  in \citet{xing2020neural}.

We compare the performances  of NGM
with DeepPINK, which is a generalization of the Model-X knockoff for neural network models,
for variable selection in non-linear models: $y_i=f_{\ell}(\bbeta^\top \bx_i)+\epsilon_i,  \ \epsilon_i\stackrel{\mbox{\small{i.i.d.}}}{\sim} N(0,1)$,
for $i=1,\dots, n$. We choose three nonlinear link functions: $f_1(t)=t+sin(t)$, $f_2(t)=0.5t^3$ and $f_3(t)=0.1t^5$.
We randomly set $k=30$ elements in $\bbeta$ to be nonzero and generated from $\mathcal{N}(0,(20\sqrt{log(p)/n})^2)$ to mimic various signal strengths in real applications. 
The FDR level is set at $q=0.1$.
Both NGM and DeepPINK are applied to the same three-layer  perceptron model, which has  two hidden layers with $N_1=\lceil 20\log(p)\rceil$ nodes for the first layer and $N_2= \lceil 10 \log(p)\rceil$ nodes for the second layer.

\begin{table*}[h]
	\centering
		\caption{Neural network fitting for single-index relationships}
		\label{tab:si}
		\begin{tabular}{cccccc}
			\toprule
			
			\multirow{2}{*}{$n=1000$} & \multirow{2}{*}{Link function} & \multicolumn{2}{c}{NGM}& \multicolumn{2}{c}{DeepPink}\\
			\cmidrule(lr){3-4} \cmidrule(lr){5-6} 
			& &FDR &Power &FDR &Power\\
			\midrule
			\multirow{3}{*}{$p=500$} & $f_1(t)=t+sin(t)$  &0.045 &0.857  &0.085 &0.900 \\
			& $f_2(t)=0.5t^3$  &0.080 &0.843  &0.155 &0.413 \\
			& $f_3(t)=0.1t^5$  &0.071 &0.847  &0.148 &0.333 \\
			\cmidrule{1-6}
			\multirow{3}{*}{$p=2000$} & $f_1(t)=t+sin(t)$  &0.081 &0.820  &0.126 &0.817 \\
			& $f_2(t)=0.5t^3$  &0.082 &0.587 &0.149 &0.320 \\
			& $f_3(t)=0.1t^5$  &0.108 &0.530  &0.477 &0.053 \\
			\bottomrule
		\end{tabular}
\end{table*}

As shown in Table \ref{tab:si}, NGMs and DeepPINK all performed well for the case with $f_1=t+sin(t)$. However, in the latter two cases where $f$ is a polynomial,  DeepPINK appears to have a significant lower power than $NGM$. In all the three  cases, NGMs are capable of maintaining the FDR below the designated level.

\bibliographystyle{chicago} %
\bibliography{ref}
\end{document}